\documentclass[%
reprint,
amsmath,amssymb,
aps,
pra,
superscriptaddress
]{revtex4-2}

\usepackage{comment}
\usepackage{graphicx}
\usepackage{dcolumn}
\usepackage{bm}
\usepackage{braket}
\usepackage{amsmath}
\usepackage{bbm}
\usepackage[dvipsnames]{xcolor}
\usepackage{mathrsfs}
\usepackage[colorlinks=true,linkcolor=black,citecolor=black,urlcolor=NavyBlue]{hyperref}
\usepackage{pifont}
\usepackage{hyperref}
\usepackage{tikz}
\usetikzlibrary{quantikz2}

\definecolor{GPOcolor}{HTML}{B3C3E5}
\definecolor{NSOcolor}{HTML}{F8BEBD}
\definecolor{statecolor}{HTML}{F8BEBD}
\definecolor{channelcolor}{HTML}{F3D7CC}
\definecolor{dblue}{rgb}{.2,0.2,.8}

\newcommand{\nocontentsline}[3]{}
\newcommand{\tocless}[2]{\bgroup\let\addcontentsline=\nocontentsline#1{#2}\egroup}

\newcommand{\RT}{R_\mathrm{T}}
\DeclareMathAlphabet{\mathdutchcal}{U}{dutchcal}{m}{n}

\usepackage{amsthm}
\theoremstyle{definition}

\theoremstyle{plain}
\newtheorem{theorem}{Theorem}

\def\BraVert{\egroup\,\mid\,\bgroup}

\def\ketbra#1#2{\ket{#1\vphantom{#2}}\!\bra{#2\vphantom{#1}}}

\def\bra#1{\mathinner{\langle{#1}|}}

\def\ket#1{\mathinner{|{#1}\rangle}}

\begin{document}

\title{Thermodynamic criteria for signalling in quantum channels}

\author{Yutong Luo}
\email{yuluo@tcd.ie}
\affiliation{School of Physics, Trinity College Dublin, Dublin 2, Ireland}
\affiliation{Trinity Quantum Alliance, Unit 16, Trinity Technology and Enterprise Centre, Pearse Street, Dublin 2, Ireland}

\author{Simon Milz}
\email{s.milz@hw.ac.uk}
\affiliation{Institute of Photonics and Quantum Sciences, Heriot-Watt University, Edinburgh EH14 4AS, 
United Kingdom}

\author{Felix C. Binder}
\email{felix.binder@tcd.ie}
\affiliation{School of Physics, Trinity College Dublin, Dublin 2, Ireland}
\affiliation{Trinity Quantum Alliance, Unit 16, Trinity Technology and Enterprise Centre, Pearse Street, Dublin 2, Ireland}

\date{\today}

\begin{abstract}
Signalling quantum channels are fundamental to quantum communication, enabling the transfer of information from input to output states. In contrast, thermalisation erases information about the initial state. This raises a crucial question: How does the thermalising tendency of a quantum channel constrain its signalling power and vice versa? 
In this work, we address this question by considering three thermodynamic tasks associated with a quantum channel: the generation, preservation, and transmission of athermality. We provide faithful measures for athermality generation and athermality preservation of quantum channels, and prove that their difference quantifies athermality transmission.  Analysing these thermodynamic tasks, we find that the signalling ability of a quantum channel is upper-bounded by its athermality preservation and lower-bounded by its athermality transmission, thereby establishing a fundamental relationship between signalling and thermodynamic properties of channels for quantum communication.
We demonstrate this interplay for the example of the quantum switch, revealing an explicit trade-off between the signalling ability and athermality of the quantum channels it can implement. 
\end{abstract}

\maketitle

\section{Introduction}
Quantum communication fundamentally relies on the ability of quantum channels to transmit information. It requires channels that are \textit{signalling}, such that changes in the input state lead to distinguishable changes at the output. In practice, quantum communication devices are subject to energetic and thermodynamic constraints~\cite{jarzyna_classical_2017,wilde_energy-constrained_2018, ding_quantum_2019, Narasimhachar2019, Korzekwa2022Encoding, Biswas2022Fluctuation} and unavoidably exposed to environmental interactions that induce thermal noise~\cite{BreuerOpenQS, schlosshauer_decoherence_2007}, driving systems toward thermal equilibrium and erasing information about the initial state. This introduces a fundamental tension between two competing effects: the signalling power of a channel and its thermalising tendency. While the thermodynamic cost of \textit{implementing} a quantum channel is well understood~\cite{faist2015minimal,faist2018fundamental,faist2021thermodynamic}, the \textit{interplay} between signalling and thermodynamic resources of a given quantum channel is physically distinct and remains to be fully characterised. Understanding this trade-off is essential for assessing the limitations of quantum communication under realistic, thermal noise. For incoherent states and dynamics, it was recently shown that a quantum channel can transmit $n$ bits of classical information if and only if it can transmit $n$ units of energy~\cite{hsieh2025dynamical}.

In this work, we consider the fully quantum case and reveal the interrelation between a quantum channel's signalling and thermodynamic capabilities by analysing three essential thermodynamic tasks: the generation, preservation, and transmission of athermality. Relating these thermodynamic tasks to resource-theoretic quantities~\cite{horodecki_quantumness_2013, Brandao2015, Coecke_mathematical_2016, QRTs_review, gour_2025}, we prove that the ability of a quantum channel to signal is thermodynamically constrained from below (above) by how much athermality it can transmit (preserve). We then derive an operational interpretation of these resources via a channel simulation problem, yielding the resource cost of implementing a channel and extending the work-cost framework discussed in Refs.~\cite{faist2018fundamental,faist2021thermodynamic}. An area where thermodynamic, signalling, and causality considerations concurrently come to the fore is the quantum switch~\cite{chiribella2013quantum}, which has attracted attention in both quantum communication and thermodynamics~\cite{ebler2018enhanced,liu2022thermodynamics,Cheng2024ExperiSwitch,matheus2023Reassess,zhen2025GenSwitch}. Using the thermodynamic resourcefulness of coherence~\cite{lostaglio2015Coherence, Streltsov2017, Francica2020} to control the causal order of two quantum channels, it is known to enhance their signalling~\cite{ebler2018enhanced, goswami_increasing_2020, chiribella_indefinite_2021} and athermal properties~\cite{felce2020quantum, liu2022thermodynamics, zhu_charging_2023}, at a thermodynamic cost. Our bounds extend these results by specifying the explicit trade-off between signalling power and athermality that the induced channel of a quantum switch exhibits in terms of invested resources. Together, our results allow us to gauge the thermodynamic and signalling abilities of quantum channels in a unified way.

The remainder of the paper is organised as follows. In Sec.~\ref{sec:resource_measures}, we introduce quantitative measures for the athermality and signalling ability of quantum channels, and define a joint measure that captures both resources collectively. In Sec.~\ref{sec:three_tasks}, we identify three thermodynamic tasks associated with a quantum channel -- generation, preservation and transmission of athermality, and show that these tasks are characterised by the defined measures. In Sec.~\ref{sec:thermal_criteria}, we formulate our thermodynamic criteria for signalling: signalling is enabled by athermality transmission but constrained by athermality preservation. To give further operational meaning for our measures, Sec.~\ref{sec:GPO_dilation} considers a channel simulation scenario and demonstrate that our measures quantify the minimal thermodynamic resources required to simulate an arbitrary quantum channel. In Sec.~\ref{sec:quantum_switch}, we apply the criteria for signalling to the quantum switch, showing that it trades between thermodynamic and signalling resources via an explicit trade-off relation. Finally, in Sec~\ref{sec:conclusion}, we summarise our results and discuss possible future directions.

Throughout, we denote the set of quantum states of system $X$ as $\mathcal{S}(X)$, and the set of quantum channels (i.e., CPTP maps) between input system $X$ and output system $Y$ as $\mathcal{O}(X\rightarrow Y)$. $X'$ denotes a copy of system $X$ with the same thermal state $\gamma_X$.

\section{Resource measures}\label{sec:resource_measures}
\subsection{Robustness of athermality and signalling}
In the language of resource theories, quantum states constitute static resources, whereas quantum operations between them provide dynamical resources.
To quantify the amount of resource in a general state/operation, we adopt the robustness measure of resources, which are widely used in the resource theories of entanglement~\cite{Vidal1999Robustness, steiner_generalized_2003, lami2023computable} and coherence~\cite{piani2016robustness,diaz2018using}, and in general resource theories~\cite{liu2019one,regula2021one,Takagi2019}.
Intuitively, robustness-based measures quantify how much noise one needs to admix before the resource possessed by a quantum state/operation disappears. Let $\mathcal{D}$ be the set of quantum states/operations and $\mathcal{F}\subset\mathcal{D}$ a free set inside $\mathcal{D}$ corresponding to some type of resource. For an object $X\in\mathcal{D}$, its robustness of resource is given by
\begin{align}
    R(X) &:= \min\left\{s\Big|\frac{X+sY}{1+s}=Z\in\mathcal{F},\, Y\in\mathcal{D}\right\},
    \label{eq:general_robustness_std}\\
    &=\min\left\{s|X\le (1+s)Z,\,Z\in\mathcal{F}\right\}
    \label{eq:general_robustness}
\end{align}
where $Y\in\mathcal{D}$ is a general ``noise'' to $X$, and in the second line $X\le (1+s)Z$ means that $(1+s)Z-X$ is a positive semidefinite operator or a completely positive map, depending on whether $\mathcal{D}$ is the set of quantum states or the set of quantum channels. Geometrically, $R(X)$ can be interpreted as the ``distance'' between $X$ and the free set $\mathcal{F}$, as is shown in Fig.~\ref{fig:general_robustness}. Besides, robustness measures can also be expressed in terms of the max-relative entropy $D_{\max}(X||Z):= \ln\min\{\lambda|X\le \lambda Z\}$~\cite{datta2009min}.
\begin{figure}
    \centering
    \includegraphics[width=0.85\linewidth]{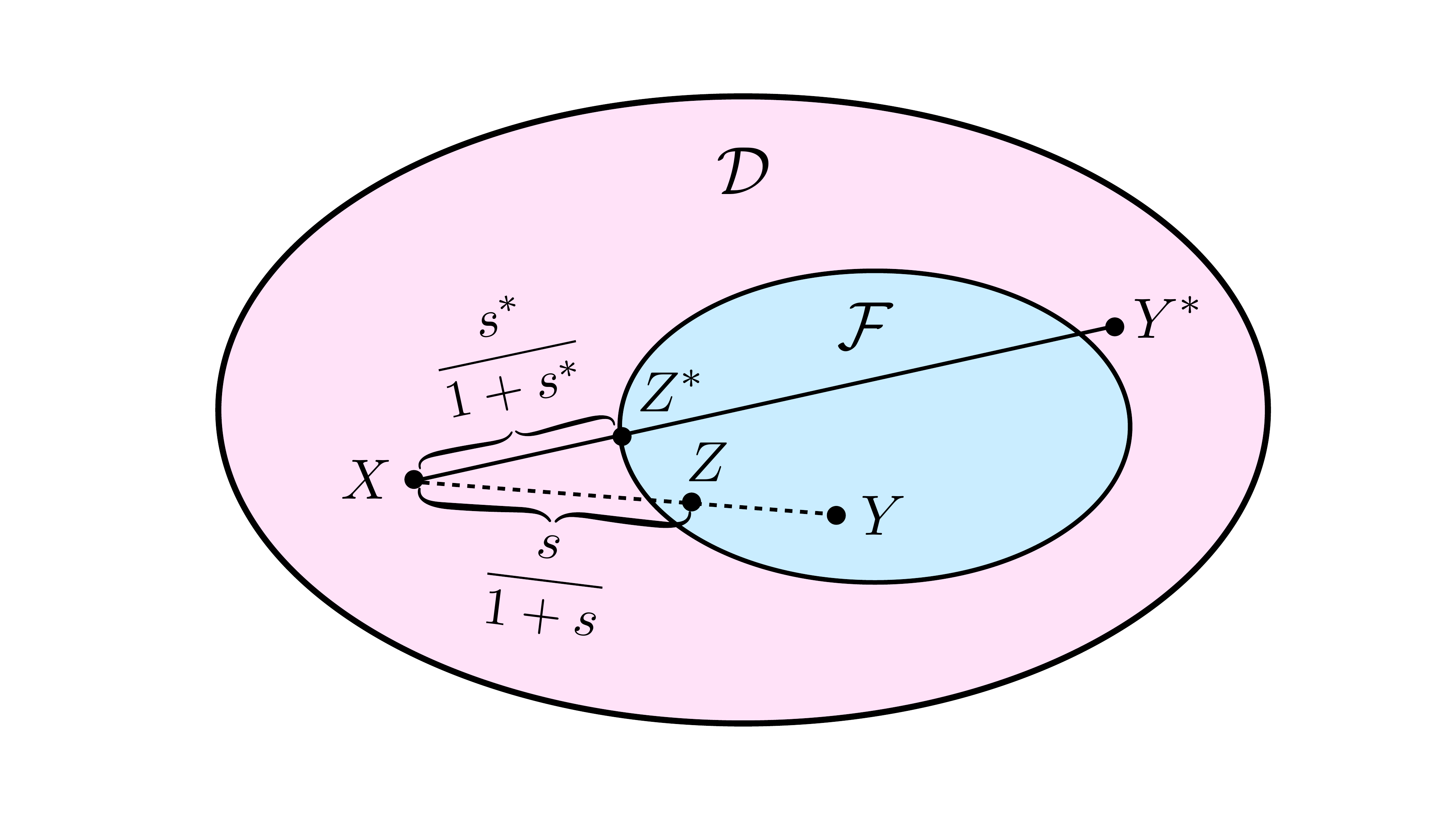}
    \caption{{\bf Geometric interpretation of robustness measure.} The distance between $X$ and the free state $Z$ corresponds to the factor $s/(1+s)$. Minimising $s$ finds the closest free state $Z^*$ to $X$, and hence, quantifies the distance between $X$ and the free set $\mathcal{F}$.}
    \label{fig:general_robustness}
\end{figure}

Here, we adopt the robustness of athermality $R_\mathrm{T}$ to measure the thermodynamic resource inherent to quantum states and channels. For a general state $\rho_A$, the robustness of athermality $R_{\rm T}(\rho_A)$ is defined as
\begin{align}
    R_\mathrm{T}(\rho_A) 
    &:=\min\left\{s|\rho_A\le (1+s) \gamma_A\right\}.
    \label{eq:R_T_st_def}
\end{align}
Throughout, we consider full-rank thermal states $\gamma$, thus avoiding divergence of $R_\mathrm{T}$. When $\rho_A$ is incoherent, $R_\mathrm{T}(\rho_A)$ relates to the necessary work for preparing $\rho_A$ from $\gamma_A$ (the \textit{work of formation})~\cite{horodecki2013fundamental}.
The set of athermality non-generating operations from $A$ to $B$ are the \textit{Gibbs-preserving operations}
\begin{align}
    \mathrm{GPO}(A\rightarrow B) := \left\{G | G(\gamma_A) = \gamma_B, \, G\in\mathcal{O}(A\rightarrow B)\right\}.
\end{align}
For a general channel $\Lambda\in\mathcal{O}(A\rightarrow B)$, the robustness of athermality $R_\mathrm{T}(\Lambda)$ is defined as 
\begin{equation}
    R_\mathrm{T}(\Lambda):= \min\left\{s| \Lambda \le (1+s) G,\, G\in \mathrm{GPO}(A\rightarrow{B})\right\},
    \label{eq:R_T_ch_def}
\end{equation}
We use $R_\mathrm{T}$ for both static and dynamical athermality because dynamical athermality can be reduced to static athermality (Theorem~\ref{thm:ER_ROA}). The type of athermality $R_\mathrm{T}$ will be unambiguous from its argument.

Turning to information transmission of a general quantum channel $\Lambda$, its signalling ability can be quantified by the robustness of signalling $R_\mathrm{S}$, where the free operations are non-signalling channels~\cite{fang2019quantum,takagi2020application,milz2022resource}.
The set of \textit{non-signalling operations} (NSOs) from $A$ to $B$ is given by trace-and-replace channels $\Phi_\sigma(\cdot) \equiv \mathrm{Tr}\{\cdot\}\sigma$ as 
\begin{align}
    \mathrm{NSO(A\rightarrow B)} := \left\{\Phi_\sigma | \sigma\in\mathcal{S}(B)\right\}.
\end{align}
For a general channel $\Lambda\in\mathcal{O}(A\rightarrow B)$, $R_\mathrm{S}(\Lambda)$ is defined as
\begin{align}
R_\mathrm{S}(\Lambda):= \min\left\{s|\Lambda\le (1+s)\Phi, \, \Phi\in\mathrm{NSO}(A\rightarrow B)\right\}.
\label{eq:R_S_def}
\end{align}
In Ref.~\cite{milz2022resource}, $R_\mathrm{S}(\Lambda)$ has been interpreted as a measure of the ``strength" of causal loops that can be closed by contracting $\Lambda$ with a general channel $\Omega\in\mathcal{O}(B\rightarrow A)$, and it has been shown that the logarithm of $R_\mathrm{S}$ asymptotically converges to the channel mutual information~\cite{fang2019quantum}.

\subsection{Joint measure of athermality and signalling of quantum channels} 

$R_\mathrm{T}$ and $R_\mathrm{S}$ measure athermality and signalling of a quantum channel, respectively, due to the two different free sets in their definitions, GPO and NSO.
Given the input system $A$ and the output system $B$, the intersection between $\mathrm{GPO}(A\rightarrow B)$ and $\mathrm{NSO}(A\rightarrow B)$ is the completely-thermalising channel
\begin{align}
    \Gamma(\cdot) \equiv \mathrm{Tr}\{\cdot\}\gamma_B.
\end{align}
Note, however, that $\Gamma$ is neither an extreme point in NSO nor in GPO. For NSO, this is straightforward. For GPO, see Ref.~\cite{giulio2024extremepoints} for the conditions on extreme points in GPO.
We thus define a joint measure that collectively gauges signalling and athermality of a general channel $\Lambda\in\mathcal{O}(A\rightarrow B)$ as
\begin{align}
    R(\Lambda||\Gamma) := \min\{s|\Lambda\le (1+s)\Gamma\}.
    \label{eq:R_def}
\end{align}
Given the close connection between robustness measures and the max-relative entropy, $R(\Lambda||\Gamma)$ can also be studied as a specific instance of a generalized channel divergence with respect to the completely thermalising channel~\cite{sohail_fundamental_2025,badhani2025thermodynamicsquantumprocessesoperational}.
Operationally, $R(\Lambda||\Gamma)$ measures the ``distance" between $\Lambda$ and the completely thermalising channel $\Gamma$, and therefore it follows from its definition that
\begin{align}
    R(\Lambda||\Gamma) &\ge \max\left\{R_\mathrm{T}(\Lambda),\, R_\mathrm{S}(\Lambda)\right\}.
    \label{eq:R(Lambda||Gamma)_lower_bound}
\end{align}
We note that equality holds when $\Phi\in\mathrm{NSO}(A\rightarrow B)$, such that $R_\mathrm{S}(\Phi) = 0$ and $R(\Phi||\Gamma) = R_\mathrm{T}(\Phi)$. However, when $G\in\mathrm{GPO}(A\rightarrow B)$, $R(G||\Gamma) \ge R_\mathrm{S}(G)$ in general.
Following the framework of resource preservability, where some free operations in a resource theory are considered as resourceful with respect to resource-destroying operations, owing to their ability to preserve resources~\cite{hsieh2020resource}, we set
\begin{align}
    P_\mathrm{T}(G) := R(G||\Gamma)\text{ for } G\in\mathrm{GPO}(A\rightarrow B),
    \label{eq:P_T_def}
\end{align}
which is the athermality preservability of the GPO $G$. 
To conclude this section, in Table~\ref{tab:resources_mainbody}, we summarise the four types of resources included in this work. For more properties of these resource measures, we refer to Appendix~\ref{sec:basic_properties}.
    \begin{table*}
    \centering
    \begin{ruledtabular}
    \begin{tabular}{cccc}
        \textrm{Resources} &  \textrm{Symbols} & \textrm{Types} &\textrm{Definitions}\\
        \colrule
         Athermality &  $R_\mathrm{T}$ & Static / Dynamical & Eq.~(\ref{eq:R_T_st_def}) / Eq.~(\ref{eq:R_T_ch_def}) \\
         Signalling & $R_\mathrm{S}$ & Dynamical & Eq.~(\ref{eq:R_S_def}) \\
         Joint resource & $R$ & Dynamical & Eq.~(\ref{eq:R_def}) \\
         Athermality preservability for GPOs & $P_\mathrm{T}$ & Dynamical & Eq.~(\ref{eq:P_T_def})
    \end{tabular}
    \end{ruledtabular}
    \caption{{Four resource measures considered in this work.}}
    \label{tab:resources_mainbody}
\end{table*}

\section{Generation, preservation and transmission of athermality}\label{sec:three_tasks}

We now introduce three thermodynamic tasks associated with a general quantum channel $\Lambda$ and show that they are characterised by the resource measures defined above. Considering a purification of $\gamma_A$ as input, for simplicity, we start with the thermofield double state~\cite{Takahashi1996ThermoField}~$\tilde{\gamma}_{AA'}\equiv \ket{\tilde{\gamma}}\!\bra{\tilde{\gamma}}_{AA'}$ with $\ket{\tilde{\gamma}}_{AA'} \equiv \sum_{i}\sqrt{g_i}\ket{i}_{A}\otimes\ket{i}_{A'}$. Here, $g_i$ is the $i$th population of $\gamma_A$, corresponding to energy eigenstate $\ket{i}_A$. We have $\mathrm{Tr}_{\{A,A'\}\text{\textbackslash} X}\{\tilde{\gamma}_{AA'}\} = \gamma_{X}$ for $X\in\{A, A'\}$. A generalisation to all purified thermal states on $AA'$ is given in Appendices~\ref{app:proof_Thm_ER_R} and \ref{app:proof_Thm_upper_bound_R}.

We consider the following three thermodynamic tasks and related resources (see Fig.~\ref{fig:thermal_tasks}):
\begin{enumerate}
    \item[I.] {\bf Athermality generation ($R_\mathrm{T}$):} The amount of athermality a quantum channel $\Lambda\in\mathcal{O}(A\rightarrow B)$ generates from the thermal state $\gamma_A$.
    \item[II.] {\bf Athermality preservation ($R$):} The amount of athermality a local quantum channel $\Lambda\in\mathcal{O}(A'\rightarrow B)$ can preserve when acting on $\tilde{\gamma}_{AA'}$. 
    \item[III.] {\bf Athermality transmission ($R - R_\mathrm{T}$):} The amount of athermality a local quantum channel $\Lambda\in\mathcal{O}(A'\rightarrow B)$ can transmit when acting on $\tilde{\gamma}_{AA'}$.
\end{enumerate}
\begin{figure}
    \centering
    \includegraphics[width=\linewidth]{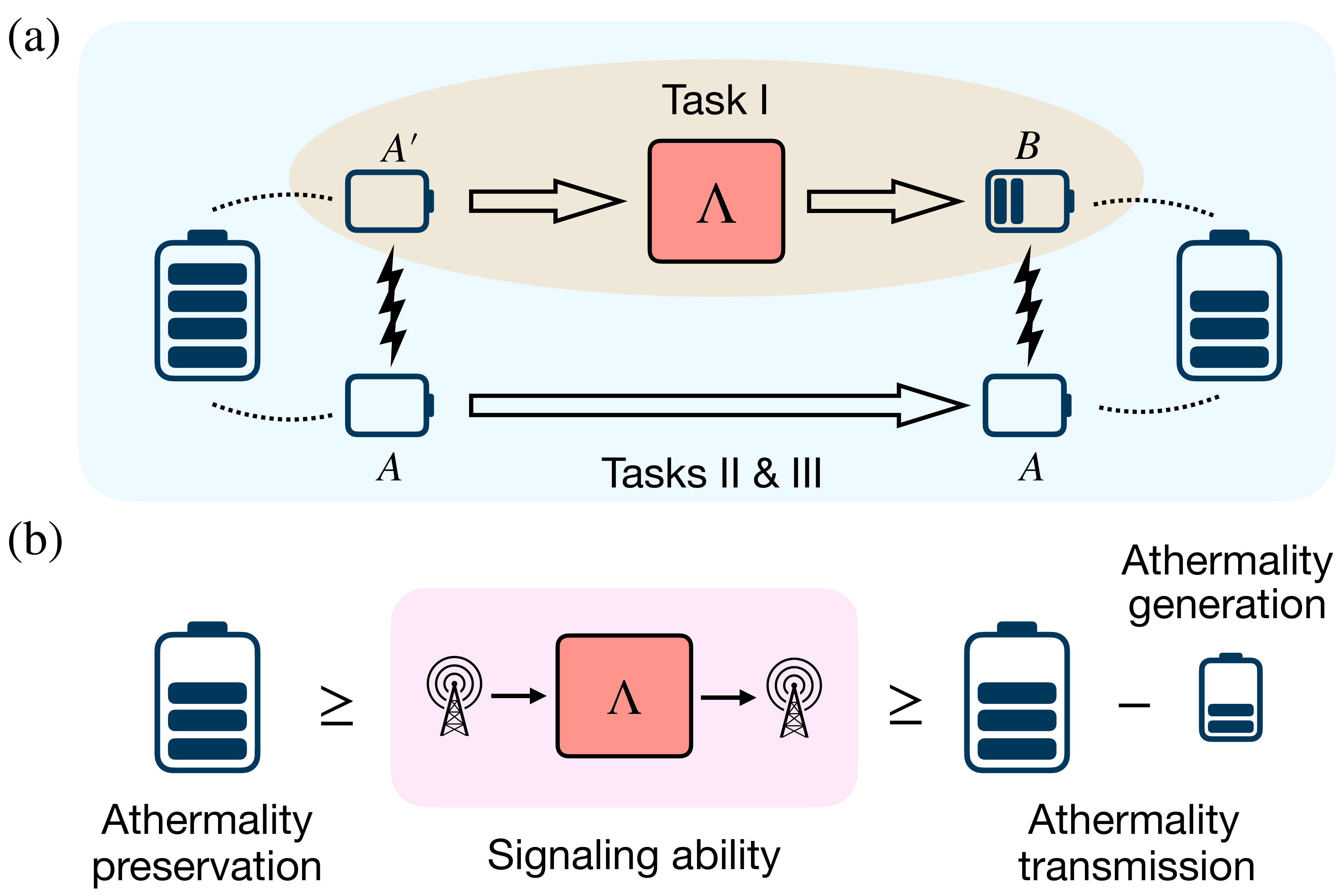}
    \caption{\textbf{Main results.} (a) Three thermodynamic tasks associated with a quantum channel $\Lambda\in\mathcal{O}(A'\rightarrow B)$ correspond to three quantifiable resources: Task I (athermality generation)  -- how much athermality can $\Lambda$ generate from a thermal input? Task II (athermality preservation) -- how much input athermality does $\Lambda$ preserve if $A'$ is part of a pure, locally thermal state? Task III (athermality transmission) -- how much of the initial global athermality from correlations between $A$ and $A'$ does $\Lambda$ transfer? Athermality is denoted by charged batteries and empty batteries indicate thermal states.
    (b) The signalling ability of the quantum channel $\Lambda$ is bounded by its thermodynamic capacities (Theorem~\ref{thm:R_S_bounds}).}
    \label{fig:thermal_tasks}
\end{figure}

The athermality generating ability of a quantum channel $\Lambda$ is equivalent to its dynamical athermality $R_\mathrm{T}(\Lambda)$:
\begin{theorem}\label{thm:ER_ROA}
    For a general channel $\Lambda\in\mathcal{O}(A\rightarrow B)$, the following equality holds:
    \begin{align}
        \RT(\Lambda) = \RT[\Lambda(\gamma_A)].
    \end{align}
\end{theorem}
The proof is provided in Appendix~\ref{app:proof_Thm_ER_ROA}.
Hence, Task~I is quantified by $R_\mathrm{T}(\Lambda)$. 
We remark that similar relations between static and dynamic resources have been found for coherence~\cite{diaz2018using} and entanglement~\cite{lami2023computable}. 
Analogously, the joint dynamical resource measure $R$ can also be reduced to a static resource:
\begin{theorem}\label{thm:ER_R}
    For a general channel $\Lambda\in\mathcal{O}(A'\rightarrow B)$, the following equality holds:
    \begin{align}
        R(\Lambda||\Gamma) = R_\mathrm{T}[\mathcal{I}\otimes\Lambda(\tilde{\gamma}_{AA'})],
    \end{align}
    where $\mathcal{I}$ is the identity channel.
\end{theorem}
A proof of a stronger version of this theorem is provided in Appendix~\ref{app:proof_Thm_ER_R}.
Therefore, $R(\Lambda||\Gamma)$ is equivalent to the output athermality when $\Lambda$ is locally applied on one part of the thermofield double state $\tilde{\gamma}_{AA'}$. To see that $R$ quantifies athermality \textit{preservation} (Task II) -- and not athermality generation, like $R_\mathrm{T}$ -- consider the following Theorem:
\begin{theorem}\label{thm:upper_bound_R}
    For a general channel $\Lambda\in\mathcal{O}(A\rightarrow B)$ with $\gamma_A = \gamma_B \equiv \gamma$, the following inequality holds:
    \begin{align}
        R(\Lambda||\Gamma) \le R_\mathrm{T}(\tilde{\gamma}_{AA'}) = \mathrm{Tr}\{\gamma^{-1}\} - 1,
    \end{align}
    where ${\rm Tr}\{\gamma^{-1}\}$ denotes the sum over the inverses of the eigenvalues of $\gamma$ and the equality holds when $\Lambda$ is unitary.
\end{theorem}
A proof of a stronger version of this theorem is provided in Appendix~\ref{app:proof_Thm_upper_bound_R}.
Combined with Theorem~\ref{thm:ER_R}, this yields a no-go relation:
\begin{align}
    R_\mathrm{T}[\mathcal{I}\otimes\Lambda(\tilde{\gamma}_{AA'})] \le R_\mathrm{T}(\tilde{\gamma}_{AA'}),
    \,\forall\, \Lambda\in\mathcal{O}(A'\rightarrow B).
\end{align}
That is,~no local channel can  increase the overall athermality of the thermofield double state $\tilde{\gamma}_{AA'}$, rendering $R(\Lambda||\Gamma)$ a quantifier of athermality \textit{preservation}.

An important property of $\tilde{\gamma}_{AA'}$ is that it is locally thermal in both subsystems with all athermality residing solely in the correlations between $A$ and $A'$. The transmission of this type of athermality is encapsulated in Task~III.

By Theorems~\ref{thm:ER_ROA} and~\ref{thm:ER_R}, $R(\Lambda||\Gamma)$ quantifies the total athermality in $\tilde{\gamma}_{AA'}$ that is preserved under the action of $\Lambda$, while $R_\mathrm{T}(\Lambda)$ quantifies the athermality that $\Lambda$ generates in the local system $A'$. The difference $R(\Lambda||\Gamma) - R_\mathrm{T}(\Lambda)$ thus quantifies the athermality contained in $\tilde{\gamma}_{AA'}$ transmitted via the local channel $\Lambda$. We conclude that Task~III is measured by $R - R_\mathrm{T}$.

\section{Thermodynamic criteria for signalling}\label{sec:thermal_criteria}

We now show how the signalling capability of a general quantum channel is limited by its thermodynamic properties.
\begin{theorem}\label{thm:R_S_bounds}
    Given a general channel $\Lambda\in\mathcal{O}(A\rightarrow B)$,
    \begin{align}
        R(\Lambda||\Gamma) \ge R_\mathrm{S}(\Lambda) \ge 2\left(g_{\min}^{(AB)}\right)^2[R(\Lambda||\Gamma)-R_\mathrm{T}(\Lambda)]^2,
        \label{eq:R_S_bounds}
    \end{align}
    where $g_{\min}^{(AB)}$ is the smallest population in the energy eigenbasis of the thermal state $\gamma_{AB} \equiv \gamma_A\otimes\gamma_B$. 
\end{theorem}
The proof is provided in Appendix~\ref{app:proof_Thm_R_S_bounds}.
The upper and lower bounds on $R_\mathrm{S}(\Lambda)$ can be viewed as necessary and sufficient conditions for signalling channels, respectively:
\begin{align*}
    &\text{Athermality preservation }(R > 0) \,\Leftarrow\, R_\mathrm{S} > 0, \\
    &\text{Athermality transmission }(R - R_\mathrm{T} > 0) \,\Rightarrow\, R_\mathrm{S} > 0. 
\end{align*}
For GPOs, athermality transmitability is equivalent to athermality preservability. Recall the definition of $P_{\rm T}(G):=R(G||\Gamma)$ for $G\in{\rm GPO}$ [Eq.~(\ref{eq:P_T_def})].
For GPOs, Eq.~(\ref{eq:R_S_bounds}) becomes
\begin{align}
    P_\mathrm{T}(G) \ge R_\mathrm{S}(G) \ge 2 \left(g_{\min}^{(AB)}\right)^2 P_\mathrm{T}(G)^2.
    \label{eq:R_S_bound_GPO}
\end{align}
That is, non-zero athermality preservability guarantees non-zero signalling ability, and vice versa. According to the bounds, the completely-thermalising channel $\Gamma$ is the only GPO can never signal, as expected.

Rearranging the lower bound in Eq.~(\ref{eq:R_S_bounds}), we obtain
\begin{align}
    R(\Lambda||\Gamma) \le R_\mathrm{T}(\Lambda) + \left(\sqrt{2}\,g_{\min}^{(AB)}\right)^{-1}\!\sqrt{R_\mathrm{S}(\Lambda)}.
    \label{eq:R(Lambda||Gamma)_upper_bound}
\end{align}
Eq.~(\ref{eq:R(Lambda||Gamma)_upper_bound}) together with Eq.~(\ref{eq:R(Lambda||Gamma)_lower_bound}) bounds the joint measure $R(\Lambda||\Gamma)$ from both sides in terms of the two resources $R_\mathrm{T}(\Lambda)$ and $R_\mathrm{S}(\Lambda)$. In addition, $R$ and $R_\mathrm{T}$ are operationally directly connected through a channel simulation task, as we now detail.

\section{GPO dilation of a quantum channel}\label{sec:GPO_dilation}
We now show that the joint measure $R$ is not only a generalisation of $P_\mathrm{T}$ to general quantum channels, but also admits an operational interpretation in terms of $P_\mathrm{T}$. We consider a channel simulation scenario that we dub as \textit{GPO dilation} of quantum channels. A GPO dilation of a general channel $\Lambda\in\mathcal{O}(A\rightarrow B)$, as is depicted in Fig.~\ref{fig:channel_simulation_GPO}, can be viewed as the combination of a GPO $G\in\mathrm{GPO}(CA\rightarrow B)$ and an athermal state $\rho_C\in\mathcal{S}(C)$, such that $\Lambda(\cdot) = G(\rho_C \otimes \cdot)$. We emphasise that in a GPO dilation, the GPO $G$ needs not to be unitary, which is different from the Stinespring's representation of quantum channels~\cite{Stinespring1955}. Now, the following theorem illustrates that
$R(\Lambda||\Gamma)$ bounds the minimal athermality preservability of the GPO dilation of $\Lambda$. 
\begin{figure}
\centering
    \includegraphics[width=0.85\linewidth]{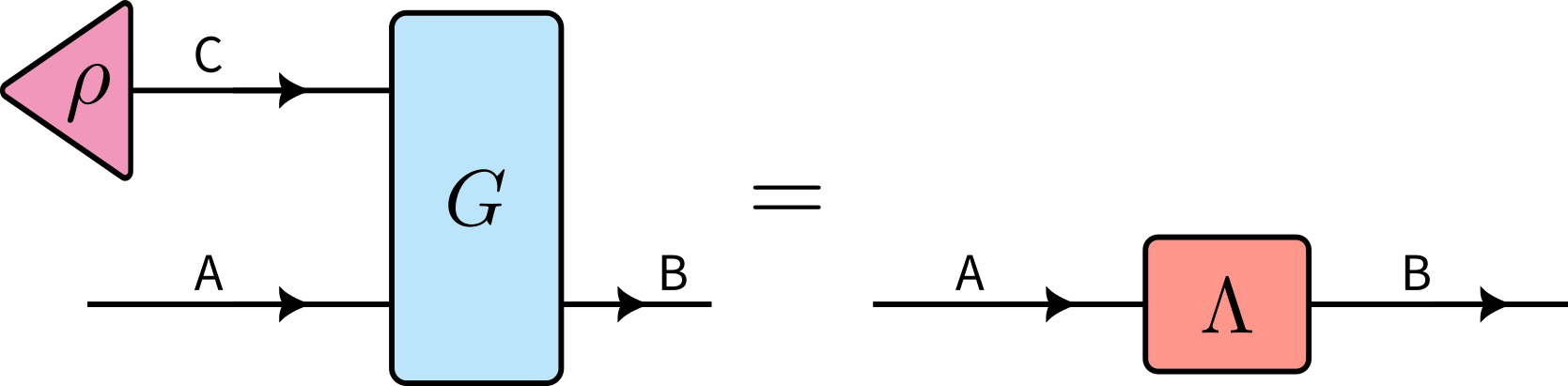}
    \caption{\textbf{GPO dilation.} A general channel $\Lambda$ can be dilated by a channel $G\in\mathrm{GPO}(CA\rightarrow B)$ with some state $\rho_C$.}
    \label{fig:channel_simulation_GPO}
\end{figure}
\begin{theorem}\label{thm:R_GPO_dilation}
Given a general quantum channel $\Lambda \in \mathcal{O}(A\rightarrow B)$, we define the set of GPOs which can simulate this channel with some state $\rho_C$ as $\mathcal{O}^{\,\mathrm{sim}}_\mathrm{GPO}(\Lambda)$. The following inequalities hold:
\begin{align}
    R(\Lambda||\Gamma) \le \min_{G\in\mathcal{O}^{\,\mathrm{sim}}_\mathrm{GPO}(\Lambda)} P_\mathrm{T}(G) \le \max\left\{1/R_\mathrm{T}(\Lambda), \, R(\Lambda||\Gamma)\right\}.
    \label{eq:R(Lambda||Gamma)_bounds_P_T(G)}
\end{align}
Especially, whenever $R_\mathrm{T}(\Lambda) R(\Lambda||\Gamma) \ge 1$, we have
\begin{align}
    R(\Lambda||\Gamma) = \min_{G\in\mathcal{O}^{\,\mathrm{sim}}_\mathrm{GPO}(\Lambda)} P_\mathrm{T}(G).
    \label{eq:R(Lambda||Gamma)_=_P_T(G)}
\end{align}
\end{theorem}
The proof is provided in Appendix~\ref{app:GPO_dilation}.
We remark that when $\Lambda \equiv G\in\mathrm{GPO}(A\rightarrow B)$, although the upper bound diverges, the lower bound becomes tight. Besides, since $R(\Lambda||\Gamma)\ge R_\mathrm{T}(\Lambda)$, $R_\mathrm{T}(\Lambda)R(\Lambda||\Gamma)\ge 1,\,\forall\, R_\mathrm{T}(\Lambda)\ge 1$. Thus, Eq.~(\ref{eq:R(Lambda||Gamma)_=_P_T(G)}) holds when $R_\mathrm{T}(\Lambda)\notin (0,1)$.

In Fig.~\ref{fig:R_T_R_S_R}, we numerically investigate the typicality of the condition $R_\mathrm{T}(\Lambda)R(\Lambda||\Gamma)\ge 1$ for a random qubit channel $\Lambda$, where 20000 qubit channels are randomly sampled via the protocol in Ref.~\cite{kukulski2021generating} which gives a flat probability measure on the set of quantum channels. We found that this condition is satisfied by $17643/20000\approx 88.2\%$ channels, implying that $R(\Lambda||\Gamma) = \min_{G\in\mathcal{O}^{\,\mathrm{sim}}_\mathrm{GPO}(\Lambda)} P_\mathrm{T}(G)$ universally holds for most channels. However, whether the equality holds even for all channels with $1/R_{\rm T} > R$
remains as an open question for future investigation.
\begin{figure}
\centering
\includegraphics[width=\linewidth]{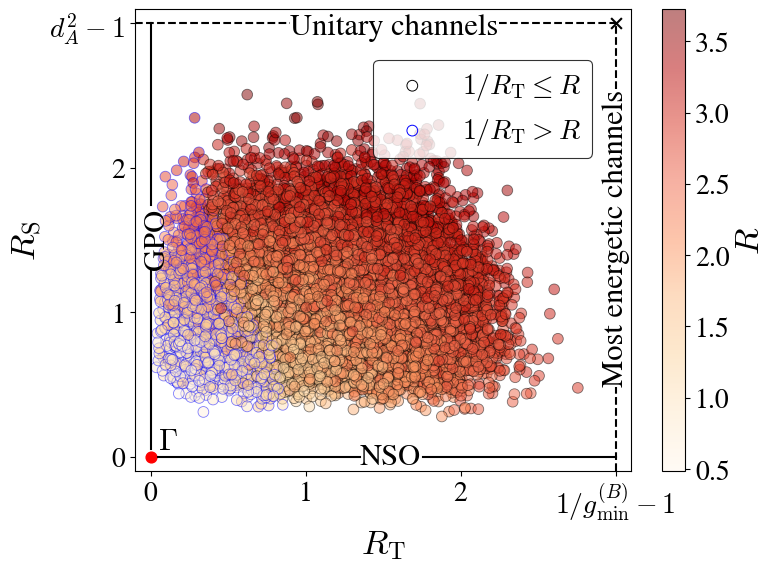}
\caption{{\bf Resource measures $R_\mathrm{T}$, $R_\mathrm{S}$ and $R$ of random qubit channels.}
$20000$ qubit channels in $\mathcal{O}(A\rightarrow B)$ with $d_A=d_B=2$ and $\gamma_A=\gamma_B = 0.75\ket{0}\!\bra{0}+0.25\ket{1}\!\bra{1}$ are randomly generated. Channels with $1/R_\mathrm{T} \le\!\!(>) R$ are marked by black (blue) edgecolor. The set GPO lies in $[0, d_A^2-1]$ on the line of $R_\mathrm{T}=0$ while the set NSO lies in $[0, 1/g_{\min}^{(B)}-1]$ on the line of $R_\mathrm{S}=0$. The completely thermalsing channel $\Gamma$ is at the origin. The highest $R_\mathrm{T}$ and $R_\mathrm{S}$ are achieved by the most energetic channels (sending $\gamma_A$ to $\ket{1}\!\bra{1}_B$) and unitary channels, respectively. The symbol \ding{53} denotes that there is no intersection between the most energetic channels and unitary channels.}
\label{fig:R_T_R_S_R}
\end{figure}

In accordance with Fig.~\ref{fig:R_T_R_S_R}, we report in Table~\ref{tab:resources_upperbounds_mainbody} the channels identified as the most resourceful under the three measures $R_{\rm T}$, $R_{\rm S}$ and $R$. The most energetic channels are defined as the channels mapping the thermal state to the highest energy eigenstate, which has the lowest thermal population when the thermal state is defined at a positive temperature.
\begin{table*}
    \centering
    \begin{ruledtabular}
    \begin{tabular}{cccc}
        \textrm{Resources} &  \textrm{Symbols} & \textrm{Most resourceful channels} &\textrm{Highest values}\\
        \colrule
         Athermality &  $R_\mathrm{T}$ & Most energetic channels & $1/g_{\rm min}^{(B)}-1$ (Theorems~\ref{appthm:ER_ROA} and~\ref{appthm:most_athermal_state}) \\
         Signalling & $R_\mathrm{S}$ & Unitary channels & $d_A^2 - 1$ (Ref.~\cite{takagi2020application}) \\
         Joint resource & $R$ & Unitary channels & $\mathrm{Tr}\{\gamma_B^{-1}\} - 1$ (Theorem~\ref{appthm:upper_bound_R})
    \end{tabular}
    \end{ruledtabular}
    \caption{Most resourceful channels in $\mathcal{O}(A\rightarrow B)$ with $d_A = d_B$ according to the three measures.}
    \label{tab:resources_upperbounds_mainbody}
\end{table*}

Additionally, we note that $R_\mathrm{T}(\Lambda)$ can also be interpreted via the GPO dilation of $\Lambda$. In Appendix~\ref{app:GPO_dilation}, we show that 
\begin{align}
    R_\mathrm{T}(\Lambda) = \min_{\rho_C\in\mathcal{S}_{\mathrm{GPO}}^{\,\mathrm{sim}}(\Lambda)}R_\mathrm{T}(\rho_C),
    \label{eq:R_T(Lambda)_=_R_T(rho_C)}
\end{align}
where $\mathcal{S}_{\mathrm{GPO}}^{\,\mathrm{sim}}(\Lambda)$ is the set of states that can simulate $\Lambda$ with a GPO using the protocol in Fig.~\ref{fig:channel_simulation_GPO}. This result can be interpreted as another manifestation of the equivalence between static and dynamical athermality.

As a consequence, $R_\mathrm{T}(\Lambda)$ and $R(\Lambda||\Gamma)$ are connected via the GPO dilation of $\Lambda$, yet they stem from distinct types of resources in the dilation: $R_{\rm T}(\Lambda)$ derives the requisite static athermality $R_{\rm T}(\rho_C)$, whereas $R(\Lambda||\Gamma)$ reflects the minimal athermality preservability $P_{\rm T}(G)$. 

We now demonstrate the interplay between the three resources $R_{\rm T}$, $R_{\rm S}$ and $R$, and showcase the application of the derived bounds via a channel simulation scenario involving the quantum switch~\cite{chiribella2013quantum}.

\section{Application: Resource analysis of quantum switch} \label{sec:quantum_switch}
The quantum switch $\mathdutchcal{S}$ realises the superposition of causal orders between quantum operations~\cite{chiribella2013quantum}. Intuitively, it transforms two quantum channels and the state of a control qubit into a channel $\Lambda_\mathdutchcal{S}$, where the input quantum channels are implemented in a superposition of their possible orderings if the control is in a superposition of $\ket{0}$ and $\ket{1}$ (see Fig.~\ref{fig:quantum_switch}). Formally, given a general quantum channel $\Omega\in\mathcal{O}(A\rightarrow B)$ with the Kraus representation $\Omega(\cdot) \equiv \sum_n K_n(\cdot) K_n^\dagger$, the Kraus representation of the channel implemented by the quantum switch $\mathdutchcal{S}[\Omega, \Omega]\in\mathcal{O}(AC\rightarrow BC')$, where $C'$ is a copy of $C$, can be written as $\mathdutchcal{S}[\Omega, \Omega](\cdot) \equiv \sum_m\sum_n S_{mn}(\cdot)S^\dagger_{mn}$, where the Kraus operator $S_{mn}$ is given by
\begin{align}
    S_{mn} \equiv \ket{0}\!\bra{0}_C\otimes K_n K_m + \ket{1}\!\bra{1}_C\otimes K_m K_n.
\end{align}
It is known that, with coherent control and two GPOs, the induced channel $\Lambda_\mathdutchcal{S}$ is both signalling and athermal~\cite{liu2022thermodynamics}. However, a comprehensive analysis of how the thermodynamic and signalling properties of $\Lambda_\mathdutchcal{S}$ depend on input control and channels remains lacking.
Considering the switch as a channel simulation task in Fig.~\ref{fig:quantum_switch}, the framework developed above is ideally suited to address this challenge.

Taking two identical signalling GPOs as input channels, we demonstrate that the joint resource $R$ of the induced channel $\Lambda_\mathdutchcal{S}$ is constrained by the supplied resources, while $R_\mathrm{T}$ and $R_\mathrm{S}$ exhibit a trade-off relation.
\begin{figure}
     \centering
     \includegraphics[width=0.9\linewidth]{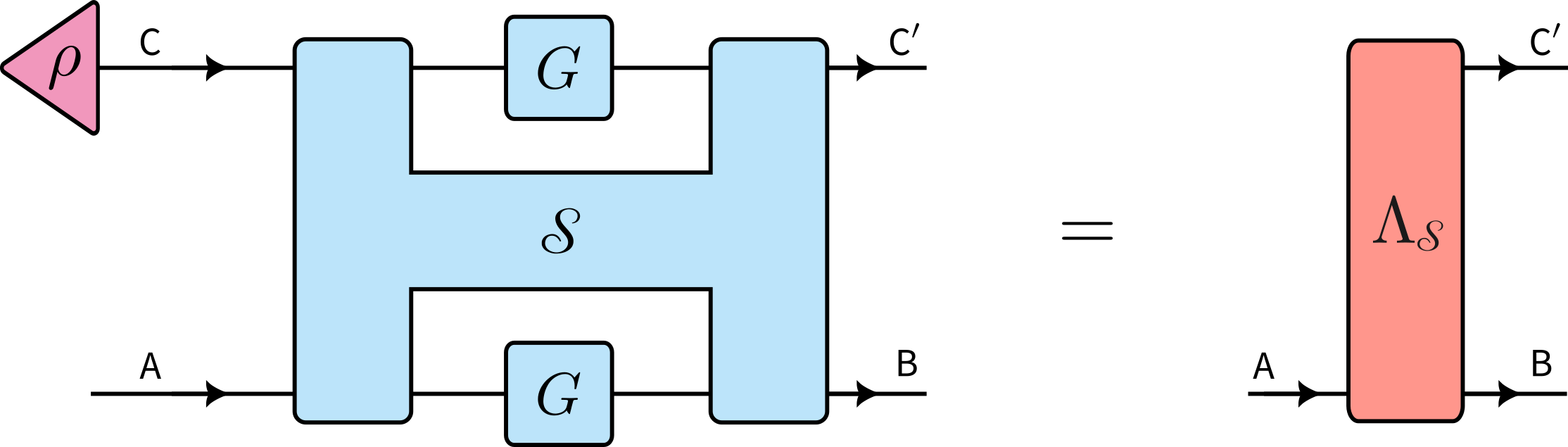}
    \caption{\textbf{Quantum switch. } A quantum switch $\mathdutchcal{S}$ with a control qubit $\rho_C$ and two GPOs as input can be considered as a channel simulation of $\Lambda_{\mathdutchcal{S}}$.}
    \label{fig:quantum_switch}
 \end{figure}

Like Ref.~\cite{liu2022thermodynamics}, we choose the initial state of the control qubit as $\rho_C = \ket{\psi}\!\bra{\psi}_C$, where 
\begin{align}
    \ket{\psi}_C \equiv \sqrt{\alpha}\ket{0}_C + \mathrm{e}^{\mathrm{i}\varphi}\sqrt{1-\alpha}\ket{1}_C,
    \label{eq:control_rho_C}
\end{align}
with $\alpha\in[0,1]$ and $\varphi\in[0,2\pi]$ (general states $\rho_C$ are considered in Appendix~\ref{app:general_rho_C}). The thermal state of the control qubit is $\gamma_C = \mathbbm{1}_C/2$, while all other systems have identical thermal states $\gamma$. To demonstrate the key features of how the quantum switch transforms signalling-athermality resources, we consider the following signalling GPO:
\begin{align}
    G\equiv s\Gamma + (1-s)\mathcal{I},
    \label{eq:signalling_GPO}
\end{align}
where $\mathcal{I}$ is the identity channel and $s\in[0,1]$ is the thermalising strength: at $s=0$, $G$ acts trivially as the identity channel (i.e., it cannot thermalise any athermal input), while at $s=1$, $G$ is completely thermalsing (i.e., it thermalises any input state). Define $\Lambda_\mathdutchcal{S}(\cdot) \equiv \mathdutchcal{S}[G,G](\rho_C\otimes \cdot) \in \mathcal{O}(A\rightarrow BC')$. We can then prove:
\begin{theorem}\label{thm:switch_bounds}
    Given the control state $\rho_C = \ketbra{\psi}{\psi}_C$ in Eq.~(\ref{eq:control_rho_C}) and 
    two input channels $G$ defined as in Eq.~(\ref{eq:signalling_GPO}), the following bounds hold
    \begin{align}
        R(\Lambda_\mathdutchcal{S}||\Gamma) &\le R_\mathrm{T}(\rho_C) + 2P_\mathrm{T}^2(G)/P_\mathrm{T}(\mathcal{I}), \label{eq:switch_bound_upper}\\
        R(\Lambda_\mathdutchcal{S}||\Gamma) &\ge R_\mathrm{S}(\Lambda_\mathdutchcal{S}) \ge (g_{\min}^4/2)[R(\Lambda_\mathdutchcal{S}||\Gamma) - R_\mathrm{T}(\Lambda_\mathdutchcal{S})]^2,
        \label{eq:switch_bound_lower}
    \end{align}
    where $R_\mathrm{T}(\rho_C) = 1,\,\forall\, \alpha \in[0,1],\, \varphi\in[0,2\pi]$ and 
    \begin{align}
        R_\mathrm{T}(\Lambda_\mathdutchcal{S}) = \sqrt{1-4(1-g_{\max}^2)[2-(1-g_{\max}^2)s^2]s^2\alpha(1-\alpha)},
        \label{eq:R_T(Lambda_S)}
    \end{align}
    for $\alpha, s\in[0,1]$, $\varphi\in[0,2\pi]$, and $g_{\min(\max)}$ the minimal (maximal) population in the energy eigenbasis of $\gamma$. The upper bound of $R(\Lambda_\mathdutchcal{S}||\Gamma)$ is saturated for either $\alpha$ or $s = 0, 1$.
\end{theorem}
The proof is provided in Appendix~\ref{app:quantum_switch}.
The upper bound [Eq.~(\ref{eq:switch_bound_upper})] shows explicitly that the joint resource $R$, and thus the signalling robustness $R_\mathrm{S}$, of the induced channel $\Lambda_\mathdutchcal{S}$ is bounded by the thermodynamic resource input into the switch, i.e.,~the athermality $\rho_C$ and the athermality preservability of the GPO $G$ (Fig.~\ref{fig:quantum_switch}). This coincides with our intuition, because to activate the switch, the control qubit should be coherent, and therefore athermal, and the athermality preserving $G$ with non-vanishing signalling ability [Eq.~(\ref{eq:R_S_bound_GPO})] should contribute to the signalling ability of the induced channel $\Lambda_\mathdutchcal{S}$. Therefore, $R(\Lambda_\mathdutchcal{S}||\Gamma)$ is the correct characterisation of the resource converted by the switch from input to output. Notably, when $s = 0$ ($1$), $R(\Lambda_\mathdutchcal{S}||\Gamma) = R_\mathrm{T}(\rho_C) + 2P_\mathrm{T}(\mathcal{I})$ ($R_\mathrm{T}(\rho_C)$) -- precisely the resource inputs. On the other hand, the lower bound [Eq.~(\ref{eq:switch_bound_lower})] is a direct application of Theorem~\ref{thm:R_S_bounds} and can be understood as a trade-off between $R_\mathrm{S}(\Lambda_\mathdutchcal{S})$ and $R_\mathrm{T}(\Lambda_\mathdutchcal{S})$ (see below). In addition, Theorem~\ref{thm:switch_bounds} provides an analytical expression for $R_\mathrm{T}(\Lambda_\mathdutchcal{S})$. Consequently, all three types of resources -- $R$, $R_\mathrm{S}$, and $R_\mathrm{T}$ -- are either bounded or expressed analytically, providing a complete resource analysis of the quantum switch.

To show the trade-off between $R_\mathrm{T}$ and $R_\mathrm{S}$ more clearly, we now consider the case when both input GPOs $G = \Gamma$ (i.e., $s=1$), such that $\rho_C$ constitutes the only resource input. This case is a canonical example in previous demonstrations of the communication/thermodynamic capabilities of quantum switches with NSOs/GPOs~\cite{ebler2018enhanced,felce2020quantum}. By Theorem~\ref{thm:switch_bounds}, we have $R(\Lambda_\mathdutchcal{S}||\Gamma) = R_\mathrm{T}(\rho_C)$. Eq.~(\ref{eq:switch_bound_lower}) reduces to the following form: 
\begin{align}
    R_\mathrm{T}(\rho_C) \ge R_\mathrm{S}(\Lambda_\mathdutchcal{S}) \ge (g_{\min}^4/2)\left[R_\mathrm{T}(\rho_C) - R_\mathrm{T}(\Lambda_\mathdutchcal{S})\right]^2,
    \label{eq:switch_bounds_Gamma}
\end{align}
for all $\alpha\in [0,1]$. 
Consistently, the resource gain, $R_\mathrm{S}(\Lambda_\mathdutchcal{S})$ and $R_\mathrm{T}(\Lambda_\mathdutchcal{S})$, are both upper-bounded by $R_\mathrm{T}(\rho_C)$. The lower bound in Eq.~(\ref{eq:switch_bounds_Gamma}) hence indicates a trade-off between $R_\mathrm{S}(\Lambda_\mathdutchcal{S})$ and $R_\mathrm{T}(\Lambda_\mathdutchcal{S})$ constrained by $R_\mathrm{T}(\rho_C)$. Since the channel $\mathdutchcal{S}[G,G]$ is Gibbs-preserving~\cite{liu2022thermodynamics}, the control state $\rho_C$ and the channel $\mathdutchcal{S}[G,G]$ can be regarded as a GPO dilation of $\Lambda_\mathdutchcal{S}$. Recall that by Eq.~(\ref{eq:R_T(Lambda)_=_R_T(rho_C)}),
optimally, simulating $\Lambda_\mathdutchcal{S}$ only requires the amount of athermality as same as $R_{\rm T}(\Lambda_\mathdutchcal{S})$.
The excess term $\RT(\rho_C) - \RT(\Lambda_\mathdutchcal{S})$ thus is the \textit{athermality waste} in the channel simulation realised by the switch. The lower bound in Eq.~(\ref{eq:switch_bounds_Gamma}) shows that the signalling ability of the induced channel $\Lambda_\mathdutchcal{S}$ is actually activated by a non-vanishing athermality waste, rather than the athermality of $\rho_C$.
\begin{figure}
    \centering
    \includegraphics[width=1\linewidth]{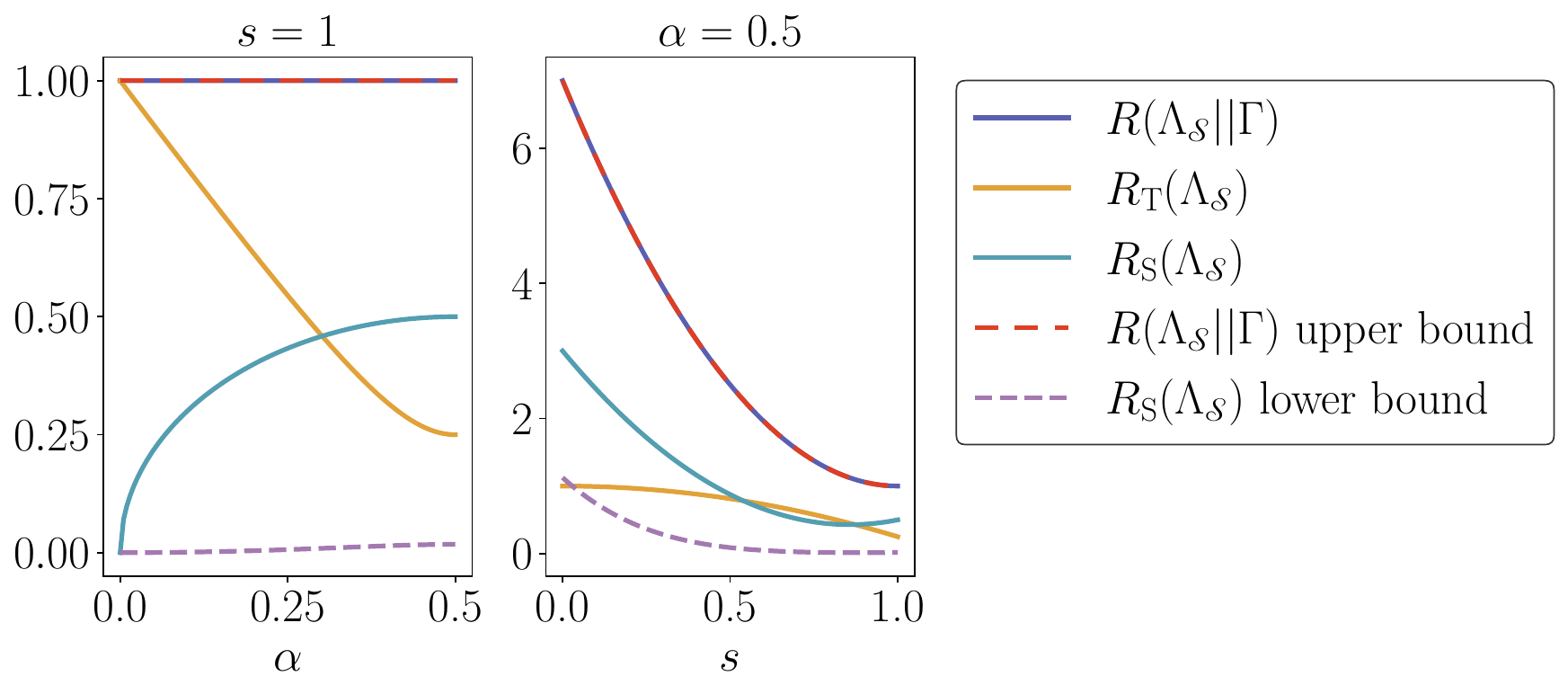}
    \caption{\textbf{Resource analysis. }$R$, $R_\mathrm{T}$, and $R_\mathrm{S}$ of the induced channel $\Lambda_\mathdutchcal{S}$ with their bounds/analytical expression given in Theorem~\ref{thm:switch_bounds} (for qubit systems $A$ and $B$  with $\gamma_A = \gamma_B = \mathbbm{1}/2$).}
    \label{fig:switch_example}
\end{figure}

In Fig.~\ref{fig:switch_example}, we exemplify Theorem~\ref{thm:switch_bounds} by considering qubits $A$ and $B$ with $\gamma_A = \gamma_B = \mathbbm{1}/2$, and plotting the resources of $\Lambda_\mathdutchcal{S}$ as functions of $\alpha$, and $s$, respectively. We find that the upper bound of $R(\Lambda_\mathdutchcal{S}||\Gamma)$ [Eq.~(\ref{eq:switch_bound_upper})] remains tight in both cases while the lower bound of $R_\mathrm{S}(\Lambda_\mathdutchcal{S})$ [Eq.~(\ref{eq:switch_bound_lower})] is only relatively tighter for more signalling input GPOs (smaller $s$). Despite the potential looseness of the lower bound, it provides an interpretation of the phenomenon that increased $R_\mathrm{S}(\Lambda_\mathdutchcal{S})$ comes at the expense of $R_\mathrm{T}(\Lambda_\mathdutchcal{S})$, aligning with their trends in both plots.

Finally, we remark that the above resource analysis is not specific to quantum switches. In Appendix~\ref{app:coherent_control}, we provide a similar analysis to coherent control of quantum channels, where instead of the order of two input channels, the implementation of one of the two channels is conditioned on the state of a control qubit~\cite{abbott2020communication}. In this setup, $R_\mathrm{T}$ and $R_\mathrm{S}$ are found to exhibit a similar trade-off.

\section{Discussion and Conclusion}\label{sec:conclusion}
In this work, we have provided a full-fledged characterisation of the signalling and thermodynamic properties of quantum channels and demonstrated how they mutually bound and condition each other. Our main results are the precise identification of the thermodynamic resources for the generation,
preservation, and transmission of athermality (Theorems~\ref{thm:ER_ROA},~\ref{thm:ER_R}, and~\ref{thm:upper_bound_R}), the characterisation of a quantum channel's signalling ability in terms of its capacity for athermality preservation and transmission (Theorem~\ref{thm:R_S_bounds}), and an operational interpretation of the resulting resource measures in terms of channel-dilatability by a GPO (Theorem~\ref{thm:R_GPO_dilation}). We finally used our results for an encompassing characerisation of resources in the quantum switch (Theorem~\ref{thm:switch_bounds}).

We note the relation of our findings to those of Ref.~\cite{hsieh2025dynamical}, which demonstrates that, in the asymptotic limit, information and energy transmission in specific incoherent setups are equivalent. Here, in contrast, we do not restrict ourselves to the classical case -- in particular, the considered states $\widetilde{\gamma}_{AA'}$ are entangled. On the other hand, the thermodynamic resource we analysed is athermality, a joint measure of work and coherence, while Ref.~\cite{hsieh2025dynamical} investigated extractable work. Finally, according to our measure of athermality transmitability, a signalling channel ($R_\mathrm{S} > 0)$ may have zero athermality transmitability ($R-R_\mathrm{T} = 0)$ (See Appendix~\ref{app:exp_signalling_ch_zero_athermality_transmitability} for examples).

Using the trade-off relations of athermality and signalling, we demonstrated the interplay of these properties for the paradigmatic example of the quantum switch. This resolves the open question posed in Ref.~\cite{liu2022thermodynamics}, showing that only limited activation of information capacity is possible when the switch is composed with GPOs, as the output resources are bounded by those at the input.

A natural next step is to analyse how the bounds we provide behave when thermal operations or related subsets of GPOs~\cite{Janzing2000, brandao2013resource, horodecki2013fundamental, CGPO2015,lostaglio_elementary_2018,Faist2015Gibbs, ding_exploring_2021, tajima2025Gibbs} are considered as thermodynamically resourceless channels instead of GPOs.  Since these channels are all contained within GPOs, our core bounds Eq.~(\ref{eq:R_S_bounds}) still hold -- albeit possibly in a looser sense. Another promising research direction is the study of general quantum supermaps as a complete framework for manipulating dynamical resources and their conversion of different types of resources~\cite{theurer_quantifying_2019, liu2019resource, Liu2020, gour_dynamical_2020, gour_dynamical_2020a, gour_entanglement_2021, gour_inevitable_2024}. Concretely, one could ask: What kind of supermaps do not increase the total resources $R_\mathrm{T}+R_\mathrm{S}$ but can exchange them? Noticing that $2R(\Lambda||\Gamma)\ge R_\mathrm{T}(\Lambda)+R_\mathrm{S}(\Lambda)$, one may think $\Gamma$-preserving supermaps are the direct answer, but we find that there exist counterexamples where $\Gamma$-preserving supermaps increase the sum $R_\mathrm{T}+R_\mathrm{S}$.

Having combined signalling and athermality, another possible direction is to study other resource theories where two distinct free sets have an intersection. In this case, the results from Ref.~\cite{Armin1992distance} may be applied to resource measures defined via distance metrics~\cite{li2020quantifying}.

\textit{Acknowledgements.---}
This publication has emanated from research conducted with the financial support of Taighde Éireann - Research Ireland under Grant number IRCLA/2022/3922. This publication was made possible through the support of Grant 62423 from the John Templeton Foundation. The opinions expressed in this publication are those of the authors and do not necessarily reflect the views of the John Templeton Foundation. YL is supported by China Scholarship Council (No.~202408060137). SM acknowledges funding from the European Union’s Horizon Europe research and innovation programme under the Marie Sk{\l}odowska-Curie grant agreement No.~101068332.

\onecolumngrid

\newtheorem{theoremApp}{Theorem}[section]
\newtheorem{lemmaApp}{Lemma}[section]
\newtheorem{corollaryApp}{Corollary}[theoremApp]
\numberwithin{figure}{section}
\appendix

\numberwithin{equation}{section}

\section{Resource measures}\label{app:resource_measures}
In this work, we focus on two quantum resources: athermality and signalling, both of which can be quantified in terms of robustness monotones~\cite{piani2016robustness,lami2023computable,liu2019one}. In Table~\ref{tab:resources}, we summarise the specific resource quantifiers considered in this paper. We remark that all these measures can be expressed as semidefinite programs (SDPs)~\cite{skrzypczyk2023semidefinite}. 
\begin{table}[b]
    \centering
    \begin{ruledtabular}
    \begin{tabular}{cccc}
        \textrm{Resources} &  \textrm{Symbols} & \textrm{Types} &\textrm{Definitions}\\
        \colrule
         Athermality &  $R_\mathrm{T}$ & Static / Dynamical & Eq.~(\ref{appeq:R_T_st_def}) / Eq.~(\ref{appeq:R_T_ch_def}) \\
         Signalling & $R_\mathrm{S}$ & Dynamical & Eq.~(\ref{appeq:R_S_def}) \\
         Joint resource & $R$ & Dynamical & Eq.~(\ref{appeq:R_def}) \\
         Athermality preservability for GPOs\footnote{GPOs: Gibbs-preserving operations} & $P_\mathrm{T}$ & Dynamical & Eq.~(\ref{appeq:P_T_def})
    \end{tabular}
    \end{ruledtabular}
    \caption{Resource measures considered in this paper.}
    \label{tab:resources}
\end{table}
As in the main body, we use $\mathcal{S}(X)$ to denote the set of quantum states of system $X$, and $\mathcal{O}(X\rightarrow Y)$ to denote the set of quantum operations (completely-positive and trace-preserving maps) between the input system $X$ and the output system $Y$. $X'$ denotes a copy of system $X$ with the same thermal state $\gamma_X$. The dimension of the system $X$ is denoted as $d_{X}$.

The standard forms of the robustness of athermality $R_\mathrm{T}$ for states and channels are given by
\begin{align}
    &R_\mathrm{T}(\rho_A) := \min\left\{s\Big|\frac{\rho_A + s\sigma_A}{1+s} = \gamma_A,\, \sigma_A\in\mathcal{S}(A)\right\}, \label{appeq:ROA_st_std}\\
    &R_\mathrm{T}(\Lambda) := \min\left\{s\Big|\frac{\Lambda + s\,\Omega}{1+s} = G,\,\Omega\in\mathcal{O}(A\rightarrow B),\, G\in\mathrm{GPO}(A\rightarrow B)\right\}, \label{appeq:ROA_ch_std}
\end{align}
where $\mathrm{GPO}(A\rightarrow B) := \left\{G | G(\gamma_A) = \gamma_B, \, G\in\mathcal{O}(A\rightarrow B)\right\}$ is the set of Gibbs-preserving operations (GPOs) from system $A$ to system $B$. We use $R_\mathrm{T}$ for both static and dynamical athermality because dynamical athermality can be reduced to static athermality (Theorem~\ref{appthm:ER_ROA}). The type of athermality will remain unambiguous from the argument of $R_\mathrm{T}$.
Following a similar argument to Section~III.A. in Ref.~\cite{piani2016robustness}, we can rewrite $R_\mathrm{T}$ as
\begin{align}
    &1 + R_\mathrm{T}(\rho_A) := \min\left\{\lambda |\rho_A\le \lambda \gamma_A\right\}, \label{appeq:R_T_st_def}\\
    &1 + R_\mathrm{T}(\Lambda) := \min\left\{\lambda|\Lambda \le \lambda G,\, G\in\mathrm{GPO}(A\rightarrow B)\right\},\label{appeq:R_T_ch_def}
\end{align}
where ``$X\le Y$" means that $Y-X$ is positive semidefinite for states and completely positive (CP) for channels. Throughout, we consider full-rank thermal states $\gamma$, thus avoiding divergence of $R_\mathrm{T}$.

The robustness of signalling $R_\mathrm{S}$ for channels is defined as
\begin{align}
    1 + R_\mathrm{S}(\Lambda) := \min\left\{\lambda | \Lambda \le \lambda\Phi, \, \Phi\in\mathrm{NSO}(A\rightarrow B)\right\},
    \label{appeq:R_S_def}
\end{align}
where $\mathrm{NSO}(A\rightarrow B):=\{\Phi_\sigma\equiv \mathrm{Tr}\{\cdot\}\sigma | \sigma\in\mathcal{S}(B)\}$ is the set of non-signalling operations (NSOs) between systems $A$ and~$B$.

Given the input system $A$ and the output system $B$, the only intersection between $\mathrm{GPO}(A\rightarrow B)$ and $\mathrm{NSO}(A\rightarrow B)$ is the completely thermalising channel: $\Gamma(\cdot) \equiv \mathrm{Tr}\{\cdot\}\gamma_B$. Note, however, that $\Gamma$ is neither an extreme point in NSO nor in GPO. For NSO, this is straightforward. For GPO, see Ref.~\cite{giulio2024extremepoints} for the conditions on extreme points in GPO.
We thus define a joint robustness measure that collectively gauges signalling and athermality of a general channel:
\begin{align}
    1 + R(\Lambda||\Gamma) := \min\{\lambda|\Lambda\le \lambda\Gamma\}.
    \label{appeq:R_def}
\end{align}
The relation among the three measures $R_\mathrm{T}$, $R_\mathrm{S}$ and $R$ is informally represented in Fig.~\ref{fig:NSO_GPO}.
\begin{figure}
    \centering
    \begin{tikzpicture}
        \draw[thick] (0,0) ellipse (1.2cm and 0.75cm);
        \fill[NSOcolor] (0,0) ellipse (1.2cm and 0.75cm);
        \draw[thick] (2.4,0) ellipse (1.2cm and 0.75cm);
        \fill[GPOcolor] (2.4,0) ellipse (1.2cm and 0.75cm);
        \fill[black] (1.2,0) circle (2pt);
        \node[left] at (1.45,-0.5) {$\Gamma$};
        \fill[black] (1.2,2.3) circle (2pt);
        \node[left] at (1.45, 2.6) {$\Lambda$};
        \draw[thick] (1.2,0) -- (1.2,2.3);
        \draw[thick] (1.2,2.3) -- (0.4, {0.707107});
        \node[left] at (0.8,1.7) {$R_\mathrm{S}(\Lambda)$};
        \draw[thick] (1.2,2.3) -- (2, {0.707107});
        \node[right] at (1.6,1.7) {$R_\mathrm{T}(\Lambda)$};
        \node[right] at (1.85,0.4) {GPO};
        \node[left] at (0.5, 0.4) {NSO};
        \draw[->] (1.3,1.1) to[out=30, in=160] (2.6,1.2);
        \node[right] at (2.55,1.15) {$R(\Lambda||\Gamma)$};
    \end{tikzpicture}
    \caption{Geometric sketch of the three resource monotones: $R_\mathrm{T}(\Lambda)$, $R_\mathrm{S}(\Lambda)$ and $R(\Lambda||\Gamma)$.}
    \label{fig:NSO_GPO}
\end{figure}
For $\Phi\in\mathrm{NSO}(A\rightarrow B)$,  $R_\mathrm{S}(\Phi) = 0$ and $R(\Phi||\Gamma) = R_\mathrm{T}(\Phi)$, while for $G\in\mathrm{GPO}(A\rightarrow B)$, $R(G||\Gamma) \ge R_\mathrm{S}(G)$ in general.  For GPOs, the joint measure $R$ reduces to the robustness of athermality preservability mentioned in Ref.~\cite{hsieh2020resource}:
\begin{align}
    P_\mathrm{T}(G) := R(G||\Gamma).
    \label{appeq:P_T_def}
\end{align}

\subsection{Basic properties}\label{sec:basic_properties}
The dual SDP of $R_\mathrm{T}$ for states is given as~\cite{watrous2020lecture2}
\begin{align}
    1 + R_\mathrm{T}(\rho_A) = \max\left\{\mathrm{Tr}\{S\rho_A\}\Big| S\ge 0, \, \mathrm{Tr}\{S\gamma_A\} = 1\right\},
    \label{appeq:R_T_st_dual}
\end{align}
for which strong duality holds, i.e.,~primal and dual SDPs yield the same optimal values~\cite{skrzypczyk2023semidefinite}.
As a resource monotone, $R_\mathrm{T}$ satisfies desired properties. Here, we show them explicitly, by following similar proofs as Ref.~\cite{diaz2018using}, where they are for the robustness of coherence.
\begin{enumerate}
    \item (Faithfulness) $\RT(\rho) = 0 \Leftrightarrow \rho = \gamma$. \\\textit{Proof.} Direct result from Eq.~\eqref{appeq:R_T_st_def}. \qedsymbol{}
    \item (Monotonicity) $\RT(\rho_A)\ge \RT[G(\rho_A)],\, \forall\, G\in \mathrm{GPO}(A\rightarrow B)$. \\
    \textit{Proof.} Let $\lambda^*\equiv1+\RT(\rho_A)$ be the optimal value that achieves the minimum. Since $G\in \mathrm{GPO}(A\rightarrow B)$ is completely positive, $\lambda^*\gamma_A - \rho_A \ge 0 \Rightarrow G(\lambda^*\gamma_A-\rho_A) = \lambda^* \gamma_B - G(\rho_A)\ge 0$. Consequently, $\lambda^*$ is in the feasible set of the minimisation defining $\RT[G(\rho_A)]$, such that $\lambda^*\ge 1+\RT[G(\rho_A)]$. Hence, $\RT(\rho_A)\ge \RT[G(\rho_A)]$. \qedsymbol{} 
    \item (Convexity) $\RT(\sum_i p_i\rho_i) \le \sum_i p_i \RT(\rho_i)$, for $\sum_i p_i = 1$ and $p_i\in [0,1], \forall\, i$.\\
    \textit{Proof.} Let $\lambda^*_i \equiv 1+\RT(\rho_i)$ be the optimal value for each $i$. $\lambda^*_i\gamma-\rho_i \ge 0, \forall\, i \Rightarrow \sum_i p_i\lambda^*_i\gamma - \sum_i p_i\rho_i \ge 0$. Consequently, $\sum_ip_i\lambda^*_i$ is a feasible for the minimisation defining $\RT(\sum_ip_i\rho_i)$, such that $\sum_ip_i\lambda^*_i\ge 1+\RT(\sum_ip_i\rho_i)$. Hence, $\RT(\sum_i p_i\rho_i) \le \sum_i p_i \RT(\rho_i)$. \qedsymbol{}
    \item (Multiplicity) $1+\RT(\rho_A\otimes\rho_B) = [1+\RT(\rho_A)][1+\RT(\rho_B)]$. \\
    \textit{Proof.} Let $\lambda_X^*$ be the optimal value such that $1+\RT(\rho_X)\equiv \lambda_X$, where $X = A, B$. $\lambda^*_X\gamma_X-\rho_X\ge 0, \forall\, X \Rightarrow \lambda_A^*\lambda_B^*\gamma_A\otimes\gamma_B \ge \rho_A\otimes\rho_B$. $\lambda_A^*\lambda_B^*$ is  feasible for the minimisation defining $\RT(\rho_A\otimes\rho_B)$, such that $\lambda_A^*\lambda_B^* = [1+\RT(\rho_A)][1+\RT(\rho_B)]\ge 1+\RT(\rho_A\otimes\rho_B)$. Let $S_X^*$ be the optimal operator of the dual SDP [Eq.~(\ref{appeq:R_T_st_dual})] for $X=A,B$, such that $1+\RT(\rho_X) = \mathrm{Tr}\{S^*_X\rho_X\}$. Now $S\equiv S^*_A\otimes S^*_B$ is feasible for the dual SDP of $\RT(\rho_A\otimes\rho_B)$, so that  $\mathrm{Tr}\{S\rho_A\otimes\rho_B\} = \mathrm{Tr}\{S^*_A\rho_A\}\mathrm{Tr}\{S^*_B\rho_B\} = [1+\RT(\rho_A)][1+\RT(\rho_B)] \le 1+\RT(\rho_A\otimes\rho_B)$. Hence, $1+\RT(\rho_A\otimes\rho_B) = [1+\RT(\rho_A)][1+\RT(\rho_B)]$. \qedsymbol{}
\end{enumerate}
By Theorem~\ref{appthm:ER_ROA} (See below), $R_\mathrm{T}(\Lambda) = R_\mathrm{T}[\Lambda(\gamma_A)],\,\forall\,\Lambda\in\mathcal{O}(A\rightarrow B)$. Thus, $R_\mathrm{T}$, as a measure of dynamical resource, satisfies the following properties:
\begin{enumerate}
    \item (Faithfulness) $\RT(\Lambda) = 0 \Leftrightarrow \Lambda\text{ is a GPO}$. 
    \item (Monotonicity) $\RT(\Lambda)\ge \RT(G\circ\Lambda\circ G'),\, \forall\, G\text{ and }G'\text{ are GPOs}$.  
    \item (Convexity) $\RT(\sum_i p_i\Lambda_i) \le \sum_i p_i \RT(\Lambda_i)$, for $\sum_i p_i = 1$ and $p_i\in [0,1], \forall\, i$.
    \item (Multiplicity) $1+\RT(\Lambda\otimes\Omega) = [1+\RT(\Lambda)][1+\RT(\Omega)]$. 
\end{enumerate}
In addition, the joint measure $R$ defined in Eq.~(\ref{appeq:R_def}) can be written in terms of the Choi states of channels~\cite{wilde_2013}:
\begin{align}
    1 + R(\Lambda||\Gamma) := \min \{\lambda | J_\Lambda \le \lambda J_\Gamma \},
\end{align}
where $J_X\equiv \mathcal{I}\otimes X(\ket{\phi^+}\!\bra{\phi^+})$ with $\ket{\phi^+}\equiv \sum_{i=0}^{d_A-1}\ket{i}\otimes\ket{i}$, is defined as the Choi state of the channel $X\in\mathcal{O}(A\rightarrow B)$. Comparing with Eq.~(\ref{appeq:R_T_st_def}), $R$ can be regarded as the $R_\mathrm{T}$ of $J_\Lambda$ with respect to the ``thermal state" $J_\Gamma$. Therefore, $R$, as well as $P_\mathrm{T}$ defined in Eq.~(\ref{appeq:P_T_def}), inherits the same properties as $R_\mathrm{T}$:
\begin{enumerate}
    \item (Faithfulness) $R(\Lambda||\Gamma) = 0 \Leftrightarrow \Lambda = \Gamma$. 
    \item (Monotonicity) $R(\Lambda||\Gamma)\ge R(G\circ\Lambda\circ G'||\Gamma),\, \forall\, G\text{ and }G'\text{ are GPOs}$.  
    \item (Convexity) $R(\sum_i p_i\Lambda_i||\Gamma) \le \sum_i p_i R(\Lambda_i||\Gamma)$, for $\sum_i p_i = 1$ and $p_i\in [0,1], \forall\, i$.
    \item (Multiplicity) $1+R(\Lambda\otimes\Omega||\Gamma) = [1+R(\Lambda||\Gamma)][1+R(\Omega||\Gamma)]$. 
\end{enumerate}
To show more properties of $R_\mathrm{T}$, we prove the following lemma first.
\begin{lemmaApp}\label{lem:S_rank_1}
    The optimal $S^*$ that maximises the dual SDP of the $R_\mathrm{T}$ of states [Eq.~(\ref{appeq:R_T_st_dual})] can always be chosen to be rank-1.
\end{lemmaApp}
\begin{proof}
    Since the Gibbs state $\gamma$ is full-rank by assumption, we can define a bijective map, $X\equiv \gamma^{1/2} S \gamma^{1/2}$, such that Eq.~(\ref{appeq:R_T_st_dual}) can be rewritten as
    \begin{equation}
        1 + \RT(\rho) = \max_{X}\left\{\text{Tr}\{\tilde{\rho}X\}\Big| X\ge 0,\, \text{Tr}\{X\}=1\right\},
        \label{appeq:R_T_st_dual_X}
    \end{equation}
    where we denote $\tilde{\rho} \equiv \gamma^{-1/2}\rho\gamma^{-1/2}$.
    The feasible set $\{X| X\ge 0, \,\text{Tr}\{X\}=1\}$ is the set of quantum states, meaning that any $X$ in the set can be written as a convex combination of rank-1 states. Let $X^*=\sum_ip_i\ket{\phi_i}\!\bra{\phi_i}$ for $\sum_i p_i = 1$ with $p_i\in [0,1], \forall\, i$, be the optimal operator, such that 
    \begin{align}
        1+\RT(\rho) = \text{Tr}\{\tilde{\rho}X^*\} = \sum_i p_i \text{Tr}\{\tilde{\rho}\ket{\phi_i}\!\bra{\phi_i}\}\le \max_{i}\text{Tr}\{\tilde{\rho}\ket{\phi_i}\!\bra{\phi_i}\}.
    \end{align}
    However, by Eq.~(\ref{appeq:R_T_st_dual_X}), $1+\RT(\rho)\ge\max_{i}\text{Tr}\{\tilde{\rho}\ket{\phi_i}\!\bra{\phi_i}\}$. We thus have
    $1+\RT(\rho)=\max_{i}\text{Tr}\{\tilde{\rho}\ket{\phi_i}\!\bra{\phi_i}\}$ and $X^*$ can be chosen to be rank-1. Note that $X^*$ does not have to be rank-1 since for some $\tilde{\rho}$, $\text{Tr}\{\tilde{\rho}\ket{\phi_i}\!\bra{\phi_i}\}$ can be independent of $i$, i.e.,~any convex combination of $\{\ket{\phi_i}\!\bra{\phi_i}\}_i$ is an optimal $X^*$.
    
    The optimal $S^*$ in Eq.~(\ref{appeq:R_T_st_dual}) is given by $S^*\equiv \gamma^{-1/2}X^*\gamma^{-1/2}$. 
    Since $X^*$ can always be chosen to be rank-1, such that $X^* \equiv \ket{\phi}\!\bra{\phi}$, $S^* \equiv \gamma^{-1/2}X^*\gamma^{-1/2} = \gamma^{-1/2}\ket{\phi}\!\bra{\phi}\gamma^{-1/2}$ is also rank-1.
\end{proof}
We now find the state with highest $R_\mathrm{T}$, i.e.,~the most athermal state.
\begin{theoremApp}[The most athermal state]\label{appthm:most_athermal_state}
    Given a non-degenerate, full-rank Gibbs state $\gamma\equiv\sum_{i=0}^{d-1}g_i\ket{i}\!\bra{i}$ with the populations $g_{d-1}<g_{d-2}<\dots<g_0$, the unique state with the maximal $R_\mathrm{T}$ is the state:
    $\rho^* \equiv \ket{d-1}\!\bra{d-1}$, with $1+\RT(\rho^*) = g_{d-1}^{-1}.$
\end{theoremApp}
\begin{proof}
    We define $\lambda(\rho)\equiv 1+\RT(\rho)$. From the dual form of $1+\RT(\rho)$ [Eq.~(\ref{appeq:R_T_st_dual})] and Lemma~\ref{lem:S_rank_1}, we have
    \begin{equation}
        \lambda(\rho) = \max_{\ket{\phi}}\left\{\text{Tr}\{\rho S\}\Big|S = \frac{\ket{\phi}\!\bra{\phi}}{\text{Tr}\{\gamma\ket{\phi}\!\bra{\phi}\}}\right\} = \max_{\ket{\phi}}\frac{\bra{\phi}\rho\ket{\phi}}{\bra{\phi}\gamma\ket{\phi}}.
    \end{equation}
    We denote the value of $\lambda$ corresponding to the maximal $1+R_\mathrm{T}$ as $\lambda_{\max}$, such that
    \begin{equation}
        \lambda_{\max}\equiv\max_{\rho}\lambda(\rho)=\max_{\rho}\max_{\ket{\phi}}\frac{\bra{\phi}\rho\ket{\phi}}{\bra{\phi}\gamma\ket{\phi}}\le\frac{\max_{\rho}\max_{\ket{\phi}}\bra{\phi}\rho\ket{\phi}}{\min_{\ket{\phi}}\bra{\phi}\gamma\ket{\phi}}. 
    \end{equation}
    It is clear that $\max_{\rho}\max_{\ket{\phi}}\bra{\phi}\rho\ket{\phi}=1$. For the denominator, since $\bra{\phi}\gamma-g_{d-1} \mathbbm{1}\ket{\phi}\ge 0, \forall \ket{\phi}$ and the equality holds as $\ket{\phi} = \ket{d-1}$, we have $\min_{\ket{\phi}}\bra{\phi}\gamma\ket{\phi} = g_{d-1}$. Therefore, 
    \begin{equation}
        \lambda_{\max}=\max_{\rho}\max_{\ket{\phi}}\frac{\bra{\phi}\rho\ket{\phi}}{\bra{\phi}\gamma\ket{\phi}}\le \frac{1}{g_{d-1}}.
    \end{equation}
    Equality holds for $\rho = \rho^* \equiv \ket{d-1}\!\bra{d-1}$ and $\ket{\phi} = \ket{d-1}$, i.e., $\rho^*$ is the most athermal state with $1+R_\mathrm{T}(\rho^*) = \lambda_{\max} = g_{d-1}^{-1}$.

    Now we prove $\rho = \ket{d-1}\!\bra{d-1}$ is the only most athermal state. 
    Due to the convexity of $R_\mathrm{T}$, the most athermal state can only be mixed if it is a convex combination of pure states with the same $R_\mathrm{T}$. Therefore, to prove that $\rho$ is unique, it is enough to assume that there is another pure state $\sigma \equiv \ket{\psi}\!\bra{\psi}$ with $\lambda(\sigma) = g_{d-1}^{-1}$ and to show that $\sigma$ does not exist. Since $\sigma\neq \rho$, $\ket{\psi}\neq\ket{d-1}$.
    Let $\ket{\phi'}$ be the optimal state that gives
    \begin{equation}
        \frac{\bra{\phi'}\sigma\ket{\phi'}}{\bra{\phi'}\gamma\ket{\phi'}} = \frac{|\langle \phi'|\psi\rangle|^2}{\bra{\phi'}\gamma\ket{\phi'}} = \frac{1}{g_{d-1}}.
        \label{appeq:1/g_n_equality}
    \end{equation}
    We note that $\ket{\phi'}\neq \ket{d-1}$, since $\dfrac{|\langle d-1|\psi\rangle|^2}{\bra{d-1}\gamma\ket{d-1}} = \dfrac{|\langle d-1|\psi\rangle|^2}{g_{d-1}} < \dfrac{1}{g_{d-1}}$. However, since $\ket{\phi'}\neq\ket{d-1}$, we have $\bra{\phi'}\gamma\ket{\phi'}> g_{d-1}$ ($\ket{d-1}$ is the only state satisfying $\bra{d-1}\gamma\ket{d-1} = g_{d-1}$
    because $\gamma$ is non-degenerate). Eq.~(\ref{appeq:1/g_n_equality}) therefore requires $|\langle \phi'|\psi\rangle|>1$. Because such a state $\sigma$ cannot exist, $\rho = \ket{d-1}\!\bra{d-1}$ is the only state with the maximal $R_\mathrm{T}$, for which  $1+\RT(\rho) = g_{d-1}^{-1}$.
\end{proof}
We remark that when the Gibbs state $\gamma$ is defined at positive temperatures, the most athermal state $\rho^*\equiv\ketbra{d-1}{d-1}$ is the most energetic state corresponding to the highest energy level. However, some systems can equilibrate at negative temperatures (see, e.~g., Ref.~\cite{Baudin2023NegativeTemperature}), in which the most athermal state $\rho^*$ will be the ground state.

We further show that $R_\mathrm{T}$ can be related with commonly used distances, which is given by a continuity property. In the context of resource theories, the continuity of resource monotone refers to the property of the monotone guaranteeing that any two states that are close in some chosen distance are close in terms of their resources according to this monotone.
The continuity of robustness measures has been studied in Ref.~\cite{schluck2023continuity} and the continuity of $R_\mathrm{T}$ has been explicitly proven in Lemma~F.3 in Ref.~\cite{hsieh2020resource}. Here, we provide a self-contained proof for the continuity of $R_\mathrm{T}$, which improves the bound of Ref.~\cite{hsieh2020resource} by a factor $1/2$.
\begin{theoremApp}[Continuity of $R_\mathrm{T}$]\label{lem:ROA_continuity}
    For any two states $\rho$, $\sigma$ and a given Gibbs state $\gamma$, the following inequality holds:
    \begin{equation}
        \big|\RT(\rho) - \RT(\sigma)\big|\le \frac{1}{2 g_{\min}}||\rho-\sigma||_1,
    \end{equation}
    where $||A||_1 := \mathrm{Tr}\{\sqrt{A^\dagger A}\}$ is the trace norm and $g_{\min}$ is the minimal population of energy eigenbasis of $\gamma$.
\end{theoremApp}
\begin{proof}
    Consider the dual SDP of $R_\mathrm{T}$ of states [Eq.~(\ref{appeq:R_T_st_dual})]. We denote the optimal operators as $S^*_\rho$ and $S^*_\sigma$, such that
    \begin{align}
        \mathrm{Tr}\{\rho S^*_\rho\} &= 1 + \RT(\rho), \label{appeq:ROA_rho}\\
        \mathrm{Tr}\{\sigma S^*_\sigma\} &= 1 + \RT(\sigma).\label{appeq:ROA_sigma}
    \end{align}
    Without loss of generality, we assume that $\RT(\rho)\ge \RT(\sigma)$, i.e.,~$\mathrm{Tr}\{\rho S^*_\rho\}\ge \mathrm{Tr}\{\sigma S^*_\sigma\}$.
    We then have
    \begin{align}
        \mathrm{Tr}\{\rho S^*_\rho\} - \mathrm{Tr}\{\sigma S^*_\sigma\} \le \mathrm{Tr}\{(\rho-\sigma)S^*_\rho\}.
    \end{align}
    By assumption, both sides of the above inequality are non-negative. We can take the absolute values of both sides and obtain
    \begin{align}
        \big|\mathrm{Tr}\{\rho S^*_\rho\} - \mathrm{Tr}\{\sigma S^*_\sigma\} \big| &\le \big| \mathrm{Tr}\{(\rho-\sigma)S^*_\rho\} \big| \nonumber\\
        &= \big| \bra{\psi_\rho}(\rho-\sigma)\ket{\psi_\rho}\big | \cdot ||S^*_\rho||_\infty \nonumber \\
        &\le \big| \bra{\psi_\rho}(\rho-\sigma)\ket{\psi_\rho}\big| \cdot g_{\min}^{-1} \nonumber \\
        &\le \frac{1}{2}||\rho-\sigma||_1 \cdot g_{\min}^{-1}.
    \end{align}
    Here, in the second line, we used that, according to Lemma~\ref{lem:S_rank_1}, $S^*_\rho$ can be chosen as rank-1, i.e.,~$S^*_\rho\equiv ||S^*_\rho||_\infty \ket{\psi_\rho}\!\bra{\psi_\rho}$, where $\ket{\psi_\rho}\!\bra{\psi_\rho}$ is a (normalized) state and $||X||_\infty$ is the operator norm of $X$ that equals to the largest eigenvalue of $X$. In the third line, according to the dual SDP of $R_\mathrm{T}$ [Eq.~(\ref{appeq:R_T_st_dual})], $\mathrm{Tr}\{\gamma S^*_\rho\} = ||S^*_\rho||_\infty \bra{\psi_\rho}\gamma\ket{\psi_\rho} = 1$. Since $\bra{\psi_\rho}\gamma\ket{\psi_\rho}\ge g_{\min}$, we have $||S^*_\rho||_\infty \le g_{\min}^{-1}$. In the forth line, we used the fact that
    \begin{align}
        \big| \bra{\psi_\rho}(\rho-\sigma)\ket{\psi_\rho}\big| \le \max_k |\mu_k| \le \sum_{k\in\{\ell|\mu_\ell\ge 0\}} \mu_k = \sum_{k\in\{\ell|\mu_\ell\le 0\}} |\mu_k| =\frac{1}{2}||\rho-\sigma||_1,
    \end{align}
    where $\mu_k$ is the $k$th eigenvalue of $(\rho-\sigma)$ satisfying $\sum_k \mu_k = 0$ due to $\mathrm{Tr}\{\rho-\sigma\} = 0$.
    Referring to Eqs.~(\ref{appeq:ROA_rho}) and (\ref{appeq:ROA_sigma}), we obtain $|\RT(\rho) - \RT(\sigma)|\le (g_{\min}^{-1}/2)||\rho-\sigma||_1$.
\end{proof}

\subsection{Proof of Theorem~\ref{thm:ER_ROA}}\label{app:proof_Thm_ER_ROA}
In this subsection, we prove Theorem~\ref{thm:ER_ROA}, which states that the robustness of athermality $R_\mathrm{T}$ of a general quantum channel $\Lambda$ is equal to the robustness of athermality $R_\mathrm{T}[\Lambda(\gamma)]$ of its output state with the input state being the Gibbs state. In other words, the dynamical athermality of a channel is equivalent to its static athermality generating power:
\begin{theoremApp}[{Equivalence between dynamical and static $R_\mathrm{T}$ (Theorem~\ref{thm:ER_ROA})}] \label{appthm:ER_ROA}
For a general channel $\Lambda \in \mathcal{O}(A\rightarrow B)$, the following equality holds:
\begin{align}
    R_\mathrm{T}(\Lambda) = R_\mathrm{T}[\Lambda(\gamma_A)].
    \label{appeq:ER_ROA}
\end{align}
\end{theoremApp}
\begin{proof}
    We firstly show that $R_\mathrm{T}(\Lambda) \ge R_\mathrm{T}[\Lambda(\gamma_A)]$. Let $\lambda^* \equiv 1+R_\mathrm{T}(\Lambda)$ and $G^*\in\mathrm{GPO}(A\rightarrow B)$ be the optimal variable and GPO in the definition of $R_\mathrm{T}(\Lambda)$ [Eq.~(\ref{appeq:R_T_ch_def})], respectively. We have
    \begin{align}
        \lambda^* G^* - \Lambda \ge 0 \Rightarrow (\lambda^* G^* - \Lambda)(\gamma_A) = \lambda^* \gamma_B - \Lambda(\gamma_A) \ge 0.
    \end{align}
    Therefore, $\lambda^*$ is a feasible variable for the SDP of $R_\mathrm{T}[\Lambda(\gamma_A)]$ in Eq.~(\ref{appeq:R_T_st_def}), which gives us $1+R_\mathrm{T}(\Lambda) \equiv \lambda^* \ge 1 + R_\mathrm{T}[\Lambda(\gamma_A)]$, i.e.,~$R_\mathrm{T}(\Lambda) \ge R_\mathrm{T}[\Lambda(\gamma_A)]$.

    To prove that $R_\mathrm{T}(\Lambda) \le R_\mathrm{T}[\Lambda(\gamma_A)]$, we consider the standard form of $R_\mathrm{T}[\Lambda(\gamma_A)]$ [Eq.~(\ref{appeq:ROA_st_std})] and denote the optimal state as $\sigma_B^*$ which satisfies
    \begin{align}
        \frac{1}{1+R_\mathrm{T}[\Lambda(\gamma_A)]}\Lambda(\gamma_A) + \frac{R_\mathrm{T}[\Lambda(\gamma_A)]}{1+R_\mathrm{T}[\Lambda(\gamma_A)]}\sigma_B^* = \gamma_B.
        \label{appeq:cvx_mixture_Lambda(gamma)_sigma}
    \end{align}
    Define the NSO $\Phi_{\sigma_B^*} (\cdot) := \mathrm{Tr}\{\cdot\}\sigma_B^*\in \mathrm{NSO}(A\rightarrow B)$. By Eq.~(\ref{appeq:cvx_mixture_Lambda(gamma)_sigma}), we have
    \begin{align}
        \frac{1}{1+R_\mathrm{T}[\Lambda(\gamma_A)]}\Lambda + \frac{R_\mathrm{T}[\Lambda(\gamma_A)]}{1+R_\mathrm{T}[\Lambda(\gamma_A)]}\Phi_{\sigma_B^*} \in \mathrm{GPO}(A\rightarrow B).
    \end{align}
    According to the standard form of $R_\mathrm{T}(\Lambda)$ in Eq.~(\ref{appeq:ROA_ch_std}), we obtain $R_\mathrm{T}(\Lambda) \le R_\mathrm{T}[\Lambda(\gamma_A)]$. Hence, Eq.~(\ref{appeq:ER_ROA}) holds.
\end{proof}

\subsection{Proof of Theorem~\ref{thm:ER_R}}\label{app:proof_Thm_ER_R}
In this subsection, we prove Theorem~\ref{thm:ER_R}, which states that the joint resource measure $R$ of a general quantum channel is equal to the $R_\mathrm{T}$ of its output state when it is locally applied on the thermofield double state $\tilde{\gamma}_{AA'}$~\cite{Takahashi1996ThermoField} defined as $\tilde{\gamma}_{AA'}\equiv \ket{\tilde{\gamma}}\!\bra{\tilde{\gamma}}_{AA'}$ with $\ket{\tilde{\gamma}}_{AA'} \equiv \sum_{i}\sqrt{g_i}\ket{i}_{A}\otimes\ket{i}_{A'}$
where $g_i$ is the population of $\gamma_A$ of the $i$th energy eigenstate $\ket{i}_A$, satisfying $\mathrm{Tr}_{\{A,A'\}\text{\textbackslash} X}\{\tilde{\gamma}_{AA'}\} = \gamma_{X}$ for $X\in\{A, A'\}$. Specifically, we have the following more general theorem:
\begin{theoremApp}[Equivalence between the dynamical resource $R$ and static $R_\mathrm{T}$] \label{appthm:ER_R}
    For a general channel $\Lambda\in\mathcal{O}(A'\rightarrow B)$, the following equality holds:
    \begin{align}
        R(\Lambda||\Gamma) = R_\mathrm{T}[\mathcal{I}\otimes\Lambda(\tau_{AA'})],
        \label{appeq:ER_R}
    \end{align}
    where $A'$ is a copy of system $A$ and $\tau_{AA'}$ is a general purification of $\gamma_A$ on system $AA'$ satisfying $\mathrm{Tr}_{A'}\{\tau_{AA'}\} = \gamma_A$.
\end{theoremApp}
\begin{proof}
    Consider the Choi operator $J_\Lambda$
    and note that the Choi operator of $\Gamma$ is $J_\Gamma \equiv \mathbbm{1}_A\otimes\gamma_B$. We have 
    \begin{align}
        J_\Lambda \le \lambda J_\Gamma &\Leftrightarrow (\gamma_A^{1/2}U_A\otimes \mathbbm{1}_B)J_\Lambda (U_A^\dagger\gamma_A^{1/2}\otimes \mathbbm{1}_B) \le \lambda (\gamma_A^{1/2}U_A\otimes \mathbbm{1}_B)J_\Gamma(U_A^\dagger\gamma_A^{1/2}\otimes \mathbbm{1}_B)\\
        &\Leftrightarrow \mathcal{I}_A\otimes\Lambda[(\gamma_A^{1/2}U_A\otimes\mathbbm{1}_{A'})\ket{\phi^+}\!\bra{\phi^+}_{AA'}(U_A^\dagger\gamma_A^{1/2}\otimes\mathbbm{1}_{A'})] \le \lambda \gamma_A\otimes\gamma_B\\
        &\Leftrightarrow \mathcal{I}_A\otimes\Lambda[\mathcal{I}_A\otimes \mathcal{U}_{A'}(\tilde{\gamma}_{AA'})]\le \lambda \gamma_A\otimes\gamma_B\\
        &\Leftrightarrow \mathcal{I}_A\otimes\Lambda(\tau_{AA'}) \le \lambda \gamma_A\otimes\gamma_B, 
    \end{align}
    where $U_A$ is a general unitary operator on system $A$, $\tilde{\gamma}_{AA'} \equiv \ket{\tilde{\gamma}}\!\bra{\tilde{\gamma}}_{AA'}$ with $\ket{\tilde{\gamma}}_{AA'}=(\gamma_A^{1/2}\otimes\mathbbm{1}_{A'})\ket{\phi^+}_{AA'} = (\mathbbm{1}_A\otimes\gamma_{A'}^{1/2})\ket{\phi^+}_{AA'}$ is the thermofield double state, $\mathcal{U}_{A'}(\cdot) \equiv U_{A'}^T(\cdot)(U_{A'}^T)^\dagger$ is a unitary channel acting on system $A'$ and $\tau_{AA'}\equiv\mathcal{I}_A\otimes\mathcal{U}_{A'}(\tilde{\gamma}_{AA'})$.
    The third equivalence relation is from the fact that
    $(X_A\otimes\mathbbm{1}_{A'})\ket{\phi^+}_{AA'} = (\mathbbm{1}_A\otimes X_{A'}^T)\ket{\phi^+}_{AA'}$,
    for any operator $X$ on system $A$ (or $A'$).
    Since any purification of $\gamma_A$ on system $AA'$ can be written as $\mathcal{I}_A\otimes \mathcal{V}_{A'}(\tilde{\gamma}_{AA'})$ with a unitary channel $\mathcal{V}_{A'}$~\cite{kirkpatrick2006schrodinger}, $\tau_{AA'}$ is a general purification of $\gamma_A$ satisfying $\mathrm{Tr}_{A'}\{\tau_{AA'}\} = \gamma_A$.
    According to the definitions of $R$ [Eq.~(\ref{appeq:R_def})] and $R_\mathrm{T}$ of states [Eq.~(\ref{appeq:R_T_st_def})], Eq.~(\ref{appeq:ER_R}) holds.
\end{proof}
Theorem~\ref{thm:ER_R} directly follows from Theorem~\ref{appthm:ER_R} by setting $\tau_{AA'} = \tilde{\gamma}_{AA'}$.
Furthermore, two important implications follow from Theorem~\ref{appthm:ER_R}:
\begin{corollaryApp}\label{appcorol:all_purfications_R_T}
    When systems $A$ and $B$ have the same dimension, all purifications of $\gamma_A$ on system $AA'$ have the same athermality with respect to the thermal state $\gamma_{AB}\equiv\gamma_A\otimes\gamma_B$, which is equal to $R(\mathcal{I}||\Gamma)$.
\end{corollaryApp}
\begin{proof}
    When systems $A$ and $B$ have the same dimension, the identity map $\mathcal{I}\in\mathcal{O}(A'\rightarrow B)$. The statement directly follows from Theorem~\ref{appthm:ER_R} by setting $\Lambda = \mathcal{I}$.
\end{proof}
\begin{corollaryApp}\label{appcorol:all_unitaries_R}
    When systems $A$ and $B$ have the same dimension, all unitary channels from $A'$ to $B$ have the same $R$ equal to $R(\mathcal{I}||\Gamma)$.
\end{corollaryApp}
\begin{proof}
    By Theorem~\ref{appthm:ER_R}, $R(\mathcal{U}||\Gamma) = R_\mathrm{T}[\mathcal{I}\otimes\mathcal{U}(\tau_{AA'})]$, for all unitary channels $\mathcal{U}\in\mathcal{O}(A'\rightarrow B)$. Since $\mathcal{I}\otimes\mathcal{U}(\tau_{AA'})$ is a purification of $\gamma_A$ on $AA'$, $R_\mathrm{T}[\mathcal{I}\otimes\mathcal{U}(\tau_{AA'})] = R(\mathcal{I}||\Gamma)$ due to Corollary~\ref{appcorol:all_purfications_R_T}. Thus, $R(\mathcal{U}||\Gamma) = R(\mathcal{I}||\Gamma)$, for all unitary channels $\mathcal{U}\in\mathcal{O}(A'\rightarrow B)$.
\end{proof}

\subsection{The most resourceful channels measured by \texorpdfstring{$R$}{R} (Proof of Theorem~\ref{thm:upper_bound_R})}\label{app:proof_Thm_upper_bound_R}
In this section, we find the most resourceful channels according to the measure $R$, which will have the largest combined resource of signalling and athermality. Before we present the result, we prove the following lemma first:
\begin{lemmaApp}\label{lem:barely_positive}
    Given two operators $X$ and $Y$ where $Y$ is positive, the following statements are equivalent:
    \begin{enumerate}
        \item $\min\{f|X\le fY\} = 1$.
        \item $\exists \ket{\psi},\, \bra{\psi}Y - X\ket{\psi} = 0$ and $Y - X \ge 0$.
    \end{enumerate}
\end{lemmaApp}
\begin{proof}
    For the direction ($1\Rightarrow 2$), it is straightforward that $\min\{f|X\le fY\} = 1 \Rightarrow Y - X \ge 0$. To show the existence of such a state $\ket{\psi}$, we assume that $\forall \ket{\psi}, \bra{\psi}Y - X\ket{\psi} > 0$. Define $\ket{\psi_{\min}} \equiv \arg\min_{\ket{\psi}} \bra{\psi}Y - X\ket{\psi}$ and $\ket{\psi_{\max}} \equiv \arg\max_{\ket{\psi}} \bra{\psi}Y\ket{\psi}$. Then, $\forall \ket{\psi}$,
    \begin{align}
    \bra{\psi}\left[\left(1-\frac{\bra{\psi_{\min}}Y-X\ket{\psi_{\min}}}{\bra{\psi_{\max}}Y\ket{\psi_{\max}}}\right)Y-X\right]\ket{\psi} & = \bra{\psi}Y-X\ket{\psi} - \frac{\bra{\psi}Y\ket{\psi}}{\bra{\psi_{\max}}Y\ket{\psi_{\max}}}\bra{\psi_{\min}}Y-X\ket{\psi_{\min}}\\& \ge \bra{\psi}Y-X\ket{\psi} -  \bra{\psi_{\min}}Y-X\ket{\psi_{\min}} \ge 0.
    \end{align}
    Therefore,
    \begin{align}
        \min\{f|X \le fY\} \le 1 - \frac{\bra{\psi_{\min}}Y-X\ket{\psi_{\min}}}{\bra{\psi_{\max}}Y\ket{\psi_{\max}}} < 1,
    \end{align}
    which contradicts $\min\{f|X\le fY\} = 1$. The direction ($1 \Rightarrow 2$) is proven.

    For the direction ($2 \Rightarrow 1$), since $Y - X\ge 0$, we have $\min\{f|X \le fY\} \le 1$. Assume that $f^* \equiv \min\{f|X \le fY\} < 1$. For the state $\ket{\psi}$ satisfying $\bra{\psi}Y-X\ket{\psi} = 0$, we have
    \begin{align}
        \bra{\psi}f^*Y-X\ket{\psi} &= \bra{\psi}[(f^*-1)Y+Y-X]\ket{\psi} = (f^*-1)\bra{\psi}Y\ket{\psi} < 0,
    \end{align}
    because $\bra{\psi}Y\ket{\psi} > 0$ for a positive $Y$. This contradicts the assumption that $\bra{\psi}f^*Y -X\ket{\psi} \ge 0,\,\forall \ket{\psi}$. Therefore, $f^*$ must be 1. The direction ($2 \Rightarrow 1$) is proven.
\end{proof}
\begin{theoremApp}[{Channels with the highest $R$}]\label{appthm:upper_bound_R}
    For a general channel $\Lambda\in\mathcal{O}(A\rightarrow B)$, the follow inequality holds:
    \begin{align}
        R(\Lambda||\Gamma) \le \mathrm{Tr}\{\gamma_B^{-1}\} - 1,
    \end{align}
    where the equality holds when $d_A = d_B$ and $\Lambda$ is unitary.
\end{theoremApp}
\begin{proof}
    The idea of the proof is to firstly show that when $d_A = d_B \equiv d$, for any unitary channel $\mathcal{U} \in\mathcal{O}(A\rightarrow B)$, $R(\mathcal{U}||\Gamma) = \mathrm{Tr}\{\gamma_B^{-1}\} - 1$, and then to show that in general $R(\Lambda||\Gamma) \le \mathrm{Tr}\{\gamma_B^{-1}\} - 1, \,\forall\, \Lambda\in\mathcal{O}(A\rightarrow B)$.
    
    By Corollary~\ref{appcorol:all_unitaries_R}, we have $R(\mathcal{U}||\Gamma) = R(\mathcal{I}||\Gamma)$, for all unitary channels $\mathcal{U}\in\mathcal{O}(A'\rightarrow B)$. Thus, for the first step, when $d_A = d_B \equiv d$, we choose the specific unitary channel $\mathcal{V}$ such that $\mathcal{I}_A\otimes\mathcal{V}(\ket{\phi^+}\!\bra{\phi^+}_{AA'}) = \ket{\phi^+}\!\bra{\phi^+}_{AB}\equiv \sum_{i,j=0}^{d-1}\ket{i}\!\bra{j}_A\otimes\ket{i}\!\bra{j}_{B}$ where $\{\ket{k}_A\}_k$ and $\{\ket{k}_B\}_k$ are eigenbases of $\gamma_A$ and $\gamma_B$, respectively. We aim to show that $R(\mathcal{V}||\Gamma) = \mathrm{Tr}\{\gamma_B^{-1}\} - 1$. According to the definition of $R$ [Eq.~(\ref{appeq:R_def})], it means that
    \begin{align}
        \min \left\{\lambda \Big| \frac{1}{\mathrm{Tr}\{\gamma_B^{-1}\}}J_\mathcal{V} \le \lambda J_\Gamma \right\} = 
        \min \left\{\lambda \Big| \frac{1}{\mathrm{Tr}\{\gamma_B^{-1}\}}\ket{\phi^+}\!\bra{\phi^+}_{AB} \le \lambda\mathbbm{1}_A\otimes\gamma_B\right\} = 1,
    \end{align}
    By Lemma~\ref{lem:barely_positive}, it is equivalent to show that
    \begin{align*}
        \exists \ket{\psi},\, \bra{\psi}\left(\mathbbm{1}_A\otimes\gamma_B - \frac{1}{\mathrm{Tr}\{\gamma_B^{-1}\}}\ket{\phi^+}\!\bra{\phi^+}_{AB} \right)\ket{\psi} = 0 \text{ and } \mathbbm{1}_A\otimes\gamma_B - \frac{1}{\mathrm{Tr}\{\gamma_B^{-1}\}}\ket{\phi^+}\!\bra{\phi^+}_{AB} \ge 0.
    \end{align*}
    Consider the (unnormalised) state $\ket{\gamma_B^{-1}}\equiv \sum_{i=0}^{d-1} 1/g^{(B)}_i\ket{i}_A\otimes\ket{i}_B$, where $g^{(B)}_k$ is the population of the $k$th eigenstate of $\gamma_B$. We have
    \begin{align}
        \bra{\gamma_B^{-1}}\left(\mathbbm{1}_A\otimes\gamma_B - \frac{1}{\mathrm{Tr}\{\gamma_B^{-1}\}}\ket{\phi^+}\!\bra{\phi^+}_{AB} \right)\ket{\gamma_B^{-1}} = \mathrm{Tr}\{\gamma_B^{-1}\} - \frac{1}{\mathrm{Tr}\{\gamma_B^{-1}\}}\left(\mathrm{Tr}\{\gamma_B^{-1}\}\right)^2 = 0.
    \end{align}
    Besides, defining $\ket{\gamma_B^{-1/2}}\equiv \sum_{i=0}^{d-1}1/\sqrt{g_i^{(B)}}\ket{i}_A\otimes\ket{i}_B$, we have 
    \begin{align}
        \mathbbm{1}_A\otimes\gamma_B - \frac{1}{\mathrm{Tr}\{\gamma_B^{-1}\}}\ket{\phi^+}\!\bra{\phi^+}_{AB} = (\mathbbm{1}_A\otimes\gamma_B^{1/2})\left(\mathbbm{1}_A\otimes\mathbbm{1}_B - \frac{1}{\mathrm{Tr}\{\gamma_B^{-1}\}}\ket{\gamma_B^{-1/2}}\!\bra{\gamma_B^{-1/2}}\right)(\mathbbm{1}_A\otimes\gamma_B^{1/2}) \ge 0,
        \label{appeq:I_A_otimes_gamma_B>phi_AB/Tr(gamma_B^-1)}
    \end{align}
    because $\mathbbm{1}_A\otimes\mathbbm{1}_B - \ket{\gamma_B^{-1/2}}\!\bra{\gamma_B^{-1/2}}/\mathrm{Tr}\{\gamma_B^{-1}\} \ge 0$ following from the fact that $\ket{\gamma_B^{-1/2}}\!\bra{\gamma_B^{-1/2}}/\mathrm{Tr}\{\gamma_B^{-1}\}\in\mathcal{S}(AB)$ is a pure quantum state. Hence, using Lemma~\ref{lem:barely_positive}, we conclude that $R(\mathcal{V}||\Gamma) = \mathrm{Tr}\{\gamma_B^{-1}\} - 1$, and therefore, $R(\mathcal{U}||\Gamma) = \mathrm{Tr}\{\gamma_B^{-1}\} - 1$ for all unitary channels $\mathcal{U}\in\mathcal{O}(A\rightarrow B)$.

    We now prove that, without conditioning on dimensions, $R(\Lambda||\Gamma) \le \mathrm{Tr}\{\gamma^{-1}_B\}-1$ for any channel $\Lambda$ in general. By the definition of $R$ in Eq.~(\ref{appeq:R_def}), this is equivalent to proving that for any channel $\Lambda$ we have $\mathrm{Tr}\{\gamma^{-1}_B\}(\mathbbm{1}_A\otimes\gamma_B) \ge J_\Lambda$. Note that $J_\Lambda \in \mathcal{L}(A\otimes B)$ where $\mathcal{L}(X)$ denotes the set of linear operators on the Hilbert space $X$. We consider the map $\Omega: \mathcal{L}(A\otimes B) \rightarrow \mathcal{L}(A\otimes B)$ defined as
        \begin{align}
            \Omega(J) := \mathrm{Tr}_B\{J\}\otimes\mathrm{Tr}\{\gamma^{-1}_B\}\gamma_B - J.
        \end{align}
        If the map $\Omega$ is CP, we have $\Omega(J_\Lambda) = \mathbbm{1}_A\otimes\mathrm{Tr}\{\gamma^{-1}_B\}\gamma_B - J_\Lambda \ge 0$ and thus compete our proof. To show that $\Omega$ is CP, we consider the Choi operator of $\Omega$ given by $J_\Omega \equiv \mathcal{I}\otimes\Omega(\Phi^+)\in\mathcal{L}(A\otimes B\otimes A'\otimes B')$ where $\Phi^+\equiv \sum_{i,k=0}^{d_A-1}\sum_{j,\ell=0}^{d_B-1}\ket{i}\!\bra{k}_{A}\otimes\ket{j}\!\bra{\ell}_{B}\otimes\ket{i}\!\bra{k}_{A'}\otimes\ket{j}\!\bra{\ell}_{B'}$. Noticing that
        \begin{align}
            \mathcal{I}\otimes\mathrm{Tr}_{B'}\{\Phi^+\}\otimes\mathrm{Tr}\{\gamma^{-1}_{B'}\}\gamma_{B'} = \sum_{i,k=0}^{d_A-1}\ket{i}\!\bra{k}_{A}\otimes\mathbbm{1}_{B}\otimes\ket{i}\!\bra{k}_{A'}\otimes\mathrm{Tr}\{\gamma^{-1}_{B'}\}\gamma_{B'},
        \end{align}
        we have
        \begin{align}
            J_\Omega &= \mathcal{I}\otimes\mathrm{Tr}_{B'}\{\Phi^+\}\otimes\mathrm{Tr}\{\gamma^{-1}_{B'}\}\gamma_{B'} - \Phi^+\\
            &= \sum_{i,k=0}^{d_A-1}\ket{i}\!\bra{k}_{A}\otimes\mathbbm{1}_{B}\otimes\ket{i}\!\bra{k}_{A'}\otimes\mathrm{Tr}\{\gamma_{B'}^{-1}\}\gamma_{B'} - \sum_{i,k=0}^{d_A-1}\sum_{j,\ell=0}^{d_B-1}\ket{i}\!\bra{k}_{A}\otimes\ket{j}\!\bra{\ell}_{B}\otimes\ket{i}\!\bra{k}_{A'}\otimes\ket{j}\!\bra{\ell}_{B'}\\
            &=  \sum_{i,k=0}^{d_A-1}\ket{i}\!\bra{k}_{A}\otimes\ket{i}\!\bra{k}_{A'}\otimes\mathbbm{1}_{B}\otimes\mathrm{Tr}\{\gamma_{B'}^{-1}\}\gamma_{B'} - \sum_{i,k=0}^{d_A-1}\sum_{j,\ell=0}^{d_B-1}\ket{i}\!\bra{k}_{A}\otimes\ket{i}\!\bra{k}_{A'}\otimes\ket{j}\!\bra{\ell}_{B}\otimes\ket{j}\!\bra{\ell}_{B'}\\
            &= \sum_{i,k=0}^{d_A-1}\ket{i}\!\bra{k}_{A}\otimes\ket{i}\!\bra{k}_{A'}\otimes\left[\mathrm{Tr}\{\gamma_{B'}^{-1}\}(\mathbbm{1}_B\otimes\gamma_{B'}) - \sum_{j,\ell=0}^{d_B-1}\ket{j}\!\bra{\ell}_{B}\otimes\ket{j}\!\bra{\ell}_{B'}\right]\\
            &= \ket{\phi^+}\!\bra{\phi^+}_{AA'}\otimes\left[\mathrm{Tr}\{\gamma_{B'}^{-1}\}(\mathbbm{1}_B\otimes\gamma_{B'}) - \ket{\phi^+}\!\bra{\phi^+}_{BB'}\right].
        \end{align}
        Since we have shown that $R(\mathcal{I}||\Gamma) = \mathrm{Tr}\{\gamma^{-1}_B\} - 1$, the term $\mathrm{Tr}\{\gamma^{-1}_{B'}\}(\mathbbm{1}_B\otimes\gamma_{B'}) - \ket{\phi^+}\!\bra{\phi^+}_{BB'} \ge 0$. Therefore $J_\Omega \ge 0$ and $\Omega$ is CP. We thus prove that $\Omega(J_\Lambda) = \mathbbm{1}_A\otimes\mathrm{Tr}\{\gamma^{-1}_B\}\gamma_B - J_\Lambda \ge 0$ for any channel $\Lambda$, which means that $R(\Lambda||\Gamma)\le \mathrm{Tr}\{\gamma^{-1}_B\} - 1,\, \forall\, \Lambda\in\mathcal{O}(A\rightarrow B)$.
\end{proof}
Theorem~\ref{thm:upper_bound_R} is a special case of Theorem~\ref{appthm:upper_bound_R} with the additional conditions $d_A = d_B$ and $\gamma_A = \gamma_B \equiv \gamma$.
\section{Channel mutual information, \texorpdfstring{$R_\mathrm{T}$}{RT} and \texorpdfstring{$R_\mathrm{S}$}{RS}}
In this section, we consider the case shown in Fig.~\ref{fig:channal_mutual_info}.
\begin{figure}
    \centering
    \includegraphics[width=0.4\linewidth]{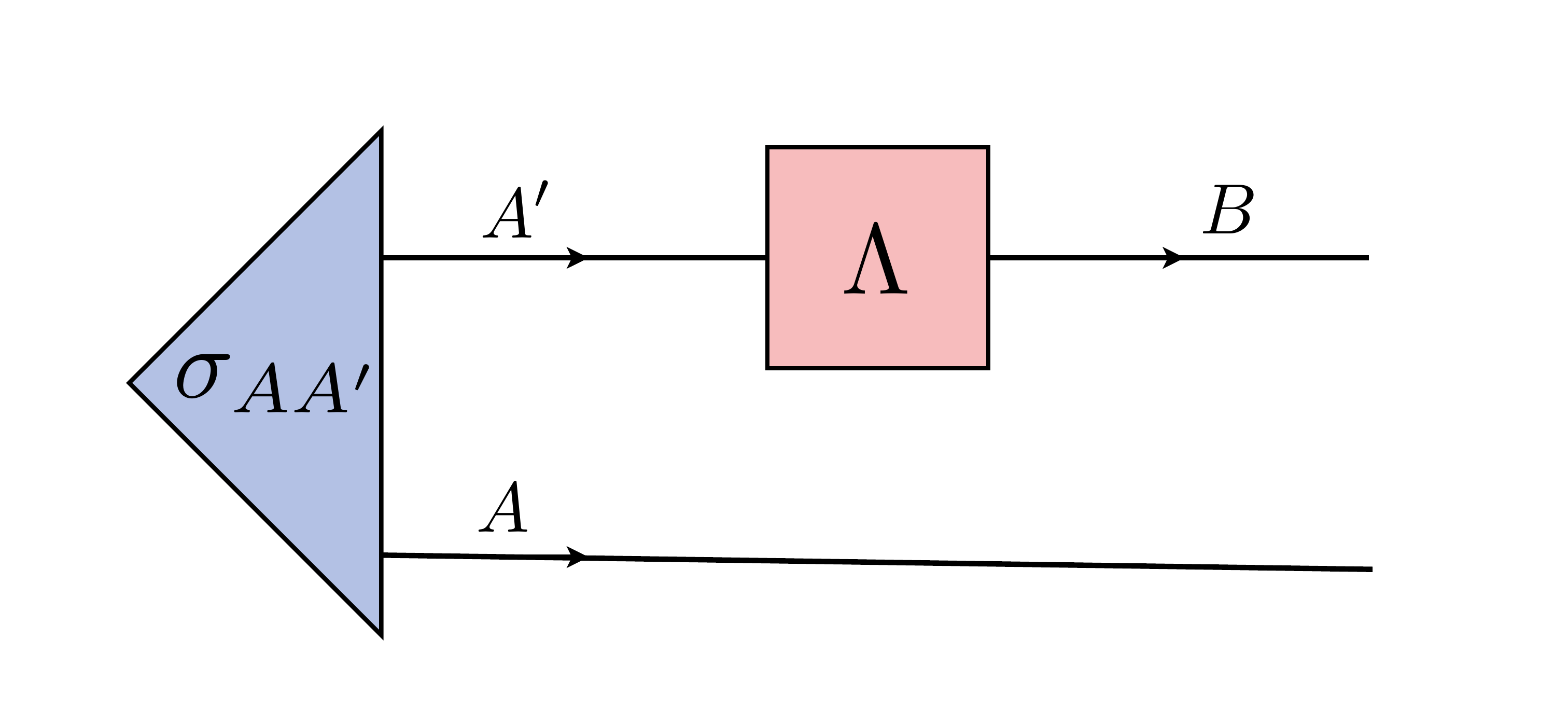}
    \caption{The scenario in the definition of the channel mutual information.}
    \label{fig:channal_mutual_info}
\end{figure}
Given system $B$ and system $A$ with its copy $A'$, the mutual information of the channel $\Lambda\in\mathcal{O}(A'\rightarrow B)$ is defined as~\cite{wilde_2013}
\begin{align}
    I(\Lambda) := \max_{\sigma_{AA'}\in\mathcal{S}(AA')}I(A:B)_{\mathcal{I}\otimes\Lambda(\sigma_{AA'})}=\max_{\sigma_{AA'}\in\mathcal{S}(AA')} D(\sigma_{AB}^\mathrm{out}||\sigma_A^\mathrm{out}\otimes\sigma_{B}^\mathrm{out}),
    \label{eq:channel_mutual_info_def}
\end{align}
where $\sigma_{AB}^\mathrm{out} \equiv \mathcal{I}\otimes\Lambda(\sigma_{AA'})$ is the output state, $\sigma_X^{\rm out} \equiv \mathrm{Tr}_{\{A,B\}\text{\textbackslash} X}\{\sigma_{AB}^{\rm out}\}$ with $X\in\{A,B\}$ and $D(\rho_1||\rho_2)\equiv \mathrm{Tr}\{\rho_1\ln\rho_1-\rho_1\ln\rho_2\}$ is the relative entropy between $\rho_1$ and $\rho_2$.
\begin{theoremApp}[Lower bound on channel mutual information by $R$ and $R_\mathrm{T}$]\label{thm:I_lower_bound_ROA}
    Given a general channel $\Lambda\in\mathcal{O}(A'\rightarrow B)$, the following bounds hold
    \begin{align}
        I(\Lambda) \ge I(A:B)_{\mathcal{I}\otimes\Lambda(\tilde{\gamma}_{AA'})} \ge 2 \left(g_{\min}^{(AB)}\right)^2\left[R(\Lambda||\Gamma) - R_\mathrm{T}(\Lambda)\right]^2.
        \label{appeq:lower_bound_mutual_info}
    \end{align}
    Here, $\tilde{\gamma}_{AA'}\equiv \ket{\tilde{\gamma}}\!\bra{\tilde{\gamma}}_{AA'}$ is the thermofield double state with
    $\ket{\tilde{\gamma}}_{AA'} \equiv \sum_{i}\sqrt{g_i}\ket{i}_A\otimes\ket{i}_{A'},$ 
    where $g_i$ is the population of $\gamma_A$ of the $i$th energy level $\ket{i}_A$, and $g_{\min}^{(AB)}$ is the smallest population for the thermal state $\gamma_{AB}\equiv \gamma_A\otimes\gamma_{B}$.
\end{theoremApp}
\begin{proof}
    The first inequality directly follows from the definition of the channel mutual information [Eq.~(\ref{eq:channel_mutual_info_def})]. To prove the second inequality, we denote $\rho_{AB}\equiv \mathcal{I}\otimes\Lambda(\tilde{\gamma}_{AA'})$. Therefore,
    \begin{align}
        I(A:B)_{\rho_{AB}} &= D(\rho_{AB}||\rho_A\otimes\rho_{B}) \\
        &= D(\rho_{AB}||\gamma_A\otimes\rho_{B})\\
        &\ge \frac{1}{2}||\rho_{AB} - \gamma_A\otimes\rho_{B}||_1^2\\
        &= \frac{1}{2}||\mathcal{I}\otimes\Lambda(\tilde{\gamma}_{AA'}) - \gamma_A\otimes\rho_{B}||_1^2\\
        &\ge 2 \left(g_{\min}^{(AB)}\right)^2\Big|R_\mathrm{T}[\mathcal{I}\otimes\Lambda(\tilde{\gamma}_{AA'})]-R_\mathrm{T}[\gamma_A\otimes\Lambda(\gamma_{A'})]\Big|^2,\\
        &= 2 \left(g_{\min}^{(AB)}\right)^2\left[R(\Lambda||\Gamma) - R_\mathrm{T}(\Lambda)\right]^2
    \end{align}
    where the first inequality comes from the Pinsker's inequality (Theorem~11.9.1 in Ref.~\cite{wilde_2013}), the second inequality follows from the continuity of $R_\mathrm{T}$ (Theorem~\ref{lem:ROA_continuity}) and the last line follows from the multiplicity of $R_\mathrm{T}$, Theorem~\ref{appthm:ER_ROA} and Theorem~\ref{appthm:ER_R}. We hence obtain the lower bound in Eq.~(\ref{appeq:lower_bound_mutual_info}).
\end{proof}
\begin{theoremApp}[Upper bound of channel mutual information by $R_\mathrm{S}$]\label{thm:I_upper_bound_ROS}
    Given a general channel $\Lambda\in\mathcal{O}(A'\rightarrow B)$, the following bound holds
    \begin{align}
        R_\mathrm{S}(\Lambda) \ge I(A:B)_{\mathcal{I}\otimes\Lambda(\tilde{\gamma}_{AA'})}.
    \end{align}
\end{theoremApp}
\begin{proof}
    We consider the max-relative entropy~\cite{datta2009min}
    \begin{align}
        D_{\max}(\rho||\sigma) := \ln \min\{\lambda|\rho\le\lambda\sigma\}.
        \label{appeq:max_relative_entropy}
    \end{align}
    Define $\rho_{AB} \equiv \mathcal{I}\otimes\Lambda(\tilde{\gamma}_{AA'})$. We then have
    \begin{align}
        I(A:B)_{\mathcal{I}\otimes\Lambda(\tilde{\gamma}_{AA'})} &=D(\rho_{AB}||\rho_A\otimes\rho_{B})\\
        &= \min_{\sigma_{B}} D(\rho_{AB}||\rho_A\otimes\sigma_{B})\\
        &= \min_{\sigma_{B}} D(\rho_{AB}||\gamma_A\otimes\sigma_{B})\\
        &\le \min_{\sigma_{B}} D_{\max}(\rho_{AB}||\gamma_A\otimes\sigma_{B})\\ & = \min_{\sigma_{B}}  D_{\max}(J_\Lambda||\mathbbm{1}_A\otimes \sigma_{B})\\ &= \ln[1+R_\mathrm{S}(\Lambda)],
    \end{align}
    where the second line follows from 
    \begin{align}
    D(\rho_{AB}||\rho_A\otimes\rho_{B})&=\min_{\sigma_{B}} \left[D(\rho_{AB}||\rho_A\otimes\rho_{B}) + D(\rho_{B}||\sigma_{B})\right] = \min_{\sigma_{B}} D(\rho_{AB}||\rho_A\otimes\sigma_{B}),
    \end{align}
    the inequality is due to $D_{\max}(\rho||\sigma) \ge D(\rho||\sigma),\,\forall \rho,\sigma$~\cite{datta2009min}, the fifth line is from the fact that
    \begin{align}
        J_\Lambda \le \lambda\mathbbm{1}_A\otimes\sigma_{B} \Leftrightarrow (\gamma_A^{1/2}\otimes \mathbbm{1}_B)J_\Lambda (\gamma_A^{1/2}\otimes \mathbbm{1}_B) \le \lambda(\gamma_A^{1/2}\otimes \mathbbm{1}_B)\mathbbm{1}_A\otimes\sigma_{B} (\gamma_A^{1/2}\otimes \mathbbm{1}_B) \Leftrightarrow \rho_{AB} \le \lambda \gamma_A\otimes\sigma_{B},
    \end{align}
    and the last line follows from the definition of $R_\mathrm{S}(\Lambda)$ [Eq.~(\ref{appeq:R_S_def})] with noticing that $J_{\Phi_\sigma} \equiv \mathbbm{1}\otimes\sigma$. Because $R_\mathrm{S}(\Lambda)\ge 0$, we have $R_\mathrm{S}(\Lambda)\ge \ln[1+R_\mathrm{S}(\Lambda)] \ge I(A:B)_{\mathcal{I}\otimes\Lambda(\tilde{\gamma}_{AA'})}$.
\end{proof}
\subsection{Proof of Theorem~\ref{thm:R_S_bounds}} \label{app:proof_Thm_R_S_bounds}
The upper bound in Theorem~\ref{thm:R_S_bounds} is straightforward and has been shown in the main text. For the lower bound, according to Theorem~\ref{thm:I_lower_bound_ROA} and Theorem~\ref{thm:I_upper_bound_ROS}, for a general channel $\Lambda\in\mathcal{O}(A'\rightarrow B)$, we have
\begin{align}
    R_\mathrm{S}(\Lambda) \ge 2\left(g_{\min}^{(AB)}\right)^2\left[R(\Lambda||\Gamma) - R_\mathrm{T}(\Lambda)\right]^2.
\end{align}
Since $A'$ is a copy of system $A$, the above also holds for the channel $\Lambda\in\mathcal{O}(A\rightarrow B)$.
Therefore, we proved the lower bound in Theorem~\ref{thm:R_S_bounds}.

\subsection{Signalling channels with zero athermality transmitability}\label{app:exp_signalling_ch_zero_athermality_transmitability}
Here we exemplify that a signalling channel $\Lambda\in\mathcal{O}(A'\rightarrow B)$ with $R_\mathrm{S}(\Lambda) > 0$ can have zero athermality transmitability, i.e.,~$R(\Lambda||\Gamma) - R_\mathrm{T}(\Lambda) = 0$. Define $\rho_{AB} \equiv \mathcal{I}\otimes\Lambda(\tilde{\gamma}_{AA'})$. By Theorem~\ref{appthm:ER_R} and Theorem~\ref{appthm:ER_ROA}, we have $R(\Lambda||\Gamma) = R_\mathrm{T}(\rho_{AB})$ and $R_\mathrm{T}(\Lambda) = R_\mathrm{T}(\rho_B)$ where $\rho_B\equiv \mathrm{Tr}_{A}\{\rho_{AB}\}$. According to Ref.~\cite{datta2009min}, $1 + R_\mathrm{T}(\rho) = ||\gamma^{-1/2}\rho\gamma^{-1/2}||_\infty$. Therefore,
$R_\mathrm{T}(\rho_{AB}) = R_\mathrm{T}(\rho_B)$ is equivalent to 
\begin{align}
    ||(\gamma_{A}\otimes\gamma_B)^{-1/2}\rho_{AB}(\gamma_{A}\otimes\gamma_B)^{-1/2}||_\infty = ||\gamma_B^{-1/2}\rho_B\gamma_B^{-1/2}||_\infty.
    \label{appeq:infty_norm_equality}
\end{align}
We consider the infinite temperature case, where $\gamma_{A(B)} = \mathbbm{1}/d_{A(B)}$. When $\Lambda$ is entanglement-breaking~\cite{horodecki2003entanglement}, $\rho_{AB}$ is separable. We consider a special $\rho_{AB}$ which can be written as
\begin{align}
    \rho_{AB} = \sum_{i=0}^{d_A-1}\sum_{j=0}^{d_B-1} p_{ij}\ket{i}\!\bra{i}_A\otimes\ket{\psi_j}\!\bra{\psi_j}_B,
\end{align}
where $p_{ij}\ge 0, \, \forall\, i,j$, $\sum_{i=0}^{d_A-1}\sum_{j=0}^{d_B-1}p_{ij} = 1$, and $\{\ket{i}_A\}_{i=0}^{d_A-1}$ is the eigenbasis of $\gamma_A$ while $\{\ket{\psi_j}_B\}_{j=0}^{d_B-1}$ is a general orthonormal basis for system $B$. By requiring $\rho_A \equiv \mathrm{Tr}_{B}\{\rho_{AB}\} = \gamma_A$, we have $\sum_{j=0}^{d_B-1}p_{ij} = d_A^{-1}, \,\forall\, i$. Also,
\begin{align}
    ||(\gamma_{A}\otimes\gamma_B)^{-1/2}\rho_{AB}(\gamma_{A}\otimes\gamma_B)^{-1/2}||_\infty &= d_Ad_B||\rho_{AB}||_\infty = d_Ad_B\max_{i,j}p_{ij},\\
    ||\gamma_B^{-1/2}\rho_B\gamma_B^{-1/2}||_\infty &= d_B||\rho_B||_\infty = d_B\max_{j}\sum_{i=0}^{d_A-1}p_{ij}.
\end{align}
Thus, Eq.~(\ref{appeq:infty_norm_equality}) can hold when $\exists \, j$, $p_{ij} = \max_{i,j}p_{ij}, \, \forall\, i$. For example, when $d_A = d_B = 3$, we consider the following channel
\begin{align}
    \Lambda(\cdot) &:= \mathrm{Tr}\{\ket{0}\!\bra{0}_A(\cdot)\}(\frac{1}{2}\ket{0}\!\bra{0}_B + \frac{1}{4}\ket{1}\!\bra{1}_B + \frac{1}{4}\ket{2}\!\bra{2}_B)\nonumber\\
    &\quad + \mathrm{Tr}\{\ket{1}\!\bra{1}_A(\cdot)\}(\frac{1}{2}\ket{0}\!\bra{0}_B + \frac{3}{10}\ket{1}\!\bra{1}_B + \frac{1}{5}\ket{2}\!\bra{2}_B)\nonumber\\
    &\quad + \mathrm{Tr}\{\ket{2}\!\bra{2}_A(\cdot)\}(\frac{1}{2}\ket{0}\!\bra{0}_B + \frac{1}{5}\ket{1}\!\bra{1}_B + \frac{3}{10}\ket{2}\!\bra{2}_B)
    \label{appeq:EB_Lambda}
\end{align}
The corresponding $\rho_{AB}$ is given by
\begin{align}
    \rho_{AB} &= \frac{1}{6}\ket{0}\!\bra{0}_A\otimes\ket{0}\!\bra{0}_B + \frac{1}{12}\ket{0}\!\bra{0}_A\otimes \ket{1}\!\bra{1}_B + \frac{1}{12}\ket{0}\!\bra{0}_A\otimes \ket{2}\!\bra{2}_B \nonumber\\
    &\quad + \frac{1}{6}\ket{1}\!\bra{1}_A\otimes\ket{0}\!\bra{0}_B + \frac{1}{10}\ket{1}\!\bra{1}_A\otimes \ket{1}\!\bra{1}_B + \frac{1}{15}\ket{1}\!\bra{1}_A\otimes \ket{2}\!\bra{2}_B\nonumber\\
    &\quad + \frac{1}{6}\ket{2}\!\bra{2}_A\otimes\ket{0}\!\bra{0}_B + \frac{1}{15}\ket{2}\!\bra{2}_A\otimes \ket{1}\!\bra{1}_B + \frac{1}{10}\ket{2}\!\bra{2}_A\otimes \ket{2}\!\bra{2}_B
\end{align}
which leads to $||(\gamma_{A}\otimes\gamma_B)^{-1/2}\rho_{AB}(\gamma_{A}\otimes\gamma_B)^{-1/2}||_\infty = ||\gamma_B^{-1/2}\rho_B\gamma_B^{-1/2}||_\infty = 1.5$. Thus, for the channel $\Lambda$ defined in Eq.~(\ref{appeq:EB_Lambda}), despite of $R_\mathrm{S}(\Lambda) > 0$, $R(\Lambda||\Gamma) - R_\mathrm{T}(\Lambda) = 0$. 
\begin{figure}[h]
    \centering
    \includegraphics[width=0.5\linewidth]{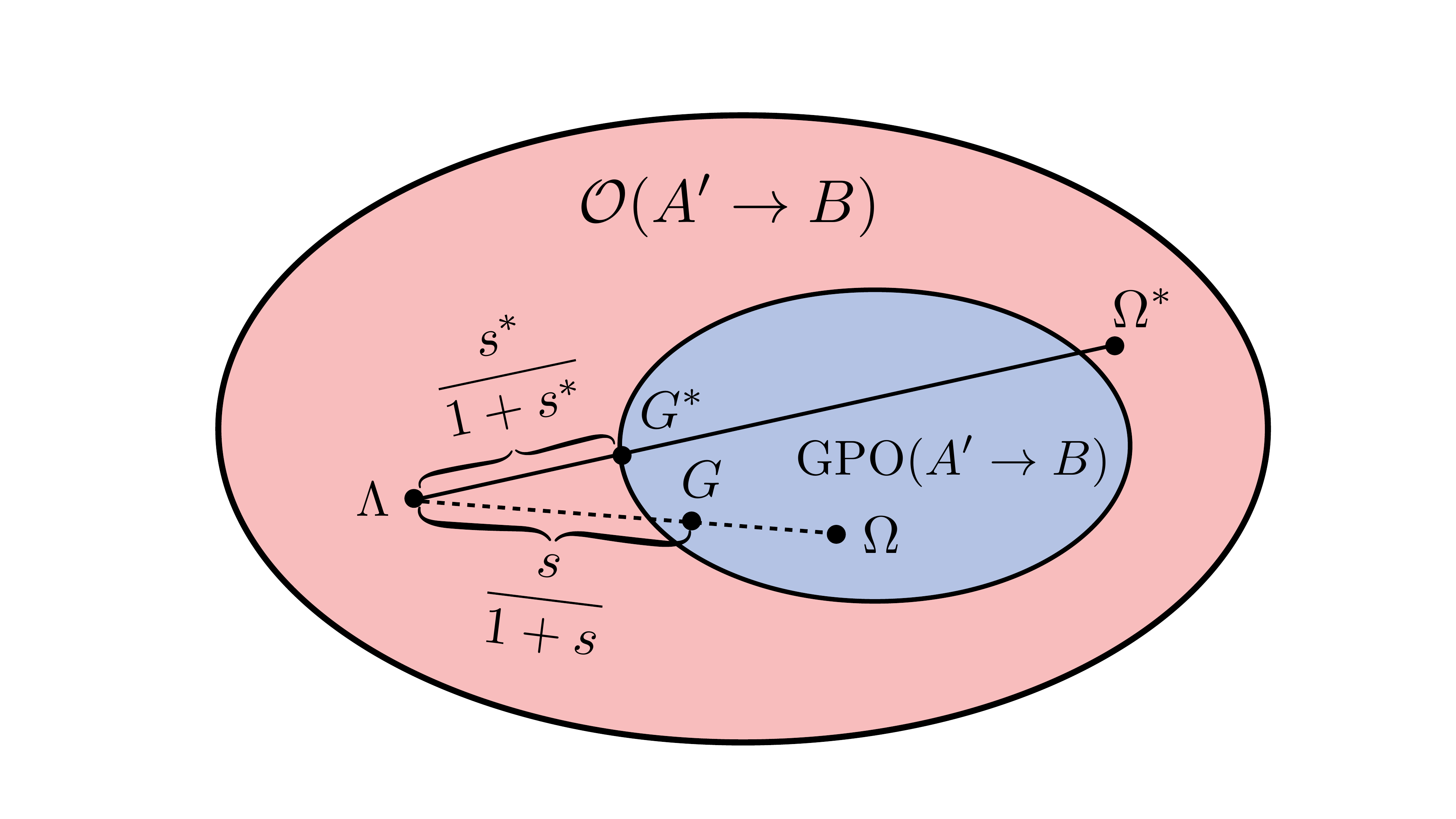}
    \caption{Geometric illustration of $R_\mathrm{T}(\Lambda)$.}
    \label{fig:robustness}
\end{figure}

Finally, we note that according to the standard definition of $R_\mathrm{T}(\Lambda)$ [Eq.~(\ref{appeq:ROA_ch_std})]:
\begin{align}
    R_\mathrm{T}(\Lambda) := \min\left\{s\Big|\frac{\Lambda + s\,\Omega}{1+s} = G,\,\Omega\in\mathcal{O}(A'\rightarrow B),\, G\in\mathrm{GPO}(A'\rightarrow B)\right\},
    \label{appeq:ROA_ch_std_exp}
\end{align}
$R_\mathrm{T}(\Lambda)$ characterises the ``distance'' of $\Lambda$ to the set $\mathrm{GPO}(A'\rightarrow B)$ with the closest GPO to $\Lambda$ being the optimal channel $G^*$ of the minimisation in Eq.~(\ref{appeq:ROA_ch_std_exp}) (See Fig.~\ref{fig:robustness}). Therefore, channels for which $R(\Lambda||\Gamma)-R_\mathrm{T}(\Lambda) = 0$ have the completely thermalising channel $\Gamma$ as their closest GPOs under this measure. This offers a heuristic interpretation of the vanishing athermality transmitability of such channels.

\section{GPO dilation of a general quantum channel (Proof of Theorem~\ref{thm:R_GPO_dilation})}\label{app:GPO_dilation}
In this section, we provide operational meanings of the measures $R_\mathrm{T}(\Lambda)$ and $R(\Lambda||\Gamma)$ by considering the GPO dilation of $\Lambda$.
A GPO dilation of a general channel $\Lambda\in\mathcal{O}(A\rightarrow B)$ can be viewed as the combination of a GPO $G\in\mathrm{GPO}(CA\rightarrow B)$ and an athermal state $\rho_C\in\mathcal{S}(C)$, such that $\Lambda(\cdot) = G(\rho_C \otimes \cdot)$. We emphasise that in a GPO dilation, the GPO $G$ needs not to be unitary, which is different from the Stinespring's representation of quantum channels~\cite{Stinespring1955}.

The following theorem relates this simulation task to $R_\mathrm{T}(\Lambda)$:
\begin{theoremApp}[Operational meaning of $R_\mathrm{T}(\Lambda)$]
Given a general quantum channel $\Lambda \in \mathcal{O}(A\rightarrow B)$, we define the set of states which can simulate this channel with a GPO as $\mathcal{S}^{\,\mathrm{sim}}_\mathrm{GPO}(\Lambda)$. The states in $\mathcal{S}^{\,\mathrm{sim}}_\mathrm{GPO}(\Lambda)$ can use different GPOs for simulation. The following equality holds:
\begin{align}
    R_\mathrm{T}(\Lambda) = \min_{\rho_C\in\mathcal{S}^{\,\mathrm{sim}}_\mathrm{GPO}(\Lambda)} R_\mathrm{T}(\rho_C).
\end{align}
\end{theoremApp}
\begin{proof}
    First, we prove that $R_\mathrm{T}(\Lambda) \le \min_{\rho_C\in\mathcal{S}^{\,\mathrm{sim}}_\mathrm{GPO}(\Lambda)} R_\mathrm{T}(\rho_C)$. Assume that  $\exists\, \rho_C\in\mathcal{S}(C), \, G\in\mathrm{GPO}(CA\rightarrow B)$ such that $G(\rho_C\otimes\cdot) = \Lambda(\cdot)$. We have
    \begin{align}
        R_\mathrm{T}(\Lambda) = R_\mathrm{T}[\Lambda(\gamma_A)] = R_\mathrm{T}[G(\rho_C\otimes\gamma_A)] \le R_\mathrm{T}(\rho_C\otimes\gamma_A) = R_\mathrm{T}(\rho_C),
    \end{align}
    where the first equality follows from Theorem~\ref{appthm:ER_ROA}, the inequality is due to the monotonicity of $R_\mathrm{T}$ and in the final equality we used the multiplicity of $R_\mathrm{T}$. Thus, $R_\mathrm{T}(\Lambda)\le R_\mathrm{T}(\rho_C)$ holds for all $\rho_C$ and $G$ that simulate $\Lambda$. We conclude that $R_\mathrm{T}(\Lambda) \le \min_{\rho_C\in\mathcal{S}^{\,\mathrm{sim}}_\mathrm{GPO}(\Lambda)} R_\mathrm{T}(\rho_C)$.

    To prove that $R_\mathrm{T}(\Lambda) \ge \min_{\rho_C\in\mathcal{S}^{\,\mathrm{sim}}_\mathrm{GPO}(\Lambda)} R_\mathrm{T}(\rho_C)$. we consider a specific channel simulation protocol using a qubit system $C$, with the thermal state $\gamma_C$ given by
    \begin{align}
        \gamma_C &= \frac{1}{1+R_\mathrm{T}(\Lambda)}\ket{0}\!\bra{0}_C + \frac{R_\mathrm{T}(\Lambda)}{1+R_\mathrm{T}(\Lambda)}\ket{1}\!\bra{1}_C.
    \end{align}    
    The joint GPO for simulation is constructed in the following:
    \begin{align}
        \tilde{G}(\tau_C \otimes \rho_{A}) := \mathrm{Tr}\{\ket{0}\!\bra{0}_C\tau_C\} \Lambda(\rho_A) + \mathrm{Tr}\{\ket{1}\!\bra{1}_C\tau_C\} \sigma_B,
        \label{eq:measure_prepare_GPO_1}
    \end{align}
    where $\sigma_B$ is determined from
    \begin{align}
        \frac{1}{1+R_\mathrm{T}(\Lambda)}\Lambda(\gamma_A) + \frac{R_\mathrm{T}(\Lambda)}{1+R_\mathrm{T}(\Lambda)}\sigma_B = \gamma_B.
    \end{align}
    According to the equivalence relation between the static and dynamical $R_\mathrm{T}$ (Theorem~\ref{appthm:ER_ROA}), i.e.,~$R_\mathrm{T}(\Lambda) = R_\mathrm{T}[\Lambda(\gamma_A)]$, we note that the state $\sigma_B$ exists and is the optimal state in the standard form of $R_\mathrm{T}[\Lambda(\gamma_A)]$ [Eq.~(\ref{appeq:ROA_st_std})]. It is straightforward to check that using the state $\rho_C \equiv \ket{0}\!
    \bra{0}_C$, the channel
    $\tilde{G}$ in Eq.~(\ref{eq:measure_prepare_GPO_1}) is a valid GPO that simulates $\Lambda$:
    \begin{align}
        \tilde{G}(\ket{0}\!\bra{0}_C\otimes\cdot) &= \Lambda(\cdot),\\
        \tilde{G}(\gamma_C\otimes\gamma_A) &= \gamma_B.
    \end{align}
    Since $R_\mathrm{T}(\rho_C) = R_\mathrm{T}(\ket{0}\!\bra{0}_C) = R_\mathrm{T}(\Lambda)$, we have $R_\mathrm{T}(\Lambda) = R_\mathrm{T}(\rho_C) \ge \min_{\rho_C\in\mathcal{S}^{\,\mathrm{sim}}_\mathrm{GPO}(\Lambda)} R_\mathrm{T}(\rho_C)$. Hence, we have
    \begin{align}
        R_\mathrm{T}(\Lambda) = \min_{\rho_C\in\mathcal{S}^{\,\mathrm{sim}}_\mathrm{GPO}(\Lambda)} R_\mathrm{T}(\rho_C).
    \end{align}
\end{proof}
In a similar way, we can provide an operational meaning for $R(\Lambda||\Gamma)$ in terms of the athermality preservability $P_\mathrm{T}$:
\begin{theoremApp}[{Operational meaning of $R(\Lambda||\Gamma)$ (Theorem~\ref{thm:R_GPO_dilation})}]
Given a general quantum channel $\Lambda \in \mathcal{O}(A\rightarrow B)$, we define the set of GPOs which can simulate this channel with a state as $\mathcal{O}^{\,\mathrm{sim}}_\mathrm{GPO}(\Lambda)$. The GPOs in $\mathcal{O}^{\,\mathrm{sim}}_\mathrm{GPO}(\Lambda)$ can use different static resources for simulation.
The following inequalities hold:
\begin{align}
    R(\Lambda||\Gamma) \le \min_{G\in\mathcal{O}^{\,\mathrm{sim}}_\mathrm{GPO}(\Lambda)} P_\mathrm{T}(G) \le \max\left\{1/R_\mathrm{T}(\Lambda), \, R(\Lambda||\Gamma)\right\}.
    \label{appeq:R(Lambda||Gamma)_bounds}
\end{align}
Especially, whenever $R_\mathrm{T}(\Lambda) R(\Lambda||\Gamma) \ge 1$, we have
\begin{align}
    R(\Lambda||\Gamma) = \min_{G\in\mathcal{O}^{\,\mathrm{sim}}_\mathrm{GPO}(\Lambda)} P_\mathrm{T}(G).
    \label{appeq:R(Lambda||Gamma)=P_T(G)}
\end{align}
\end{theoremApp}
\begin{proof}
    Firstly, we prove the lower bound. 
    Assume that
    $\exists\, \rho_C\in\mathcal{S}(C), \, G\in\mathrm{GPO}(CA\rightarrow B)$ and $R_\mathrm{T}(\rho_C)\ge R_\mathrm{T}(\Lambda)$ that give $G(\rho_C\otimes\cdot) = \Lambda(\cdot)$. Recalling the definition of $P_\mathrm{T}$ in Eq.~(\ref{appeq:P_T_def}), we then have
    \begin{align}
        1 + P_\mathrm{T}(G) &\equiv \min\{\lambda \,|\, J_G \le \lambda(\mathbbm{1}_C\otimes\mathbbm{1}_A\otimes\gamma_B)\} \\
        &\ge \min\left\{\lambda\, | \,\mathrm{Tr}_C\{(\rho_C^T\otimes\mathbbm{1}_A\otimes\mathbbm{1}_B)J_G\} \le \lambda \mathrm{Tr}_C\{(\rho_C^T\otimes\mathbbm{1}_A\otimes\mathbbm{1}_B)(\mathbbm{1}_C\otimes\mathbbm{1}_A\otimes\gamma_B)\}\right\}\\
        &= \min\{\lambda \,|\, J_\Lambda \le \lambda (\mathbbm{1}_A\otimes\gamma_B)\}\\
        &= 1 + R(\Lambda||\Gamma),
    \end{align}
    where, in the second line we have used that multiplication with $\rho_C^\mathrm{T}$ followed by a trace over $C$ is a CP map. Hence, $R(\Lambda||\Gamma) \le P_\mathrm{T}(G)$ for all $\rho_C$ and $G$ used to simulate $\Lambda$, We conclude that $R(\Lambda||\Gamma) \le \min_{G\in \mathcal{O}^{\,\mathrm{sim}}_{\mathrm{GPO}}(\Lambda)} P_\mathrm{T}(G)$.
    
    To prove the upper bound, we consider a specific channel simulation protocol using a qubit system $C$, with the thermal state $\gamma_C$ given by
    \begin{align}
        \gamma_C &= \frac{1}{1+R_\mathrm{T}(\Lambda)}\ket{0}\!\bra{0}_C + \frac{R_\mathrm{T}(\Lambda)}{1+R_\mathrm{T}(\Lambda)}\ket{1}\!\bra{1}_C.
    \end{align}    
    The joint GPO for simulation is constructed in the following:
    \begin{align}
        \tilde{G}(\tau_C \otimes \rho_{A}) := \mathrm{Tr}\{\ket{0}\!\bra{0}_C\tau_C\} \Lambda(\rho_A) + \mathrm{Tr}\{\ket{1}\!\bra{1}_C\tau_C\} \sigma_B,
        \label{appeq:measure_prepare_GPO}
    \end{align}
    where $\sigma_B$ is determined from
    \begin{align}
        \frac{1}{1+R_\mathrm{T}(\Lambda)}\Lambda(\gamma_A) + \frac{R_\mathrm{T}(\Lambda)}{1+R_\mathrm{T}(\Lambda)}\sigma_B = \gamma_B.
        \label{appeq:sigma_B_def_relation}
    \end{align}
    According to the equivalence relation between the static and dynamical $R_\mathrm{T}$ (Theorem~\ref{appthm:ER_ROA}), i.e.,~$R_\mathrm{T}(\Lambda) = R_\mathrm{T}[\Lambda(\gamma_A)]$, we note that $\sigma_B$ is the optimal state in the standard form of $R_\mathrm{T}[\Lambda(\gamma_A)]$ [Eq.~(\ref{appeq:ROA_st_std})]. It is straightforward to check that using the state $\rho_C \equiv \ket{0}\!
    \bra{0}_C$, the channel
    $\tilde{G}$ in Eq.~(\ref{appeq:measure_prepare_GPO}) is a valid GPO that simulates $\Lambda$:
    \begin{align}
        \tilde{G}(\ket{0}\!\bra{0}_C\otimes\cdot) &= \Lambda(\cdot),\\
        \tilde{G}(\gamma_C\otimes\gamma_A) &= \gamma_B.
    \end{align}
    Consider the Choi operator of $\tilde{G}$. We have 
    \begin{align}
        J_{\tilde{G}} &= \ket{0}\!\bra{0}_C\otimes J_\Lambda + \ket{1}\!\bra{1}_C\otimes\mathbbm{1}_A\otimes\sigma_B \\
        &\le \left[1 + R(\Lambda||\Gamma)\right] \ket{0}\!\bra{0}_C\otimes \mathbbm{1}_A\otimes\gamma_B + \ket{1}\!\bra{1}_C\otimes\mathbbm{1}_A\otimes\sigma_B\\
        &\le \left[1 + R(\Lambda||\Gamma)\right] \ket{0}\!\bra{0}_C\otimes \mathbbm{1}_A\otimes\gamma_B + \left[1 + 1/R_\mathrm{T}(\Lambda)\right]\ket{1}\!\bra{1}_C\otimes \mathbbm{1}_A\otimes\gamma_B\\
        &\le \left[1 + \max\left\{1/R_\mathrm{T}(\Lambda), R(\Lambda||\Gamma)\right\}\right]\mathbbm{1}_C\otimes\mathbbm{1}_A\otimes\gamma_B,
    \end{align}
    where the first inequality is from the definition of $R(\Lambda||\Gamma)$ and the second inequality is from the definition of $\sigma_B$ in Eq.~(\ref{appeq:sigma_B_def_relation}). Recalling the definition of $P_\mathrm{T}(G)$, we then have
    \begin{align}
        P_\mathrm{T}(\tilde{G}) \le \max\left\{1/R_\mathrm{T}(\Lambda),\, R(\Lambda||\Gamma)\right\}.
    \end{align}
    Therefore,
    \begin{align}
        \min_{G\in \mathcal{O}^{\,\mathrm{sim}}_{\mathrm{GPO}}(\Lambda)} P_\mathrm{T}(G) \le P_\mathrm{T}(\tilde{G}) \le \max\left\{1/R_\mathrm{T}(\Lambda),\, R(\Lambda||\Gamma)\right\}.
    \end{align}
\end{proof}

\section{Resource analysis for quantum switch}\label{app:quantum_switch}
\subsection{Proof of the upper bound in Theorem~\ref{thm:switch_bounds}}\label{app:proof_switch_upper_bound}
The lower bound [Eq.~(\ref{eq:switch_bound_lower})] in Theorem~\ref{thm:switch_bounds} is a direct application of Theorem~\ref{thm:R_S_bounds}. In the rest of this subsection, we prove the upper bound [Eq.~(\ref{eq:switch_bound_upper})] in Theorem~\ref{thm:switch_bounds}.

Given a general quantum channel $\Omega\in\mathcal{O}(A\rightarrow B)$ with the Kraus representation $\Omega(\cdot) \equiv \sum_n K_n(\cdot) K_n^\dagger$, the Kraus representation of the quantum switch $\mathdutchcal{S}[\Omega, \Omega]\in\mathcal{O}(AC\rightarrow BC')$, where $C'$ is a copy of $C$, can be written as $\mathdutchcal{S}[\Omega, \Omega](\cdot) \equiv \sum_m\sum_n S_{mn}(\cdot)S^\dagger_{mn}$, where the Kraus operator $S_{mn}$ is given by
\begin{align}
    S_{mn} \equiv \ket{0}\!\bra{0}_C\otimes K_n K_m + \ket{1}\!\bra{1}_C\otimes K_m K_n.
    \label{eq:swith_Kraus_general}
\end{align}
Consider the case when systems $A$ and $B$ have the same thermal state, i.e.,~$\gamma_A = \gamma_B \equiv \gamma$ with the dimension of $d$. The completely thermalising channel $\Gamma\in\mathrm{GPO}(A\rightarrow B)$ has the Kraus representation:
\begin{align}
    \Gamma(\cdot) = \mathrm{Tr}\{\cdot\}\gamma = \sum_{i=0}^{d-1}\sum_{j=0}^{d-1}\sqrt{g_j}\ket{j}\!\bra{i}\cdot\ket{i}\!\bra{j}\sqrt{g_j} \equiv \sum_{i=0}^{d-1}\sum_{j=0}^{d-1} \tilde{K}_{ij}(\cdot)\tilde{K}_{ij}^\dagger,
    \label{eq:Gamma_Kraus}
\end{align}
where $g_j$ is the population of $\gamma$ on the $j$th energy level and $\tilde{K}_{ij}\equiv \sqrt{g_j}\ket{j}\!\bra{i}$. 
According to Eq.~(\ref{eq:swith_Kraus_general}), the Kraus operator of the switch $\mathdutchcal{S}[\Gamma, \Gamma]$ is given by
\begin{align}
    S_{ijk\ell}^{(\Gamma)} \equiv \ket{0}\!\bra{0}_C\otimes\tilde{K}_{k\ell}\tilde{K}_{ij} + \ket{1}\!\bra{1}_C\otimes\tilde{K}_{ij}\tilde{K}_{k\ell}, \quad \forall\, i,j,k,\ell = 0,1,\dots, d-1,
\end{align}
such that 
\begin{align}
    \mathdutchcal{S}[\Gamma,\Gamma](\cdot) = \sum_{i=0}^{d-1}\sum_{j=0}^{d-1}\sum_{k=0}^{d-1}\sum_{\ell=0}^{d-1} S_{ijk\ell}^{(\Gamma)}(\cdot)\left(S_{ijk\ell}^{(\Gamma)}\right)^\dagger.
    \label{eq:S[Gamma,Gamma]_Kraus}
\end{align}
The signalling GPO input to the switch is defined as
\begin{align}
    G(\cdot) &\equiv s\Gamma + (1-s)\mathcal{I} \label{appeq:signalling_G}\\
    &= s\sum_{i=0}^{d-1}\sum_{j=0}^{d-1} \tilde{K}_{ij}(\cdot)\tilde{K}_{ij}^\dagger + (1-s)\mathbbm{1}(\cdot)\mathbbm{1}\\
    &\equiv \sum_{i=0}^{d-1}\sum_{j=0}^{d-1}K_{ij}(\cdot)K_{ij}^\dagger + K_{dd}(\cdot)K_{dd}^\dagger,
    \label{eq:signalling_G_Kraus}
\end{align}
where $s\in[0,1]$ is the thermalising strength and we defined the Kraus operators:
\begin{align}
    K_{ij} &\equiv \sqrt{s}\tilde{K}_{ij} = \sqrt{s}\sqrt{g_j}\ket{j}\!\bra{i},\quad\forall\, i,j = 0,1,\dots,d-1, \\
    K_{dd} &\equiv \sqrt{1-s} \mathbbm{1}.
\end{align}
According to Eq.~(\ref{eq:swith_Kraus_general}), the Kraus operators of the switch $\mathdutchcal{S}[G,G]$ are given by
\begin{align}
    S_{ijk\ell}^{(G)} &\equiv \ket{0}\!\bra{0}_C\otimes K_{k\ell}K_{ij} + \ket{1}\!\bra{1}_C\otimes K_{ij}K_{k\ell}\\
    &= s\left(\ket{0}\!\bra{0}_C\otimes\tilde{K}_{k\ell}\tilde{K}_{ij} + \ket{1}\!\bra{1}_C\otimes\tilde{K}_{ij}\tilde{K}_{k\ell} \right)\\
    &= s\,S_{ijk\ell}^{(\Gamma)},\\
    S_{ijdd}^{(G)} &\equiv \ket{0}\!\bra{0}_C\otimes K_{dd}K_{ij} + \ket{1}\!\bra{1}_C \otimes K_{ij}K_{dd}\\
    &= \ket{0}\!\bra{0}_C\otimes\sqrt{1-s}\mathbbm{1}\sqrt{s}\tilde{K}_{ij} + \ket{1}\!\bra{1}_C\otimes \sqrt{s}\tilde{K}_{ij}\sqrt{1-s}\mathbbm{1}\\
    &= \sqrt{s(1-s)}\mathbbm{1}_C\otimes \tilde{K}_{ij},\\
    S_{ddk\ell}^{(G)} &\equiv \ket{0}\!\bra{0}_C\otimes K_{k\ell}K_{dd} + \ket{1}\!\bra{1}_C \otimes K_{dd}K_{k\ell}\\
    &= \ket{0}\!\bra{0}_C\otimes\sqrt{s}\tilde{K}_{k\ell}\sqrt{1-s}\mathbbm{1} + \ket{1}\!\bra{1}_C\otimes\sqrt{1-s}\mathbbm{1}\sqrt{s}\tilde{K}_{k\ell}\\
    &= \sqrt{s(1-s)}\mathbbm{1}_C\otimes \tilde{K}_{k\ell},\\
    S_{dddd}^{(G)} &\equiv \ket{0}\!\bra{0}_C \otimes K_{dd}K_{dd} + \ket{1}\!\bra{1}_C\otimes K_{dd}K_{dd}\\
    &= (1-s)\mathbbm{1}_C\otimes\mathbbm{1},
\end{align}
$\forall\, i,j,k,\ell = 0,1,\dots, d-1$. The channel $\mathdutchcal{S}[G,G]$ can therefore be written as
\begin{align}
    \mathdutchcal{S}[G,G](\cdot) &= \sum_{i=0}^{d-1}\sum_{j=0}^{d-1}\sum_{k=0}^{d-1}\sum_{\ell=0}^{d-1} S_{ijk\ell}^{(G)}(\cdot)\left(S_{ijk\ell}^{(G)}\right)^\dagger + \sum_{i=0}^{d-1}\sum_{j=0}^{d-1} S_{ijdd}^{(G)}(\cdot) \left(S_{ijdd}^{(G)}\right)^\dagger + \sum_{k=0}^{d-1}\sum_{\ell=0}^{d-1} S_{ddk\ell}^{(G)}(\cdot) \left(S_{ddk\ell}^{(G)}\right)^\dagger + S_{dddd}^{(G)}(\cdot) \left(S_{dddd}^{(G)}\right)^\dagger\\
    &= s^2\sum_{i=0}^{d-1}\sum_{j=0}^{d-1}\sum_{k=0}^{d-1}\sum_{\ell=0}^{d-1}S_{ijk\ell}^{(\Gamma)}(\cdot)\left(S_{ijk\ell}^{(\Gamma)}\right)^\dagger+ s(1-s)\sum_{i=0}^{d-1}\sum_{j=0}^{d-1}(\mathbbm{1}_C\otimes\tilde{K}_{ij})(\cdot)(\mathbbm{1}_C\otimes\tilde{K}_{ij}^\dagger)\nonumber\\
    &\quad+ s(1-s)\sum_{k=0}^{d-1}\sum_{\ell=0}^{d-1}(\mathbbm{1}_C\otimes\tilde{K}_{k\ell})(\cdot)(\mathbbm{1}_C\otimes\tilde{K}_{k\ell}^\dagger)+ (1-s)^2(\mathbbm{1}_C\otimes\mathbbm{1})(\cdot)(\mathbbm{1}_C\otimes\mathbbm{1})\\
    &= s^2 \mathdutchcal{S}[\Gamma,\Gamma](\cdot) + 2s(1-s)\,\mathcal{I}_C\otimes\Gamma(\cdot) + (1-s)^2\mathcal{I}_C\otimes\mathcal{I}(\cdot),
    \label{eq:S[G,G]_cvx_mixture}
\end{align}
where in the last line we used the Kraus representations for $\mathdutchcal{S}[\Gamma,\Gamma]$ [Eq.~(\ref{eq:S[Gamma,Gamma]_Kraus})] and $\Gamma$ [Eq.~(\ref{eq:Gamma_Kraus})]. From Eq.~(\ref{eq:S[G,G]_cvx_mixture}), the induced channel $\Lambda_\mathdutchcal{S}^{(G)}(\cdot) \equiv \mathdutchcal{S}[G,G](\rho_C\otimes\cdot) \in \mathcal{O}(A\rightarrow BC')$ can be written as
\begin{align}
    \Lambda_\mathdutchcal{S}^{(G)} = s^2\Lambda_\mathdutchcal{S}^{(\Gamma)} + 2s(1-s)\rho_C\otimes\Gamma + (1-s)^2\rho_C\otimes\mathcal{I},
    \label{eq:Lambda_S^G_cvx_mixture}
\end{align}
where $\Lambda_\mathdutchcal{S}^{(\Gamma)}(\cdot) \equiv \mathdutchcal{S}[\Gamma,\Gamma](\rho_C\otimes\cdot) \in \mathcal{O}(A\rightarrow BC')$ is the channel induced by $\mathdutchcal{S}[\Gamma,\Gamma]$. 
We then consider the control qubit in the state:
\begin{align}
    \rho_C\equiv \ket{\psi}\!\bra{\psi}_C \text{ with } \ket{\psi}_C\equiv \sqrt{\alpha}\ket{0} + \mathrm{e}^{\mathrm{i}\varphi}\sqrt{1-\alpha}\ket{1},
    \label{appeq:rho_C}
\end{align}
for $\alpha\in [0,1]$ and $\varphi\in [0, 2\pi]$. The thermal state of the control qubit is $\gamma_C \equiv \mathbbm{1}_C/2$. Since $\gamma_C$ can be rewritten as $\gamma_C = (\ket{\psi}\!\bra{\psi}_C + \ket{\psi^\perp}\!\bra{\psi^\perp}_C)/2$, where $\langle \psi|\psi^\perp \rangle_C = 0$, it is easy to check that $R_\mathrm{T}(\rho_C) = 1,\,\forall\, \alpha\in[0,1]$ and $\varphi\in [0, 2\pi]$.

In order to prove the upper bound of $R(\Lambda_\mathdutchcal{S}^{(G)}||\Gamma_{A\rightarrow BC'})$ in Eq.~(\ref{eq:switch_bound_upper}), we first prove the following lemma.
\begin{lemmaApp}\label{lem:R_Lambda_S^Gamma}
    Given the control qubit in Eq.~(\ref{appeq:rho_C}), the following equality holds
    \begin{align}
        R(\Lambda_\mathdutchcal{S}^{(\Gamma)}||\Gamma_{A\rightarrow BC'}) = 1, \quad\forall\, \alpha\in [0,1],\, \varphi\in [0, 2\pi].
    \end{align}
\end{lemmaApp}
\begin{proof}
    By Theorem~\ref{appthm:ER_R}, we have
    \begin{align}
        R(\Lambda_\mathdutchcal{S}^{(\Gamma)}||\Gamma_{A\rightarrow BC'}) = R_\mathrm{T}[\mathcal{I}\otimes\Lambda_\mathdutchcal{S}^{(\Gamma)}(\tilde{\gamma}_{AA'})],
        \label{eq:ER_Lambda_S^Gamma}
    \end{align}
    where $\tilde{\gamma}_{AA'} = \ket{\tilde{\gamma}}\!\bra{\tilde{\gamma}}_{AA'}$ with $\ket{\tilde{\gamma}}_{AA'} \equiv \sum_{i=0}^{d-1}\sqrt{g_i}\ket{i}_{A}\otimes\ket{i}_{A'}$. Define the output state 
    \begin{align}
        \rho^{\mathrm{out}}_{ABC'} &\equiv \mathcal{I}\otimes\Lambda_\mathdutchcal{S}^{(\Gamma)}(\tilde{\gamma}_{AA'})\\
        &= \mathcal{I}\otimes\mathdutchcal{S}[\Gamma,\Gamma](\rho_C\otimes\tilde{\gamma}_{AA'})\\
        &=\sum_{k,\ell=0}^{d-1}\sum_{m,n=0}^{d-1}(S_{k\ell mn}^{(\Gamma)}\otimes \mathbbm{1}_A)(\rho_C\otimes\tilde{\gamma}_{A'A})\left[\left(S_{k\ell mn}^{(\Gamma)}\right)^\dagger\otimes \mathbbm{1}_A\right]\\ &=\sum_{i,j=0}^{d-1}\sum_{k,\ell=0}^{d-1}\sum_{m,n=0}^{d-1}\sqrt{g_i}\sqrt{g_j} \left[\alpha\ket{0}\!\bra{0}_C\otimes \tilde{K}_{mn}\tilde{K}_{k\ell}(\ket{i}\!\bra{j})\tilde{K}^\dagger_{k\ell} \tilde{K}_{mn}^{\dagger}+(1-\alpha)\ket{1}\!\bra{1}_C\otimes  \tilde{K}_{k\ell}\tilde{K}_{mn}(\ket{i}\!\bra{j})\tilde{K}_{mn}^{\dagger} \tilde{K}_{k\ell}^{\dagger}\right. \nonumber\\&\left.\quad + \mathrm{e}^{-\mathrm{i}\varphi}\sqrt{\alpha(1-\alpha)}\ket{0}\!\bra{1}_C\otimes \tilde{K}_{mn}\tilde{K}_{k\ell}(\ket{i}\!\bra{j})\tilde{K}_{mn}^{\dagger} \tilde{K}_{k\ell}^{\dagger}+\mathrm{e}^{\mathrm{i}\varphi}\sqrt{\alpha(1-\alpha)}\ket{1}\!\bra{0}_C\otimes \tilde{K}_{k\ell}\tilde{K}_{mn}(\ket{i}\!\bra{j})\tilde{K}_{k\ell}^{\dagger} \tilde{K}_{mn}^{\dagger}\right]\otimes\ket{i}\!\bra{j}_A\\&=(\alpha\ket{0}\!\bra{0}_C+(1-\alpha)\ket{1}\!\bra{1}_C)\otimes\gamma_{B}\otimes\gamma_A + \sqrt{\alpha(1-\alpha)}(\mathrm{e}^{-\mathrm{i}\varphi}\ket{0}\!\bra{1}_C+\mathrm{e}^{\mathrm{i}\varphi}\ket{1}\!\bra{0}_C)\otimes\sigma_{BA},
        \label{eq:rho_ABC'_out}
    \end{align}
    where for the convenience of writing, we permute the order of the subsystems but use subscripts to avoid potential ambiguity. The state $\sigma_{BA}$ in Eq.~(\ref{eq:rho_ABC'_out}) is given by
    \begin{align}
        \sigma_{BA} \equiv (\gamma_B\otimes\mathbbm{1}_A)\tilde{\gamma}_{A'A}(\gamma_B\otimes \mathbbm{1}_A) = \sum_{i,j=0}^{d-1}g_i^{3/2}g_j^{3/2}\ket{i}\!\bra{j}_B\otimes\ket{i}\!\bra{j}_A,
        \label{eq:sigma_BA}
    \end{align}
    where we note that $\gamma_A = \gamma_B = \gamma_{A'} \equiv \gamma$. Recalling the definition in Eq.~(\ref{appeq:R_T_st_def}), we have
    \begin{align}
        1 + R_\mathrm{T}[\mathcal{I}\otimes\Lambda_\mathdutchcal{S}^{(\Gamma)}(\tilde{\gamma}_{AA'})] = \min\left\{\lambda | \rho_{ABC'}^\mathrm{out} \le \lambda \gamma_{C'}\otimes\gamma_B\otimes\gamma_A\right\}.
        \label{eq:rho_ABC'_out_ROA}
    \end{align}
    To find the minimal $\lambda$, we write $(\lambda \gamma_{C'}\otimes\gamma_B\otimes\gamma_A - \rho_{ABC'}^\mathrm{out})$ in the energy eigenbasis of $C'$ as a block matrix $M$ (from now we ignore the labels of systems):
    \begin{align}
        M = \begin{pmatrix}
            (\lambda/2 - \alpha)\gamma\otimes\gamma & -\mathrm{e}^{-\mathrm{i}\varphi}\sqrt{\alpha(1-\alpha)}\sigma \\
            -\mathrm{e}^{\mathrm{i}\varphi}\sqrt{\alpha(1-\alpha)}\sigma & (\lambda/2 + \alpha - 1)\gamma\otimes\gamma
        \end{pmatrix},\label{eq:M_matrix}
    \end{align}
    where $\lambda\ge 1$. When $\alpha = 0$ or $1$, it is straightforward that the minimal $\lambda$ ensuring $M\ge 0$ is 2, and therefore we have $R_\mathrm{T}[\mathcal{I}\otimes\Lambda_\mathdutchcal{S}^{(\Gamma)}(\tilde{\gamma}_{AA'})] = 2-1 = 1$. For $\alpha\in(0,1)$, to let $M\ge 0$, we need $\lambda/2 > \max\{\alpha, 1-\alpha\}$, which renders the top-left and bottom-right blocks in $M$ [Eq.~(\ref{eq:M_matrix})] be positive definite.
    By Schur complement (Theorem~1.12 in Ref.~\cite{zhang2006schur}), we know that
    \begin{align}
        M\ge 0 \Leftrightarrow (\lambda/2+\alpha-1)\gamma\otimes\gamma - \frac{\alpha(1-\alpha)}{\lambda/2-\alpha}\sigma(\gamma^{-1}\otimes\gamma^{-1})\sigma \ge 0.
    \end{align}
    The condition can be rewritten as:
    \begin{align}
        \sigma \le \frac{(\lambda/2+\alpha-1)(\lambda/2-\alpha)}{\alpha(1-\alpha)}\gamma\otimes\gamma,
        \label{eq:schur_complement_cond_M}
    \end{align}
    where we used the fact that $\sigma(\gamma^{-1}\otimes\gamma^{-1})\sigma = \sigma$ which can be seen from the following:
    \begin{align}
        (\gamma^{1/2}\otimes\gamma^{1/2})\tilde{\gamma}(\gamma^{1/2}\otimes\gamma^{1/2}) &= (\gamma^{1/2}\otimes\gamma^{1/2})(\gamma^{1/2}\otimes\mathbbm{1})\ket{\phi^+}\!\bra{\phi^+}(\gamma^{1/2}\otimes\mathbbm{1})(\gamma^{1/2}\otimes\gamma^{1/2}) \\
        &= (\gamma\otimes\mathbbm{1})(\mathbbm{1}\otimes\gamma^{1/2})\ket{\phi^+}\!\bra{\phi^+}(\mathbbm{1}\otimes\gamma^{1/2})(\gamma\otimes\mathbbm{1})\\
        &= (\gamma\otimes\mathbbm{1})\tilde{\gamma}(\gamma\otimes\mathbbm{1})\\
        &= \sigma,
        \label{eq:sigma_BA_alternative_expression}
    \end{align}
    where $\ket{\phi^+}\equiv\sum_{i=0}^{d-1}\ket{i}\otimes\ket{i}$ and we used Eq.~(\ref{eq:sigma_BA}) and the fact that $\ket{\tilde{\gamma}} = (\gamma^{1/2}\otimes\mathbbm{1})\ket{\phi^+} = (\mathbbm{1}\otimes\gamma^{1/2})\ket{\phi^+}$.

    Define the function $f(\lambda) \equiv \frac{(\lambda/2+\alpha-1)(\lambda/2-\alpha)}{\alpha(1-\alpha)}$. $f(\lambda)$ is monotonically increasing with $\lambda$. In the following, we will show that for Eq.~(\ref{eq:schur_complement_cond_M}) to hold, $f(\lambda) \ge 1$ such that the minimum $\lambda_{\min} = 2$ (note that $\lambda \ge 1$) satisfying $f(\lambda_{\min}) = 1$. We then conclude that $M\ge 0 \Leftrightarrow \lambda\ge \lambda_{\min} = 2$. By Eq.~(\ref{eq:rho_ABC'_out_ROA}), we have $R_\mathrm{T}[\mathcal{I}\otimes\Lambda_\mathdutchcal{S}^{(\Gamma)}(\tilde{\gamma}_{AA'})] = \lambda_{\min}-1 = 1,\,\forall\,\alpha\in(0,1),\,\varphi\in[0,2\pi]$. Since $R_\mathrm{T}[\mathcal{I}\otimes\Lambda_\mathdutchcal{S}^{(\Gamma)}(\tilde{\gamma}_{AA'})] = 1$ also for $\alpha = 0$ or $1$, $\varphi\in[0,2\pi]$,
    using Eq.~(\ref{eq:ER_Lambda_S^Gamma}), we therefore obtain $R(\Lambda_\mathdutchcal{S}^{(\Gamma)}||\Gamma_{A\rightarrow BC'}) = 1,\,\forall\,\alpha\in[0,1],\,\varphi\in[0,2\pi]$.

    Now the only thing left is to show that
    \begin{align}
        \min\{f|\sigma \le f\gamma\otimes\gamma\} = 1.
        \label{eq:minimisation_barely_positive}
    \end{align}
By Lemma~\ref{lem:barely_positive}, we note that Eq.~(\ref{eq:minimisation_barely_positive}) $\Leftrightarrow$ $\exists \ket{\phi},\, \bra{\phi}\gamma\otimes\gamma-\sigma\ket{\phi} = 0$ and $\gamma\otimes\gamma-\sigma \ge 0$. The (unnormalised) state $\ket{\phi}$ can be found as $\ket{\phi} \equiv (\gamma^{-1/2}\otimes\mathbbm{1})\ket{\phi^+}$. We have
    \begin{align}
        \bra{\phi}\gamma\otimes\gamma-\sigma\ket{\phi} &= \bra{\phi^+}(\gamma^{-1/2}\otimes\mathbbm{1})(\gamma\otimes\gamma-\sigma)(\gamma^{-1/2}\otimes\mathbbm{1})\ket{\phi^+}\\ &= \bra{\phi^+}\mathbbm{1}\otimes\gamma\ket{\phi^+} - \bra{\phi^+}\gamma\otimes\mathbbm{1}\ket{\phi^+}\bra{\phi^+}\gamma\otimes\mathbbm{1}\ket{\phi^+}\\ &= \mathrm{Tr}\{\gamma\}-\mathrm{Tr}\{\gamma\}\mathrm{Tr}\{\gamma\}\\ & = 0,
    \end{align}
    where in the second line we used the fact that $\sigma = (\gamma^{3/2}\otimes\mathbbm{1})\ket{\phi^+}\!\bra{\phi^+}(\gamma^{3/2}\otimes\mathbbm{1})$.

    Now we show that $\gamma\otimes\gamma-\sigma \ge 0$. Using  
    Eq.~(\ref{eq:sigma_BA_alternative_expression}), we have 
    \begin{align}
        \gamma\otimes\gamma - \sigma = (\gamma^{1/2}\otimes\gamma^{1/2})(\mathbbm{1}\otimes\mathbbm{1} - \tilde{\gamma})(\gamma^{1/2}\otimes\gamma^{1/2}).
    \end{align}
    Since $\mathbbm{1}\otimes\mathbbm{1} - \tilde{\gamma} \ge 0$ and $\gamma$ is full-rank, we obtain $\gamma\otimes\gamma - \sigma \ge 0$.
    Hence, Eq.~(\ref{eq:minimisation_barely_positive}) is proven and as we have argued above, this leads us to $R(\Lambda_\mathdutchcal{S}^{(\Gamma)}||\Gamma_{A\rightarrow BC'}) = 1,\,\forall\,\alpha\in[0,1],\,\varphi\in[0,2\pi]$.
\end{proof}
Now we prove the upper bound [Eq.~(\ref{eq:switch_bound_upper})] in Theorem~\ref{thm:switch_bounds}. 
According to Eq.~(\ref{eq:Lambda_S^G_cvx_mixture}), we can rewrite $\Lambda_\mathdutchcal{S}^{(G)}$ as
\begin{align}
    \Lambda_\mathdutchcal{S}^{(G)} &= s^2\Lambda_\mathdutchcal{S}^{(\Gamma)} + 2s(1-s)\rho_C\otimes\Gamma + (1-s)^2\rho_C\otimes\mathcal{I}\label{appeq:Lambda_S^G_expression}\\
    &= s^2\Lambda_\mathdutchcal{S}^{(\Gamma)} + (1-s^2)\rho_C\otimes G',
\end{align}
where $G'\in\mathrm{GPO}(A\rightarrow B)$ is defined as
\begin{align}
    G'\equiv \frac{2s(1-s)}{1-s^2}\Gamma + \frac{(1-s)^2}{1-s^2}\mathcal{I}.
\end{align}
We then have
\begin{align}
    R(\Lambda_\mathdutchcal{S}^{(G)}||\Gamma_{A\rightarrow BC'}) &\le s^2 R(\Lambda_\mathdutchcal{S}^{(\Gamma)}||\Gamma_{A\rightarrow BC'}) + (1-s^2) R(\rho_C\otimes G'||\Gamma_{A\rightarrow BC'})\\
    &= s^2 + (1-s^2)\left[\left(1+R(G'||\Gamma)\right)\left(1+R_\mathrm{T}(\rho_C)\right) - 1\right]\\
    &= s^2 + (1-s^2)\left[2\left(1+P_\mathrm{T}(G')\right) - 1\right]\\
    &=s^2 + (1-s^2)\left[1 + \frac{2(1-s)^2}{1-s^2}P_\mathrm{T}(\mathcal{I})\right]\\
    &= 1 + 2(1-s)^2 P_\mathrm{T}(\mathcal{I})\\
    &= R_\mathrm{T}(\rho_C) + 2 P_\mathrm{T}(G)^2/P_\mathrm{T}(\mathcal{I}), \label{appeq:R(Lambda_S^G||Gamma)_upper_bound}
\end{align}
where the inequality is from the convexity of $R(\Lambda||\Gamma)$, the second line is from Lemma~\ref{lem:R_Lambda_S^Gamma} and the multiplicity of $R(\Lambda||\Gamma)$ by considering $\rho_C$ as a state-preparing map, and in the third and the last line we used the fact that $R_\mathrm{T}(\rho_C) = 1,\,\forall\,\alpha\in[0,1],\,\varphi\in[0,2\pi]$. Moreover, in the forth line and the last line we used the following Lemma~\ref{lem:P_T(G)}: 
\begin{lemmaApp}\label{lem:P_T(G)}
    Given a GPO $G \equiv q\Gamma + (1-q)\mathcal{I}$ for $q\in[0,1]$, the following equality holds:
    \begin{align}
        P_\mathrm{T}(G) = (1-q)P_\mathrm{T}(\mathcal{I}).
    \end{align}
\end{lemmaApp}
\begin{proof}
    It is straightforward that the equality holds when $q = 0$ or $1$. For $q\in(0,1)$, 
    define $\lambda^* \equiv 1+P_\mathrm{T}(G) \equiv 1+R(G||\Gamma)$. From the definition of $R(G||\Gamma)$ [Eq.~(\ref{appeq:R_def})], we have $G\le \lambda^*\Gamma \Rightarrow \mathcal{I} \le \frac{\lambda^* - q}{1-q}\Gamma \Rightarrow (1-q)P_\mathrm{T}(\mathcal{I})\le \lambda^*-1\equiv P_\mathrm{T}(G)$. On the other hand, due to the convexity of $R(G||\Gamma)$, we have $P_\mathrm{T}(G)\equiv R(G||\Gamma) \le (1-q)P_\mathrm{T}(\mathcal{I})$. Hence, $P_\mathrm{T}(G) = (1-q)P_\mathrm{T}(\mathcal{I})$.
\end{proof}
We remark that the upper bound Eq.~(\ref{appeq:R(Lambda_S^G||Gamma)_upper_bound}) is saturated as either $\alpha$ or $s = 0$ or $1$.  To see this, we note that when $\alpha = 0$ or $1$, the switch is fully off and according to Eq.~(\ref{appeq:Lambda_S^G_expression}), the induced channel $\Lambda_\mathdutchcal{S}^{(G)} = s^2\rho_C\otimes\Gamma + 2s(1-s)\rho_C\otimes\Gamma + (1-s)^2\rho_C\otimes \mathcal{I} = \rho_C\otimes[(1-(1-s)^2)\Gamma + (1-s)^2\mathcal{I}]$. Thus, using the multiplicity of $R$, $R_\mathrm{T}(\rho_C) = 1$ and Lemma.~\ref{lem:P_T(G)}, we have $R(\Lambda_\mathdutchcal{S}^{(G)}||\Gamma_{A\rightarrow BC'}) = R_\mathrm{T}(\rho_C) + 2 P_\mathrm{T}(G)^2/P_\mathrm{T}(\mathcal{I})$. The upper bound in Eq.~(\ref{appeq:R(Lambda_S^G||Gamma)_upper_bound}) is achieved as $\alpha = 0$ or $1$. When $s=0$, by Eq.~(\ref{appeq:Lambda_S^G_expression}), we have $\Lambda_\mathdutchcal{S}^{(G)} = \rho_C\otimes \mathcal{I}$ with $R(\Lambda_\mathdutchcal{S}^{(G)}||\Gamma_{A\rightarrow BC'}) = R_\mathrm{T}(\rho_C) + 2P_\mathrm{T}(\mathcal{I})$. When $s=1$, $\Lambda_\mathdutchcal{S}^{(G)} = \Lambda_\mathdutchcal{S}^{(\Gamma)}$ such that $R(\Lambda_\mathdutchcal{S}^{(G)}||\Gamma_{A\rightarrow BC'}) = R_\mathrm{T}(\rho_C)$ by Lemma~\ref{lem:R_Lambda_S^Gamma}. Hence, the upper bound in Eq.~(\ref{appeq:R(Lambda_S^G||Gamma)_upper_bound}) is saturated when $s=0$ or $1$.

The proof of the upper bound [Eq.~(\ref{eq:switch_bound_upper})] in Theorem~\ref{thm:switch_bounds} is completed.

\subsection{Derivation of Eq.~(\ref{eq:R_T(Lambda_S)}) in Theorem~\ref{thm:switch_bounds}}\label{app:R_T(Lambda_S)}
In this subsection, we derive the analytical expression of $R_\mathrm{T}(\Lambda_\mathdutchcal{S}^{(G)})$ in Eq.~(\ref{eq:R_T(Lambda_S)}). By Theorem~\ref{appthm:ER_ROA}, we know that $R_\mathrm{T}(\Lambda_\mathdutchcal{S}^{(G)}) = R_\mathrm{T}[\Lambda_\mathdutchcal{S}^{(G)}(\gamma_A)]$.
To find $\Lambda_\mathdutchcal{S}^{(G)}(\gamma_A)$, we 
recall Eq.~(\ref{eq:Lambda_S^G_cvx_mixture}) and obtain
\begin{align}
    \Lambda_\mathdutchcal{S}^{(G)}(\gamma_A) &= s^2\Lambda_\mathdutchcal{S}^{(\Gamma)}(\gamma_A) + 2s(1-s)\rho_C\otimes\Gamma(\gamma_A) + (1-s)^2\rho_C\otimes\gamma_B\\
    &= s^2\,\mathrm{Tr}_A\!\left\{\rho_{ABC'}^\mathrm{out}\right\} + (1-s^2)\rho_C\otimes\gamma_B,
    \label{appeq:Lambda_S^G(gamma_A)}
\end{align}
where $\rho_{ABC'}^\mathrm{out}$ given in Eq.~(\ref{eq:rho_ABC'_out}) is the output of $\mathcal{I}\otimes\Lambda_\mathdutchcal{S}^{(\Gamma)}(\tilde{\gamma}_{AA'})$. The last line follows from the fact that $\mathrm{Tr}_A\!\left\{\rho_{ABC'}^\mathrm{out}\right\} = \mathrm{Tr}_A\left\{\mathcal{I}\otimes\Lambda_\mathdutchcal{S}^{(\Gamma)}(\tilde{\gamma}_{AA'})\right\} = \Lambda_\mathdutchcal{S}^{(\Gamma)}\left(\mathrm{Tr}_A\left\{\tilde{\gamma}_{AA'}\right\}\right) = \Lambda_\mathdutchcal{S}^{(\Gamma)}(\gamma_{A'}) = \Lambda_\mathdutchcal{S}^{(\Gamma)}(\gamma_A)$ (note that $A'$ is a copy of $A$). Given the $\rho_C$ in Eq.~(\ref{appeq:rho_C}) and $\lambda\ge 1$, we write $[\lambda\gamma_{C'}\otimes\gamma_B - \Lambda_\mathdutchcal{S}^{(G)}(\gamma_A)]$ in the energy eigenbasis of $C'$ as a block matrix $N$ (from now we ignore the labels of systems):
\begin{align}
    N &= \lambda\begin{pmatrix}
        \gamma/2 & 0 \\
        0 & \gamma/2
    \end{pmatrix}
    -s^2\begin{pmatrix}
        \alpha\gamma & \mathrm{e}^{-\mathrm{i}\varphi}\sqrt{\alpha(1-\alpha)}\gamma^3 \\
        \mathrm{e}^{\mathrm{i}\varphi}\sqrt{\alpha(1-\alpha)}\gamma^3 & (1-\alpha)\gamma
    \end{pmatrix} - (1-s^2)\begin{pmatrix}
        \alpha\gamma & \mathrm{e}^{-\mathrm{i}\varphi}\sqrt{\alpha(1-\alpha)}\gamma \\
        \mathrm{e}^{\mathrm{i}\varphi}\sqrt{\alpha(1-\alpha)}\gamma & (1-\alpha)\gamma
    \end{pmatrix}\\
    &= \begin{pmatrix}
        (\lambda/2 - \alpha)\gamma & -\mathrm{e}^{-\mathrm{i}\varphi}\sqrt{\alpha(1-\alpha)}[s^2\gamma^3 + (1-s^2)\gamma] \\
        -\mathrm{e}^{\mathrm{i}\varphi}\sqrt{\alpha(1-\alpha)}[s^2\gamma^3 + (1-s^2)\gamma] & (\lambda/2 + \alpha-1) \gamma
    \end{pmatrix}. \label{eq:N_matrix}
\end{align}
When $\alpha = 0$ or $1$, it is easy to find that the minimal $\lambda$ ensuring $N\ge 0$ is 2, and thus $R_\mathrm{T}(\Lambda_\mathdutchcal{S}^{(G)}) = R_\mathrm{T}[\Lambda_\mathdutchcal{S}^{(G)}(\gamma_A)] = 2-1 = 1$. For $\alpha\in(0,1)$, 
to let $N\ge 0$, we need $\lambda/2 > \max\{\alpha, 1-\alpha\}$, which renders the top-left and bottom-right blocks in $N$ [Eq.~(\ref{eq:N_matrix})] be positive definite.
By Schur complement (Theorem~1.12 in Ref.~\cite{zhang2006schur}), we have
\begin{align}
    N\ge 0 \Leftrightarrow (\lambda/2+\alpha-1)\gamma - \frac{\alpha(1-\alpha)}{\lambda/2-\alpha}[s^2\gamma^3 + (1-s^2)\gamma]\gamma^{-1}[s^2\gamma^3 + (1-s^2)\gamma] \ge 0.
\end{align}
The condition can be rewritten as
\begin{align}
    \gamma \ge \frac{\alpha(1-\alpha)}{(\lambda/2+\alpha-1)(\lambda/2-\alpha)}[s^2\gamma^3 + (1-s^2)\gamma]\gamma^{-1}[s^2\gamma^3 + (1-s^2)\gamma].
\end{align}
Since both sides are diagonal in the same basis, it is equivalent to require
\begin{align}
    &g_i \ge \frac{\alpha(1-\alpha)}{(\lambda/2+\alpha-1)(\lambda/2-\alpha)}[s^2g_i^3 + (1-s^2)g_i]g_i^{-1}[s^2g_i^3 + (1-s^2)g_i],\quad \forall\, i=0,1,\dots,d-1, \\
    \Rightarrow \quad & \lambda \ge 1 + \sqrt{1-4(1-g_i^2)[2-(1-g_i^2)s^2]s^2\alpha(1-\alpha)}, \quad \forall\, i=0,1,\dots,d-1,\\
    \Rightarrow \quad & \lambda \ge 1 + \sqrt{1-4(1-g_{\max}^2)[2-(1-g_{\max}^2)s^2]s^2\alpha(1-\alpha)},
\end{align}
where $g_{\max}\equiv \max_i g_i$ is the maximal diagonal of $\gamma$. In the second line we have used $\alpha\in(0,1)$, $\lambda\ge 1$ and $\lambda/2 > \max\{\alpha, 1-\alpha\}$ to neglect other ranges of $\lambda$ and the last line follows from that the right-hand side of the second line is a monotonically increasing function of $g_i$. 
Therefore, we have 
\begin{align}
    R_\mathrm{T}(\Lambda_\mathdutchcal{S}^{(G)}) = R_\mathrm{T}[\Lambda_\mathdutchcal{S}^{(G)}(\gamma_A)] = \lambda_{\min} - 1 = \sqrt{1-4(1-g_{\max}^2)[2-(1-g_{\max}^2)s^2]s^2\alpha(1-\alpha)},
    \label{appeq:R_T(Lambda_S^G)}
\end{align}
$\forall\, \alpha\in[0,1],\,\varphi\in[0,2\pi],\,s\in[0,1]$, where we have included the case of $\alpha = 0$ or 1 because the expression correctly gives $R_\mathrm{T}(\Lambda_\mathdutchcal{S}^{(G)}) = 1$ as $\alpha = 0$ or $1$. Eq.~(\ref{eq:R_T(Lambda_S)}) is hence derived.
\subsection{Generalisation of Theorem~\ref{thm:switch_bounds} for a general \texorpdfstring{$\rho_C$}{RC}}\label{app:general_rho_C}
In previous subsections, Theorem~\ref{thm:switch_bounds} is proved for a pure state $\rho_C$ defined in Eq.~(\ref{appeq:rho_C}). In this subsection, we extend Theorem~\ref{thm:switch_bounds} for a fully general $\rho_C$, which can always be written as
\begin{align}
    \rho_C\equiv r\ket{\psi}\!\bra{\psi}_C + (1-r)\gamma_C \text{ with } \ket{\psi}_C\equiv \sqrt{\alpha}\ket{0} + \mathrm{e}^{\mathrm{i}\varphi}\sqrt{1-\alpha}\ket{1},
    \label{appeq:rho_C_gen}
\end{align}
for $r\in[0,1]$, $\alpha\in [0,1]$ and $\varphi\in [0, 2\pi]$. Since $\gamma_C \equiv \mathbbm{1}_C/2$ is the centre of Bloch sphere while $\ket{\psi}\!\bra{\psi}_C$ parametrises all points on the Bloch sphere, $\rho_C$ as a convex combination of $\gamma_C$ and $\ket{\psi}\!\bra{\psi}_C$ covers all possible quantum states of the control qubit $C$. Following the same idea of the proof of Lemma~\ref{lem:P_T(G)}, it is easy to show that $R_\mathrm{T}(\rho_C) = rR_\mathrm{T}(\ket{\psi}\!\bra{\psi}_C) = r,\,\forall\,\alpha\in[0,1],\,\varphi\in[0,2\pi]$. In this case,
\begin{align}
    \Lambda_\mathdutchcal{S}^{(\Gamma)}(\cdot)\equiv\mathdutchcal{S}[\Gamma,\Gamma](\rho_C\otimes\cdot) = r\mathdutchcal{S}[\Gamma,\Gamma](\ket{\psi}\!\bra{\psi}_C\otimes\cdot) + (1-r)\mathdutchcal{S}[\Gamma,\Gamma](\gamma_C\otimes\cdot)
    = r\tilde{\Lambda}_\mathdutchcal{S}^{(\Gamma)}(\cdot) + (1-r)\gamma_{C}\otimes\Gamma(\cdot),
\end{align}
where $\tilde{\Lambda}_\mathdutchcal{S}^{(\Gamma)}$ denotes $\Lambda_\mathdutchcal{S}^{(\Gamma)}$ in Section~\ref{app:proof_switch_upper_bound}, which is the induced channel with the pure control qubit $\ket{\psi}\!\bra{\psi}_C$. By Eq.~(\ref{eq:Lambda_S^G_cvx_mixture}), we can write $\Lambda_\mathdutchcal{S}^{(G)}(\cdot) \equiv \mathdutchcal{S}[G,G](\rho_C\otimes\cdot)$ as
\begin{align}
    \Lambda_\mathdutchcal{S}^{(G)}&= s^2\Lambda_\mathdutchcal{S}^{(\Gamma)} + 2s(1-s)\rho_C\otimes\Gamma + (1-s)^2\rho_C\otimes\mathcal{I}\\
    &= rs^2 \tilde{\Lambda}_\mathdutchcal{S}^{(\Gamma)} + (1-r)s^2\gamma_{C}\otimes\Gamma + (1-s^2)\rho_C\otimes G',
    \label{appeq:Lambda_S^G_expression_gen}
\end{align}
where $G'\in\mathrm{GPO}(A\rightarrow B)$ is defined as
\begin{align}
    G' \equiv \frac{2s(1-s)}{1-s^2}\Gamma + \frac{(1-s)^2}{1-s^2}\mathcal{I}.
\end{align}
We then have
\begin{align}
    R(\Lambda_\mathdutchcal{S}^{(G)}||\Gamma_{A\rightarrow BC'}) &\le rs^2 R(\tilde{\Lambda}_\mathdutchcal{S}^{(\Gamma)}||\Gamma_{A\rightarrow BC'}) + (1-r)s^2 R(\gamma_{C}\otimes\Gamma||\Gamma_{A\rightarrow BC'}) + (1-s^2) R(\rho_C\otimes G'||\Gamma_{A\rightarrow BC'})\\
    &= rs^2 + (1-s^2)\left[\left(1+R(G'||\Gamma)\right)\left(1+R_\mathrm{T}(\rho_C)\right) - 1\right]\\
    &= rs^2 + (1-s^2)\left[\left(1+P_\mathrm{T}(G')\right)(1+r) - 1\right]\\
    &=rs^2 + (1-s^2)\left[r+\frac{(1+r)(1-s)^2}{1-s^2}P_\mathrm{T}(\mathcal{I})\right]\\
    &= r + (1+r)(1-s)^2 P_\mathrm{T}(\mathcal{I})\\
    &= R_\mathrm{T}(\rho_C) + [1+R_\mathrm{T}(\rho_C)] P_\mathrm{T}(G)^2/P_\mathrm{T}(\mathcal{I}), \label{appeq:R(Lambda_S^G||Gamma)_upper_bound_gen}
\end{align}
where the inequality is from the convexity of $R(\Lambda||\Gamma)$, the second line is from Lemma~\ref{lem:R_Lambda_S^Gamma} and the multiplicity of $R(\Lambda||\Gamma)$ by considering $\rho_C$ as a state-preparing map, and in the third and the last line we used the fact that $R_\mathrm{T}(\rho_C) = r,\,\forall\,\alpha\in[0,1],\,\varphi\in[0,2\pi]$. Moreover, in the forth line and the last line we used Lemma~\ref{lem:P_T(G)}. By similar arguments around Eq.~(\ref{appeq:R(Lambda_S^G||Gamma)_upper_bound}), we find that the upper bound Eq.~(\ref{appeq:R(Lambda_S^G||Gamma)_upper_bound_gen}) is saturated as either $\alpha$ or $s = 0$ or 1.

To derive the analytical expression of $R_\mathrm{T}(\Lambda_\mathdutchcal{S}^{(G)})$ for a general $\rho_C$, by Theorem~\ref{appthm:ER_ROA}, we know that $R_\mathrm{T}(\Lambda_\mathdutchcal{S}^{(G)}) = R_\mathrm{T}[\Lambda_\mathdutchcal{S}^{(G)}(\gamma_A)]$.
Using Eqs.~(\ref{appeq:rho_C_gen}) and (\ref{appeq:Lambda_S^G_expression_gen}), we rewrite $\Lambda_\mathdutchcal{S}^{(G)}(\gamma_A)$ as
\begin{align}
    \Lambda_\mathdutchcal{S}^{(G)}(\gamma_A) &= rs^2\tilde{\Lambda}_\mathdutchcal{S}^{(\Gamma)}(\gamma_A) + (1-r)s^2\gamma_C\otimes\gamma_B + (1-s^2)\rho_C\otimes\gamma_B\\
    &= rs^2\tilde{\Lambda}_\mathdutchcal{S}^{(\Gamma)}(\gamma_A) + r(1-s^2)\ket{\psi}\!\bra{\psi}_C\otimes\gamma_B + (1-r)\gamma_C\otimes\gamma_B\\
    &= r\tilde{\Lambda}_\mathdutchcal{S}^{(G)}(\gamma_A) + (1-r)\gamma_C\otimes\gamma_B,
\end{align}
where 
\begin{align}
    \tilde{\Lambda}_\mathdutchcal{S}^{(G)}(\gamma_A) \equiv s^2\tilde{\Lambda}_\mathdutchcal{S}^{(\Gamma)}(\gamma_A) +(1-s^2)\ket{\psi}\!\bra{\psi}_C\otimes\gamma_B.
\end{align}
Following the same idea of the proof of Lemma~\ref{lem:P_T(G)}, it is straightforward to show that $R_\mathrm{T}[\Lambda_\mathdutchcal{S}^{(G)}(\gamma_A)] = rR_\mathrm{T}[\tilde{\Lambda}_\mathdutchcal{S}^{(G)}(\gamma_A)]$. Noticing that $\tilde{\Lambda}_\mathdutchcal{S}^{(G)}(\gamma_A)$ denotes $\Lambda_\mathdutchcal{S}^{(G)}(\gamma_A)$ in Eq.~(\ref{appeq:Lambda_S^G(gamma_A)}) in Section~\ref{app:R_T(Lambda_S)}, corresponding to the induced channel with the pure control qubit $\ket{\psi}\!\bra{\psi}_C$. Using the expression of $R_\mathrm{T}[\tilde{\Lambda}_\mathdutchcal{S}^{(G)}(\gamma_A)]$ in Eq.~(\ref{appeq:R_T(Lambda_S^G)}), we have
\begin{align}
    R_\mathrm{T}[\Lambda_\mathdutchcal{S}^{(G)}(\gamma_A)] = rR_\mathrm{T}[\tilde{\Lambda}_\mathdutchcal{S}^{(G)}(\gamma_A)] = r\sqrt{1-4(1-g_{\max}^2)[2-(1-g_{\max}^2)s^2]s^2\alpha(1-\alpha)},
    \label{appeq:R_T(Lambda_S^G)_gen}
\end{align}
$\forall\, r\in[0,1],\,\alpha\in[0,1],\,\varphi\in[0,2\pi],\,s\in[0,1]$.

Comparing the results with a general $\rho_C$ [Eqs.~(\ref{appeq:R(Lambda_S^G||Gamma)_upper_bound_gen}) and~(\ref{appeq:R_T(Lambda_S^G)_gen})] with those with a pure $\rho_C$ [Eqs.~(\ref{appeq:R(Lambda_S^G||Gamma)_upper_bound}) and~(\ref{appeq:R_T(Lambda_S^G)})], we note that the mixedness of $\rho_C$ only changes our results by a pre-factor $r$, which relates to the purity of $\rho_C$ as $\mathcal{P}(\rho_C)\equiv \mathrm{Tr}\{\rho_C^2\} = (1+r^2)/2$.

\section{Resource analysis for coherent control}\label{app:coherent_control}
An alternative way to activate the signalling ability of a quantum channel by multiple using it is the coherent control of quantum channels~\cite{abbott2020communication}. In this setup, the implementation of one of two quantum channels is conditioned on the state of a control qubit, either $\ket{0}$ or $\ket{1}$. By preparing the control qubit in the superposition state $\ket{+}$, the two channel implementations can be coherently superposed. This differs from the case of the quantum switch, where it is the causal order of the two channels—not their implementations—that is placed in a quantum superposition.
Due to the similar structures, the performances of coherent control and quantum switch have been compared in the literature~\cite{loizeau2020channel,abbott2020communication}. In this section, we perform a similar resource analysis as we did for the quantum switch and show how the two objects are similar.
\begin{figure}
    \centering
    \includegraphics[width=0.6\linewidth]{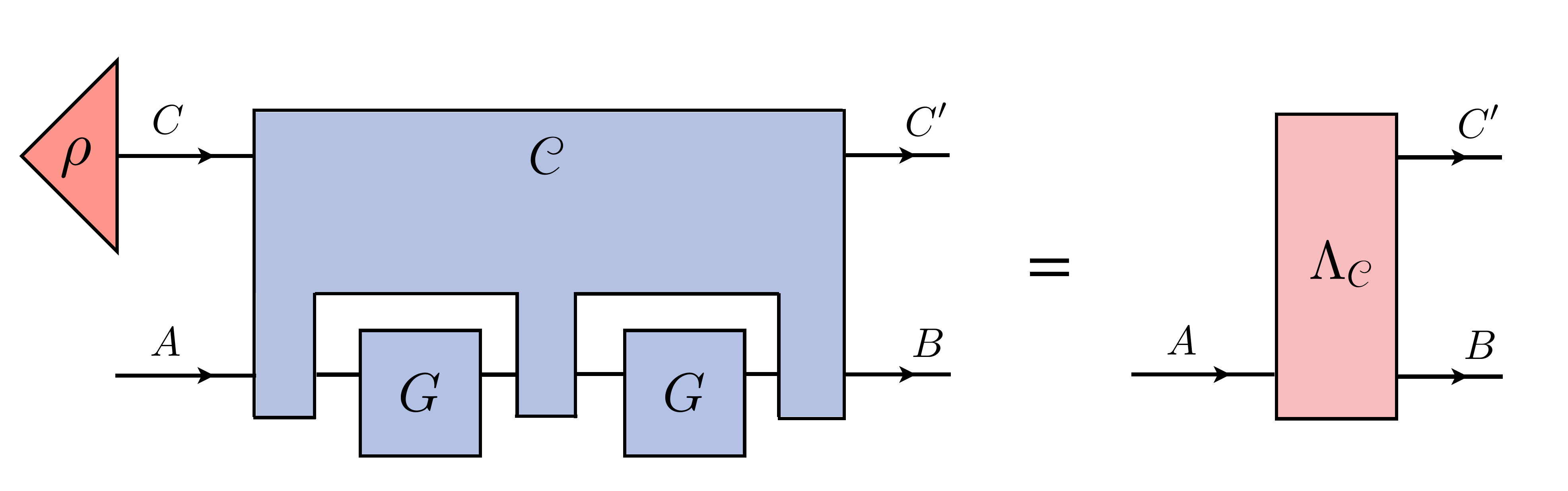}
    \caption{A coherent control of two GPOs with a control qubit $\rho_C$ can be considered as a channel simulation of $\Lambda_\mathdutchcal{C}$.}
    \label{fig:coherent_control}
\end{figure}

For a fair comparison, we choose the same control qubit $\rho_C$ [Eq.~(\ref{appeq:rho_C})], systems $A$ and $B$, and the signalling GPO [Eq.~(\ref{appeq:signalling_G})] as those for the quantum switch. The channel induced by the coherent control is denoted as $\Lambda_\mathdutchcal{C} \equiv \mathdutchcal{C}[G,G](\rho_C\otimes\cdot) \in \mathcal{O}(A\rightarrow BC')$ (See Fig.~\ref{fig:coherent_control}). We have the following theorem:
\begin{theoremApp}\label{appthm:coh_cntrl_bounds}
    Given the coherent control with the two input channels as $G$ in Eq.~(\ref{appeq:signalling_G}), the following bounds hold
    \begin{align}
        R(\Lambda_\mathdutchcal{C}||\Gamma) &\le R_\mathrm{T}(\rho_C) + 2P_\mathrm{T}(G) + \sqrt{\alpha(1-\alpha)s(1-s)}f(\gamma), \\
        R(\Lambda_\mathdutchcal{C}||\Gamma) &\ge R_\mathrm{S}(\Lambda_\mathdutchcal{C}) \ge (g_{\min}^4/2)\left[R(\Lambda_\mathdutchcal{C}||\Gamma) - R_\mathrm{T}(\Lambda_\mathdutchcal{C})\right]^2,
        \label{appeq:coh_ctrl_bound_lower}
    \end{align}
    where $R_\mathrm{T}(\rho_C) = 1,\,\forall\,\alpha\in[0,1],\,\varphi\in[0,2\pi]$ and $f(\gamma)>0$ is a function of $\gamma$ (See Section~\ref{app:derivation_of_R_T(Lambda_C)}). When $s=1$, i.e.,~the input GPOs are $G\equiv \Gamma$, 
    \begin{align}
        R_\mathrm{T}(\Lambda_\mathdutchcal{C}) = \sqrt{(1-2\alpha)^2 + 4d^{-2}\alpha(1-\alpha)},\,\, \forall\, \alpha\in[0,1],\,\varphi\in[0,2\pi].
        \label{eq:R_T(Lambda_C)}
    \end{align}
    The upper bound of $R(\Lambda_\mathdutchcal{C}||\Gamma)$ is saturated for $s=0$ or $1$.
\end{theoremApp}
When $s = 1$, by Theorem~\ref{appthm:coh_cntrl_bounds}, we have $R(\Lambda_\mathdutchcal{C}||\Gamma) = R_\mathrm{T}(\rho_C)$. Eq.~(\ref{appeq:coh_ctrl_bound_lower}) reduces to
\begin{align}
    R_\mathrm{T}(\rho_C) \ge R_\mathrm{S}(\Lambda_\mathdutchcal{C}) \ge (g_{\min}^4/2)[R_\mathrm{T}(\rho_C) - R_\mathrm{T}(\Lambda_\mathdutchcal{C})]^2.
\end{align}
$R_\mathrm{T}(\Lambda_\mathdutchcal{C})$ and $R_\mathrm{S}(\Lambda_\mathdutchcal{C})$ thus exhibit a similar trade-off as in the case of quantum switch.
\subsection{Proof of the upper bound in Theorem~\ref{appthm:coh_cntrl_bounds}}
Given a general quantum channel $\Omega\in\mathcal{O}(A\rightarrow B)$ with the Kraus representation $\Omega(\cdot) \equiv \sum_{n=0}^{N_{\mathrm{Kraus}}-1} K_n(\cdot) K_n^\dagger$ where $N_\mathrm{Kraus}$ is the number of Kraus operators, the Kraus representation of the coherent control $\mathdutchcal{C}[\Omega,\Omega]\in\mathcal{O}(AC\rightarrow BC')$, where $C'$ is a copy of $C$, can be written as $\mathdutchcal{C}[\Omega,\Omega](\cdot) \equiv \sum_{m=0}^{N_{\mathrm{Kraus}}-1}\sum_{n=0}^{N_{\mathrm{Kraus}}-1} C_{mn}(\cdot)C_{mn}^\dagger$, where the Kraus operator $C_{mn}$ is given by
\begin{align}
    C_{mn} \equiv \frac{1}{\sqrt{N_\mathrm{Kraus}}}\left(\ket{0}\!\bra{0}_C\otimes K_m + \ket{1}\!\bra{1}_C\otimes K_n\right),
    \label{eq:coherent_control_Kraus_general}
\end{align}
where the normalised factor $1/\sqrt{N_\mathrm{Kraus}}$ ensures $\sum_{m=0}^{N_{\mathrm{Kraus}}-1}\sum_{n=0}^{N_{\mathrm{Kraus}}-1} C_{mn}^\dagger C_{mn} = \mathbbm{1}$. As pointed in Ref.~\cite{abbott2020communication}, the action of coherent control can depend on the implementation of the controlled channels. By considering the Stinespring's dilation of the channel $\Omega(\cdot) = \mathrm{Tr}_{E}\{U(\cdot\otimes\ket{\eta}\!\bra{\eta}_E)U^\dagger\}$, the Kraus operator $C_{mn}$ in Eq.~(\ref{eq:coherent_control_Kraus_general}) corresponds to the case when $\ket{\eta}_E = \sum_{j=0}^{N_{\text{Kraus}}-1}1/\sqrt{N_{\text{Kraus}}}\ket{j}_E$ where $\{\ket{j}_E\}_j$ is an orthonormal basis of the environment $E$ (See the discussion around Eq.~(5) in Ref.~\cite{abbott2020communication}). We choose this specific form for a fair comparison with the quantum switch, whose Kraus operators are given in Eq.~(\ref{eq:swith_Kraus_general}).

Consider the case when systems $A$ and $B$ have the same thermal state, i.e.,~$\gamma_A = \gamma_B \equiv \gamma$ with the dimension of $d$. The coherent control of the completely thermalising channel $\Gamma\in\mathrm{GPO}(A\rightarrow B)$ is $\mathdutchcal{C}[\Gamma,\Gamma]$, whose Kraus operator is given by
\begin{align}
    C_{ijk\ell}^{(\Gamma)} \equiv \frac{1}{d}\left(\ket{0}\!\bra{0}_C\otimes\tilde{K}_{ij} + \ket{1}\!\bra{1}_C\otimes\tilde{K}_{k\ell}\right),
\end{align}
where $\tilde{K}_{ij}\equiv \sqrt{g_j}\ket{j}\!\bra{i}$ is the Kraus operator of $\Gamma$ [See Eq.~(\ref{eq:Gamma_Kraus})]. Thus,
\begin{align}
    \mathdutchcal{C}[\Gamma,\Gamma](\cdot) \equiv \sum_{i=0}^{d-1}\sum_{j=0}^{d-1}\sum_{k=0}^{d-1}\sum_{\ell=0}^{d-1}C_{ijk\ell}^{(\Gamma)}(\cdot)\left(C_{ijk\ell}^{(\Gamma)}\right)^\dagger.
    \label{eq:C[Gamma,Gamma]_Kraus}
\end{align}
Now consider the signalling GPO $G$ in Eq.~(\ref{appeq:signalling_G}). 
\begin{align}
    G(\cdot) &\equiv s\Gamma + (1-s)\mathcal{I} \\
    &= s\sum_{i=0}^{d-1}\sum_{j=0}^{d-1} \tilde{K}_{ij}(\cdot)\tilde{K}_{ij}^\dagger + (1-s)\mathbbm{1}(\cdot)\mathbbm{1}\\
    &\equiv \sum_{i=0}^{d-1}\sum_{j=0}^{d-1}K_{ij}(\cdot)K_{ij}^\dagger + K_{dd}(\cdot)K_{dd}^\dagger,
\end{align}
where $s\in[0,1]$ and the Kraus operators are given by 
\begin{align}
    K_{ij} &\equiv \sqrt{s}\tilde{K}_{ij} = \sqrt{s}\sqrt{g_j}\ket{j}\!\bra{i},\quad\forall\, i,j = 0,1,\dots,d-1, \\
    K_{dd} &\equiv \sqrt{1-s} \mathbbm{1}.
\end{align}
Now we consider the case when $s\in (0,1)$. The channel $G$ thus has $d^2+1$ Kraus operators. 
According to Eq.~(\ref{eq:coherent_control_Kraus_general}), the Kraus operator of the coherent control $\mathdutchcal{C}[G,G]$ is given by
\begin{align}
    C_{ijk\ell}^{(G)} &\equiv \frac{1}{\sqrt{d^2+1}}\left(\ket{0}\!\bra{0}_C\otimes K_{ij} + \ket{1}\!\bra{1}_C\otimes K_{k\ell}\right)\\
    &= \frac{\sqrt{s}}{\sqrt{d^2+1}}\left(\ket{0}\!\bra{0}_C\otimes\tilde{K}_{ij} + \ket{1}\!\bra{1}_C\otimes\tilde{K}_{k\ell} \right)\\
    &= \frac{d\sqrt{s}}{\sqrt{d^2+1}}\,C_{ijk\ell}^{(\Gamma)},\\
    C_{ijdd}^{(G)} &\equiv \frac{1}{\sqrt{d^2+1}}\left(\ket{0}\!\bra{0}_C\otimes K_{ij} + \ket{1}\!\bra{1}_C \otimes K_{dd}\right)\\
    &= \frac{1}{\sqrt{d^2+1}}\left(\ket{0}\!\bra{0}_C\otimes\sqrt{s}\tilde{K}_{ij} + \ket{1}\!\bra{1}_C\otimes \sqrt{1-s}\mathbbm{1}\right),\\
    C_{ddk\ell}^{(G)} &\equiv \frac{1}{\sqrt{d^2+1}}\left(\ket{0}\!\bra{0}_C\otimes K_{dd} + \ket{1}\!\bra{1}_C \otimes K_{k\ell}\right)\\
    &= \frac{1}{\sqrt{d^2+1}}\left(\ket{0}\!\bra{0}_C\otimes\sqrt{1-s}\mathbbm{1} + \ket{1}\!\bra{1}_C\otimes\sqrt{s}\tilde{K}_{k\ell}\right),\\
    C_{dddd}^{(G)} &\equiv \frac{1}{\sqrt{d^2+1}}\left(\ket{0}\!\bra{0}_C \otimes K_{dd} + \ket{1}\!\bra{1}_C\otimes K_{dd}\right)\\
    &= \frac{\sqrt{1-s}}{\sqrt{d^2+1}}\mathbbm{1}_C\otimes\mathbbm{1},
\end{align}
$\forall\, i,j,k,\ell = 0,1,\dots,d-1$. The factor $1/\sqrt{d^2+1}$ comes from the fact that $G$ has $d^2+1$ Kraus operators. The channel $\mathdutchcal{C}[G,G]$ can therefore be written as
\begin{align}
    \mathdutchcal{C}[G,G](\cdot) &= \sum_{i=0}^{d-1}\sum_{j=0}^{d-1}\sum_{k=0}^{d-1}\sum_{\ell=0}^{d-1} C_{ijk\ell}^{(G)}(\cdot)\left(C_{ijk\ell}^{(G)}\right)^\dagger + \sum_{i=0}^{d-1}\sum_{j=0}^{d-1} C_{ijdd}^{(G)}(\cdot) \left(C_{ijdd}^{(G)}\right)^\dagger + \sum_{k=0}^{d-1}\sum_{\ell=0}^{d-1} C_{ddk\ell}^{(G)}(\cdot) \left(C_{ddk\ell}^{(G)}\right)^\dagger + C_{dddd}^{(G)}(\cdot) \left(C_{dddd}^{(G)}\right)^\dagger \\
    &= \frac{d^2 s}{d^2+1}\sum_{i=0}^{d-1}\sum_{j=0}^{d-1}\sum_{k=0}^{d-1}\sum_{\ell=0}^{d-1}C_{ijk\ell}^{(\Gamma)}(\cdot)\left(C_{ijk\ell}^{(\Gamma)}\right)^\dagger + \sum_{i=0}^{d-1}\sum_{j=0}^{d-1} C_{ijdd}^{(G)}(\cdot) \left(C_{ijdd}^{(G)}\right)^\dagger + \sum_{k=0}^{d-1}\sum_{\ell=0}^{d-1} C_{ddk\ell}^{(G)}(\cdot) \left(C_{ddk\ell}^{(G)}\right)^\dagger \nonumber\\
    &\quad + \frac{1-s}{d^2+1}(\mathbbm{1}_C\otimes\mathbbm{1})(\cdot)(\mathbbm{1}_C\otimes\mathbbm{1})\\
    &= \frac{d^2s}{d^2+1}\mathdutchcal{C}[\Gamma,\Gamma](\cdot) + \frac{s+(1-s)d^2}{d^2+1}\mathcal{M}(\cdot) + \frac{1-s}{d^2+1}\mathcal{I}_C\otimes\mathcal{I}(\cdot),
    \label{eq:C[Gamma,Gamma]_cvx_mixture}
\end{align}
where in the last line we used the Kraus representation for $\mathdutchcal{C}[\Gamma,\Gamma]$ [Eq.~(\ref{eq:C[Gamma,Gamma]_Kraus})] and $\Gamma$ [Eq.~(\ref{eq:Gamma_Kraus})], and
the channel $\mathcal{M}\in\mathcal{O}(AC\rightarrow BC')$ is defined as
\begin{align}
    \mathcal{M}(\cdot) := \frac{d^2+1}{s+(1-s)d^2}\left[\sum_{i=0}^{d-1}\sum_{j=0}^{d-1} C_{ijdd}^{(G)}(\cdot) \left(C_{ijdd}^{(G)}\right)^\dagger + \sum_{k=0}^{d-1}\sum_{\ell=0}^{d-1} C_{ddk\ell}^{(G)}(\cdot) \left(C_{ddk\ell}^{(G)}\right)^\dagger\right].
    \label{eq:M_channel}
\end{align}
$\mathcal{M}$ is trace-preserving because
\begin{align}
    &\frac{d^2+1}{s+(1-s)d^2}\left[\sum_{i=0}^{d-1}\sum_{j=0}^{d-1} \left(C_{ijdd}^{(G)}\right)^\dagger C_{ijdd}^{(G)} + \sum_{k=0}^{d-1}\sum_{\ell=0}^{d-1} \left(C_{ddk\ell}^{(G)}\right)^\dagger C_{ddk\ell}^{(G)}\right] \nonumber\\
    =\, \,&\frac{1}{s+(1-s)d^2}\left[\ket{0}\!\bra{0}_C\otimes s\sum_{i,j=0}^{d-1}
    \tilde{K}_{ij}^\dagger\tilde{K}_{ij} + \ket{1}\!\bra{1}_C\otimes (1-s)d^2\mathbbm{1} + \ket{0}\!\bra{0}_C\otimes(1-s)d^2\mathbbm{1} + \ket{1}\!\bra{1}_C\otimes s\sum_{k,\ell=0}^{d-1}\tilde{K}_{k\ell}^\dagger\tilde{K}_{k\ell}\right]\\
    =\,\,&\frac{1}{s+(1-s)d^2}\left[\ket{0}\!\bra{0}_C\otimes s\mathbbm{1} + \ket{1}\!\bra{1}_C\otimes (1-s)d^2\mathbbm{1} + \ket{0}\!\bra{0}_C\otimes(1-s)d^2\mathbbm{1} + \ket{1}\!\bra{1}_C\otimes s\mathbbm{1}\right]\\
    =\,\,&\mathbbm{1}_C\otimes\mathbbm{1}.
\end{align}
From Eq.~(\ref{eq:C[Gamma,Gamma]_cvx_mixture}), the induced channel $\Lambda_\mathdutchcal{C}^{(G)}(\cdot) \equiv \mathdutchcal{C}[G,G](\rho_C\otimes\cdot)\in\mathcal{O}(A\rightarrow BC')$ can be written as
\begin{align}
    \Lambda_\mathdutchcal{C}^{(G)} = \frac{d^2s}{d^2+1}\Lambda_\mathdutchcal{C}^{(\Gamma)} + \frac{s+(1-s)d^2}{d^2+1}\mathcal{N} + \frac{1-s}{d^2+1}\rho_C\otimes\mathcal{I},
    \label{eq:Lambda_C^G_cvx_mixture}
\end{align}
where $\Lambda_\mathdutchcal{C}^{(\Gamma)}(\cdot) \equiv \mathdutchcal{C}[\Gamma,\Gamma](\rho_C\otimes\cdot)$ and $\mathcal{N}(\cdot)\equiv \mathcal{M}(\rho_C\otimes\cdot)$.

\begin{lemmaApp}\label{lem:R_Lambda_C^Gamma}
    Given the control qubit in Eq.~(\ref{appeq:rho_C}), the following equality holds
    \begin{align}
        R(\Lambda_\mathdutchcal{C}^{(\Gamma)}||\Gamma_{A\rightarrow BC'}) = 1,\quad \forall\,\alpha\in[0,1],\,\varphi\in[0,2\pi]
    \end{align}
\end{lemmaApp}
\begin{proof}
    By Theorem~\ref{appthm:ER_R}, we have
    \begin{align}
        R(\Lambda_\mathdutchcal{C}^{(\Gamma)}||\Gamma_{A\rightarrow BC'}) = R_\mathrm{T}[\mathcal{I}\otimes\Lambda_\mathdutchcal{C}^{(\Gamma)}(\tilde{\gamma}_{AA'})],
        \label{eq:ER_Lambda_C_Gamma}
    \end{align}
    where $\tilde{\gamma}_{AA'} = \ket{\tilde{\gamma}}\!\bra{\tilde{\gamma}}_{AA'}$ with $\ket{\tilde{\gamma}}_{AA'} \equiv \sum_{i=0}^{d-1}\sqrt{g_i}\ket{i}_A\otimes\ket{i}_{A'}$. Define the output state
    \begin{align}
        \rho_{ABC'}^{\mathrm{out}} &\equiv \mathcal{I}\otimes\Lambda_\mathdutchcal{C}^{(\Gamma)}(\tilde{\gamma}_{AA'})\\
        &= \mathcal{I}\otimes\mathdutchcal{C}[\Gamma,\Gamma](\rho_C\otimes\tilde{\gamma}_{AA'})\\
        &= \sum_{k,\ell=0}^{d-1}\sum_{m,n=0}^{d-1}(C_{k\ell mn}^{(\Gamma)}\otimes \mathbbm{1}_A)(\rho_C\otimes\tilde{\gamma}_{A'A})\left[\left(C_{k\ell mn}^{(\Gamma)}\right)^\dagger\otimes \mathbbm{1}_A\right]\\
        &= \sum_{i,j=0}^{d-1}\sum_{k,\ell=0}^{d-1}\sum_{m,n=0}^{d-1}\frac{\sqrt{g_i}\sqrt{g_j}}{d^2}\left[\alpha\ket{0}\!\bra{0}_C\otimes\tilde{K}_{k\ell}(\ket{i}\!\bra{j})\tilde{K}_{k\ell}^\dagger+(1-\alpha)\ket{1}\!\bra{1}_C\otimes\tilde{K}_{mn}(\ket{i}\!\bra{j})\tilde{K}_{mn}^\dagger \right. \nonumber\\
        &\left. \quad + \mathrm{e}^{-\mathrm{i}\varphi}\sqrt{\alpha(1-\alpha)}\ket{0}\!\bra{1}_C\otimes\tilde{K}_{k\ell}(\ket{i}\!\bra{j})\tilde{K}_{mn}^\dagger +  \mathrm{e}^{\mathrm{i}\varphi}\sqrt{\alpha(1-\alpha)}\ket{1}\!\bra{0}_C\otimes\tilde{K}_{mn}(\ket{i}\!\bra{j})\tilde{K}_{k\ell}^\dagger\right]\otimes\ket{i}\!\bra{j}_A\\
        &= (\alpha\ket{0}\!\bra{0}_C + (1-\alpha)\ket{1}\!\bra{1}_C)\otimes\gamma_B\otimes\gamma_A + \frac{\sqrt{\alpha(1-\alpha)}}{d^2}( \mathrm{e}^{-\mathrm{i}\varphi}\ket{0}\!\bra{1}_C +  \mathrm{e}^{\mathrm{i}\varphi}\ket{1}\!\bra{0}_C)\otimes\hat{\gamma}_B\otimes\hat{\gamma}_A,
        \label{eq:rho_ABC'_out_coh_cntrl}
    \end{align}
    where we note that $\gamma_A = \gamma_B \equiv \gamma$ and the state $\hat{\gamma}$ is defined as
    \begin{align}
        \hat{\gamma} \equiv \sum_{i,j=0}^{d-1}\sqrt{g_i}\sqrt{g_j}\ket{i}\!\bra{j}.
        \label{eq:gamma_hat}
    \end{align}
    Recalling the definition in Eq.~(\ref{appeq:R_T_st_def}), we have
    \begin{align}
        1 + R_\mathrm{T}[\mathcal{I}\otimes\Lambda_\mathdutchcal{C}^{(\Gamma)}(\tilde{\gamma}_{AA'})] = \min\left\{\lambda | \rho_{ABC'}^\mathrm{out}\le \lambda\gamma_{C'}\otimes\gamma_B\otimes\gamma_A\right\}. 
        \label{eq:rho_ABC'_out_ROA_coh_cntrl}
    \end{align}
    To find the minimal $\lambda$, we write $(\lambda\gamma_{C'}\otimes\gamma_B\otimes\gamma_A - \rho_{ABC'}^{\mathrm{out}})$ in the energy eigenbasis of $C'$ as a block matrix $M$ (from now we ignore the labels of systems):
    \begin{align}
        M = \begin{pmatrix}
            (\lambda/2-\alpha)\gamma\otimes\gamma & -d^{-2} \mathrm{e}^{-\mathrm{i}\varphi}\sqrt{\alpha(1-\alpha)}\hat{\gamma}\otimes\hat{\gamma} \\
            -d^{-2} \mathrm{e}^{\mathrm{i}\varphi}\sqrt{\alpha(1-\alpha)}\hat{\gamma}\otimes\hat{\gamma} & (\lambda/2+\alpha-1)\gamma\otimes\gamma
        \end{pmatrix},
        \label{eq:M_matrix_coh_cntrl}
    \end{align}
    with $\lambda\ge 1$. When $\alpha = 0$ or $1$, it is clear that the minimal $\lambda$ guaranteeing $M\ge 0$ is $2$, and thus $R_\mathrm{T}[\mathcal{I}\otimes\Lambda_\mathdutchcal{C}^{(\Gamma)}(\tilde{\gamma}_{AA'})] = 2-1 = 1$. For $\alpha\in(0,1)$, to let $M\ge 0$, we need $\lambda/2 > \max\{\alpha,1-\alpha\}$, which makes the top-left and bottom-right blocks in $M$ [Eq.~(\ref{eq:M_matrix_coh_cntrl})] be positive definite. By Schur complement (Theorem~1.12 in Ref.~\cite{zhang2006schur}), we have
    \begin{align}
        M\ge 0 \Leftrightarrow (\lambda/2+\alpha - 1)\gamma\otimes\gamma - \frac{\alpha(1-\alpha)}{d^4(\lambda/2-\alpha)}(\hat{\gamma}\otimes\hat{\gamma})(\gamma^{-1}\otimes\gamma^{-1})(\hat{\gamma}\otimes\hat{\gamma})\ge 0.
    \end{align}
    This condition is equivalent to the following:
    \begin{align}
        \hat{\gamma}\otimes\hat{\gamma} \le d^2\frac{(\lambda/2+\alpha-1)(\lambda/2-\alpha)}{\alpha(1-\alpha)}\gamma\otimes\gamma,
        \label{eq:schur_complement_M_coh_cntrl}
    \end{align}
    which comes from
    \begin{align}
        \hat{\gamma}\gamma^{-1}\hat{\gamma} = \sum_{i,j=0}^{d-1}\sum_{k,\ell=0}^{d-1}\sqrt{g_i}\sqrt{g_j}\sqrt{g_k}\sqrt{g_\ell}\ket{i}\!\bra{j}g_k^{-1}\ket{k}\!\bra{\ell}
        = \sum_{i,\ell=0}^{d-1}\sum_{j=0}^{d-1}\sqrt{g_i}\sqrt{g_\ell}\ket{i}\!\bra{\ell}
        = d\hat{\gamma}.
        \label{eq:hat_gamma_gamma(-1)_hat_gamma}
    \end{align}
    Define the function $f(\lambda) \equiv \frac{\alpha(1-\alpha)}{(\lambda/2+\alpha-1)(\lambda/2-\alpha)}$. $f(\lambda)$ is monotonically increasing with $\lambda$. In the following, we will show that for Eq.~(\ref{eq:schur_complement_M_coh_cntrl}) to hold, $f(\lambda)\ge 1$ such that the minimum $\lambda_{\min} = 2$ (note that $\lambda \ge 1$) satisfying $f(\lambda_{\min})$ = 1. We then conclude that $M\ge 0 \Leftrightarrow \lambda\ge \lambda_{\min}=2$. By Eq.~(\ref{eq:rho_ABC'_out_ROA_coh_cntrl}), we have $R_\mathrm{T}[\mathcal{I}\otimes\Lambda_\mathdutchcal{C}^{(\Gamma)}(\tilde{\gamma}_{AA'})] = \lambda_{\min}-1 = 1,\,\forall\,\alpha\in(0,1)$. Since $R_\mathrm{T}[\mathcal{I}\otimes\Lambda_\mathdutchcal{C}^{(\Gamma)}(\tilde{\gamma}_{AA'})] = 1$ also for $\alpha = 0$ or $1$, using Eq.~(\ref{eq:ER_Lambda_C_Gamma}), we therefore obtain $R(\Lambda_\mathdutchcal{C}^{(\Gamma)}||\Gamma_{A\rightarrow BC'}) = 1,\, \forall \, \alpha\in [0,1],\,\varphi\in[0,2\pi]$.
    
    Now the only thing left is to show that
    \begin{align}
        \min\{f|\hat{\gamma}\otimes\hat{\gamma} \le f\cdot d^2(\gamma\otimes\gamma)\} = 1.
        \label{eq:minimisation_barely_positive_coh_cntrl}
    \end{align}
    By Lemma~\ref{lem:barely_positive}, we have: Eq.~(\ref{eq:minimisation_barely_positive_coh_cntrl}) $\Leftrightarrow$ $\exists \ket{\phi},\, \bra{\phi}d^2 (\gamma\otimes\gamma)-\hat{\gamma}\otimes\hat{\gamma}\ket{\phi} = 0$ and $d^2(\gamma\otimes\gamma)-\hat{\gamma}\otimes\hat{\gamma} \ge 0$. Such a state $\ket{\phi}$ can be chosen as $\ket{\phi}\equiv \sum_{i,j=0}^{d-1}g_i^{-1/2}g_j^{-1/2}\ket{i}\otimes\ket{j}$ so that
    \begin{align}
        \bra{\phi}d^2 (\gamma\otimes\gamma) \ket{\phi} &= d^2\left(\sum_{k,\ell=0}^{d-1}g_k^{-1/2}g_\ell^{-1/2}\bra{k}\otimes\bra{\ell}\right)\left(\sum_{i,j=0}^{d-1}g_i^{1/2}g_j^{1/2}\ket{i}\otimes\ket{j}\right)
        = d^2\sum_{i,j=0}^{d-1} 1 = d^4, \\
        \bra{\phi} \hat{\gamma}\otimes\hat{\gamma} \ket{\phi} &= \bra{\phi}\sum_{i,j=0}^{d-1}\sum_{k,\ell=0}^{d-1}\sum_{m,n=0}^{d-1}g_k^{1/2}g_\ell^{1/2}g_m^{1/2}g_n^{1/2}g_i^{-1/2}g_j^{-1/2}(\ket{k}\!\bra{\ell}\otimes\ket{m}\!\bra{n})(\ket{i}\otimes\ket{j})\\
        &= \bra{\phi}\left(\sum_{i,j=0}^{d-1} 1\right)\sum_{k,m=0}^{d-1}g_k^{1/2}g_m^{1/2}\ket{k}\otimes\ket{m}\\
        &= d^2\left(\sum_{i,j=0}^{d-1}g_i^{-1/2}g_j^{-1/2}\bra{i}\otimes\bra{j}\right)\left(\sum_{k,m=0}^{d-1}g_k^{1/2}g_m^{1/2}\ket{k}\otimes\ket{m}\right)\\
        &= d^2\sum_{i,j=0}^{d-1} 1\\
        &= d^4, \\
        \Rightarrow \quad \bra{\phi}d^2 (\gamma\otimes\gamma) &-\hat{\gamma}\otimes\hat{\gamma}\ket{\phi} = d^4 - d^4 = 0.
    \end{align}
    To show that $d^2(\gamma\otimes\gamma) -\hat{\gamma}\otimes\hat{\gamma} \ge 0$, we note that 
    \begin{align}
        d^2(\gamma\otimes\gamma) - \hat{\gamma}\otimes\hat{\gamma} = (\gamma^{1/2}\otimes\gamma^{1/2})(d\mathbbm{1}\otimes d\mathbbm{1} - O_d\otimes O_d)(\gamma^{1/2}\otimes\gamma^{1/2}),
    \end{align}
    where $O_d\equiv \sum_{i,j=0}^{d-1}\ket{i}\!\bra{j}$. Since $d\mathbbm{1}\otimes d\mathbbm{1} - O_d\otimes O_d\ge 0$ and $\gamma$ is full-rank, we have $d^2(\gamma\otimes\gamma) -\hat{\gamma}\otimes\hat{\gamma} \ge 0$. Therefore, Eq.~(\ref{eq:minimisation_barely_positive_coh_cntrl}) is proven and as we have argued above, this leads us to $R(\Lambda_\mathdutchcal{C}^{(\Gamma)}||\Gamma_{A\rightarrow BC'}) = 1,\,\forall\, \alpha\in[0,1],\,\varphi\in[0,2\pi]$.
\end{proof}

    For later use, we calculate the action of $\mathcal{N}$ on a state can be obtained from $\mathcal{M}$ in Eq.~(\ref{eq:M_channel}) by using $\mathcal{N}(\cdot)\equiv \mathcal{M}(\rho_C\otimes\cdot)$. We can find
    \begin{align}
        \mathcal{N}(\cdot) &:= \frac{1}{s+(1-s)d^2}\Bigg[\left(\alpha\ket{0}\!\bra{0}_C+(1-\alpha)\ket{1}\!\bra{1}_C\right)\otimes\left(s\Gamma + (1-s)d^2\mathcal{I}\right)(\cdot) \nonumber\\
        & \quad + \sqrt{\alpha(1-\alpha)}\sqrt{s(1-s)}\left(\mathrm{e}^{-\mathrm{i}\varphi}\ket{0}\!\bra{1}_C + \mathrm{e}^{\mathrm{i}\varphi}\ket{1}\!\bra{0}_C\right)\otimes \sum_{k,\ell=0}^{d-1}\left(\tilde{K}_{k\ell}(\cdot) + (\cdot)\tilde{K}_{k\ell}^\dagger\right)\Bigg].
    \end{align}
    The output state of $\mathcal{I}\otimes\mathcal{N}$ with the input $\tilde{\gamma}_{AA'}$ is thus given as
    \begin{align}
        \rho_{ABC'}^{\mathrm{out}(\mathcal{N})} &\equiv \mathcal{I}\otimes\mathcal{N}(\tilde{\gamma}_{AA'})\\
        &=\left(\alpha\ket{0}\!\bra{0}_C+(1-\alpha)\ket{1}\!\bra{1}_C\right)\otimes\frac{s\gamma\otimes\gamma + (1-s)d^2\tilde{\gamma}}{s+(1-s)d^2}\nonumber\\
        &\quad + \frac{\sqrt{\alpha(1-\alpha)s(1-s)}}{s+(1-s)d^2}\left(\mathrm{e}^{-\mathrm{i}\varphi}\ket{0}\!\bra{1}_C + \mathrm{e}^{\mathrm{i}\varphi}\ket{1}\!\bra{0}_C\right)\otimes\tau ,
    \end{align}
    where $\tau$ is defined as
    \begin{align}
        \tau &\equiv |\hat{\gamma}\rangle\!\rangle\!\bra{\tilde{\gamma}} + \ket{\tilde{\gamma}}\!\langle\!\langle\hat{\gamma}|
        = \sum_{i=0}^{d-1}\sum_{k,\ell=0}^{d-1}\sqrt{g_i}\sqrt{g_k}\sqrt{g_\ell}\left(\ket{k}\!\bra{i}\otimes\ket{\ell}\!\bra{i} + \ket{i}\!\bra{k}\otimes\ket{i}\!\bra{\ell}\right),
    \end{align}
    where $|\hat{\gamma}\rangle\!\rangle \equiv \sum_{k,\ell=0}^{d-1}\sqrt{g_k}\sqrt{g_\ell}\ket{k}\otimes\ket{\ell}$ is the vectorisation of $\hat{\gamma}$ defined in Eq.~(\ref{eq:gamma_hat}). By Theorem~\ref{appthm:ER_R}, we have
    \begin{align}
        R(\mathcal{N}||\Gamma) = R_\mathrm{T}\left[\mathcal{I}\otimes\mathcal{N}(\tilde{\gamma}_{AA'})\right] &= R_\mathrm{T}(\rho_{ABC'}^{\mathrm{out}(\mathcal{N})}) \\
        &\le R_\mathrm{T}\left[\left(\alpha\ket{0}\!\bra{0}_C+(1-\alpha)\ket{1}\!\bra{1}_C\right)\otimes\frac{s\gamma\otimes\gamma + (1-s)d^2\tilde{\gamma}}{s+(1-s)d^2}\right] \nonumber \\
        &\quad + 1 + R_\mathrm{T}\left[\frac{\sqrt{\alpha(1-\alpha)s(1-s)}}{s+(1-s)d^2}\left(\mathrm{e}^{-\mathrm{i}\varphi}\ket{0}\!\bra{1}_C + \mathrm{e}^{\mathrm{i}\varphi}\ket{1}\!\bra{0}_C\right)\otimes\tau\right]\\
        &= [1+R_\mathrm{T}(\alpha\ket{0}\!\bra{0}_C+(1-\alpha)\ket{1}\!\bra{1}_C)]\left[1+\frac{(1-s)d^2}{s+(1-s)d^2} R_\mathrm{T}(\tilde{\gamma})\right]-1 \nonumber \\
        &\quad + \frac{\sqrt{\alpha(1-\alpha)s(1-s)}}{s+(1-s)d^2}\nonumber\\
        &\quad \times \left[1+R_\mathrm{T}((\mathrm{e}^{-\mathrm{i}\varphi}\ketbra{0}{1}_C + \mathrm{e}^{\mathrm{i}\varphi}\ketbra{1}{0}_C)\otimes\tau)\right]\\
        &\le \frac{2(1-s)d^2}{s+(1-s)d^2}P_\mathrm{T}(\mathcal{I})+1 + \frac{\sqrt{\alpha(1-\alpha)s(1-s)}}{s+(1-s)d^2}\left[1+R_\mathrm{T}(\chi)\right],
        \label{eq:R(N||Gamma)_bound}
    \end{align}
    where the first inequality comes from the fact that $R_\mathrm{T}(\rho+\sigma) \le 1 + R_\mathrm{T}(\rho) + R_\mathrm{T}(\sigma)$, the second line follows from the equality $1+R_\mathrm{T}(c\rho) = c[1+R_\mathrm{T}(\rho)]$ for a constant $c\ge 0$, and
    the second inequality comes from the fact that $1 + R_\mathrm{T}(\alpha\ket{0}\!\bra{0}_C+(1-\alpha)\ket{1}\!\bra{1}_C) = 2\max\{\alpha,1-\alpha\} \le 2$, $\chi\equiv(\mathrm{e}^{-\mathrm{i}\varphi}\ket{0}\!\bra{1}_C+\mathrm{e}^{\mathrm{i}\varphi}\ket{1}\!\bra{0}_C)\otimes\tau$ and $P_\mathrm{T}(\mathcal{I}) = R_\mathrm{T}(\tilde{\gamma})$ (Corollary~\ref{appcorol:all_purfications_R_T}).

    According to Eq.~(\ref{eq:Lambda_C^G_cvx_mixture}), we have
    \begin{align}
        R(\Lambda_\mathdutchcal{C}^{(G)}||\Gamma) &\le \frac{d^2 s}{d^2+1}R(\Lambda_\mathdutchcal{C}^{(\Gamma)}||\Gamma) + \frac{s+(1-s)d^2}{d^2+1}R(\mathcal{N}||\Gamma) + \frac{1-s}{d^2+1}\left[\left(1+R_\mathrm{T}(\rho_C)\right)\left(1+P_\mathrm{T}(\mathcal{I})\right)-1\right]\\
        &\le \frac{d^2 s}{d^2+1} + \frac{2(1-s)d^2}{d^2+1}P_\mathrm{T}(\mathcal{I}) + \frac{s+(1-s)d^2}{d^2+1} + \frac{\sqrt{\alpha(1-\alpha)}\sqrt{s(1-s)}}{d^2+1}\left[1+R_\mathrm{T}(\chi) \right] + \frac{1-s}{d^2+1}\left(2P_\mathrm{T}(\mathcal{I}) + 1\right)\\
        &= \frac{d^2 s}{d^2+1} + \frac{2d^2}{d^2+1}P_\mathrm{T}(G) + \frac{s+(1-s)d^2}{d^2+1} + \frac{\sqrt{\alpha(1-\alpha)}\sqrt{s(1-s)}}{d^2+1}\left[1+R_\mathrm{T}(\chi)\right] + \frac{2}{d^2+1}P_\mathrm{T}(G) + \frac{1-s}{d^2+1}\\
        &= 1 + 2P_\mathrm{T}(G) + \frac{\sqrt{\alpha(1-\alpha)}\sqrt{s(1-s)}}{d^2+1}[1+R_\mathrm{T}(\chi)]\\
        &= R_\mathrm{T}(\rho_C) + 2P_\mathrm{T}(G) + \frac{\sqrt{\alpha(1-\alpha)}\sqrt{s(1-s)}}{d^2+1}[1+R_\mathrm{T}(\chi)],
        \label{eq:R(Lambda_C^G||Gamma)_upper_bound}
    \end{align}
    where the first inequality comes from the convexity and the multiplicity of $R(\Lambda||\Gamma)$, the second inequality is obtained by substituting the upper bound of $R(\mathcal{N}||\Gamma)$ in Eq.~(\ref{eq:R(N||Gamma)_bound}), the third line follows the Lemma.~\ref{lem:P_T(G)} and the last line used the fact that $R_\mathrm{T}(\rho_C) = 1,\,\forall\,\alpha\in[0,1]$. Note that the upper bound in Eq.~(\ref{eq:R(Lambda_C^G||Gamma)_upper_bound}) is derived for the case when $s\in (0,1)$.
    When $s = 0$, we have $\Lambda_{\mathdutchcal{C}}^{(G)} = \rho_C\otimes\mathcal{I}$ with $R(\Lambda_\mathdutchcal{C}^{(G)}||\Gamma) = R_\mathrm{T}(\rho_C) + 2P_\mathrm{T}(\mathcal{I})$.
    When $s = 1$, 
    $\Lambda_\mathdutchcal{C}^{(G)}$ reduces to $\Lambda_\mathdutchcal{C}^{(\Gamma)}$ and
    by Lemma~\ref{lem:R_Lambda_C^Gamma}, $R(\Lambda_\mathdutchcal{C}^{(\Gamma)}||\Gamma) = 1 = R_\mathrm{T}(\rho_C)$. Therefore, Eq.~(\ref{eq:R(Lambda_C^G||Gamma)_upper_bound}) is an upper bound of $R(\Lambda_\mathdutchcal{C}^{(G)}||\Gamma)$ for $s\in [0,1]$ and is saturated at $s=0$ or $1$. We eventually obtain the upper bound in Theorem~\ref{appthm:coh_cntrl_bounds}:
    \begin{align}
        R(\Lambda_\mathdutchcal{C}^{(\Gamma)}||\Gamma) \le R_\mathrm{T}(\rho_C) + 2P_\mathrm{T}(G) + \sqrt{\alpha(1-\alpha)s(1-s)}f(\gamma), \quad\forall\, \alpha, s\in[0,1],\,\varphi\in[0,2\pi].
    \end{align}
    where $f(\gamma)\equiv [1+R_\mathrm{T}(\chi)]/(d^2+1)$.

\subsection{Derivation of Eq.~(\ref{eq:R_T(Lambda_C)}) in Theorem~\ref{appthm:coh_cntrl_bounds}}\label{app:derivation_of_R_T(Lambda_C)}
In this subsection, we derive the analytical expression of $R_\mathrm{T}(\Lambda_\mathdutchcal{C}^{(\Gamma)})$ in Eq.~(\ref{eq:R_T(Lambda_C)}). By Theorem~\ref{appthm:ER_ROA}, we know that $R_\mathrm{T}(\Lambda_\mathdutchcal{C}^{(\Gamma)}) = R_\mathrm{T}[\Lambda_\mathdutchcal{C}^{(\Gamma)}(\gamma_A)]$. To find $\Lambda_\mathdutchcal{C}^{(\Gamma)}(\gamma_A)$, we recall $\rho_{ABC'}^{\mathrm{out}}$ in Eq.~(\ref{eq:rho_ABC'_out_coh_cntrl}) and have $\mathrm{Tr}_A\{\rho_{ABC'}^\mathrm{out}\} = \mathrm{Tr}_A\left\{\mathcal{I}\otimes\Lambda_\mathdutchcal{C}^{(\Gamma)}(\tilde{\gamma}_{AA'})\right\} = \Lambda_\mathdutchcal{C}^{(\Gamma)}(\mathrm{Tr}_A\{\tilde{\gamma}_{AA'}\}) = \Lambda_\mathdutchcal{C}^{(\Gamma)}(\gamma_{A'}) = \Lambda_\mathdutchcal{C}^{(\Gamma)}(\gamma_A)$ (note that $A'$ is a copy of $A$). Given $\lambda\ge 1$, we write $[\lambda\gamma_{C'}\otimes\gamma_B - \Lambda_\mathdutchcal{C}^{(\Gamma)}(\gamma_A)]$ in the energy eigenbasis of $C$ as a block matrix $N$ (from now we ignore the labels of systems):
\begin{align}
    N &= \lambda\begin{pmatrix}
        \gamma/2 & 0 \\
        0 & \gamma/2
    \end{pmatrix} - \begin{pmatrix}
        \alpha\gamma & d^{-2}\mathrm{e}^{-\mathrm{i}\varphi}\sqrt{\alpha(1-\alpha)}\hat{\gamma} \\
        d^{-2}\mathrm{e}^{\mathrm{i}\varphi}\sqrt{\alpha(1-\alpha)}\hat{\gamma} & (1-\alpha)\gamma
    \end{pmatrix}\\
    &= \begin{pmatrix}
        (\lambda/2 - \alpha)\gamma & -d^{-2}\mathrm{e}^{-\mathrm{i}\varphi}\sqrt{\alpha(1-\alpha)}\hat{\gamma} \\
        -d^{-2}\mathrm{e}^{\mathrm{i}\varphi}\sqrt{\alpha(1-\alpha)}\hat{\gamma} & (\lambda/2+\alpha - 1)\gamma
    \end{pmatrix}.
    \label{eq:N_matrix_coh_cntrl}
\end{align}
When $\alpha = 0$ or $1$, it is easy to find that the minimal $\lambda$ ensuring $N\ge 0$ is $2$, and thus $R_\mathrm{T}(\Lambda_\mathdutchcal{C}^{(\Gamma)}) = R_\mathrm{T}[\Lambda_\mathdutchcal{C}^{(\Gamma)}(\gamma_A)] = 2-1 = 1$. For $\alpha\in(0,1)$, to let $N\ge 0$, we need $\lambda/2 > \max\{\alpha, 1-\alpha\}$, which renders the top-left and bottom-right blocks in $N$ [Eq.~(\ref{eq:N_matrix_coh_cntrl})] be positive definite. By Schur complement (Theorem~1.12 in Ref.~\cite{zhang2006schur}), we have
\begin{align}
    N\ge 0 \Leftrightarrow (\lambda/2 + \alpha - 1)\gamma - \frac{\alpha(1-\alpha)}{d^4(\lambda/2 - \alpha)}\,\hat{\gamma}\gamma^{-1}\hat{\gamma}\ge 0.
\end{align}
Using the equality $\hat{\gamma}\gamma^{-1}\hat{\gamma} = d\hat{\gamma}$ [Eq.~(\ref{eq:hat_gamma_gamma(-1)_hat_gamma})], the condition can be rewritten as
\begin{align}
    \gamma \ge \frac{\alpha(1-\alpha)}{d^3(\lambda/2+\alpha - 1)(\lambda/2 - \alpha)}\hat{\gamma}.
    \label{appeq:gamma_>=_f*hat_gamma}
\end{align}
Now we show that $\min\{f|\hat{\gamma}\le f\cdot d\gamma\} = 1$. By Lemma~\ref{lem:barely_positive}, it is equivalent to show that $\exists\,\ket{\phi},\, \bra{\phi}d\gamma - \hat{\gamma}\ket{\phi} = 0$ and $d\gamma - \hat{\gamma} \ge 0$. Such a state $\ket{\phi}$ can be chosen as $\ket{\phi} \equiv \sum_{i=0}^{d-1}g_i^{-1/2}\ket{i}$ so that
\begin{align}
    \bra{\phi}d\gamma \ket{\phi} &= d\left(\sum_{i=0}^{d-1}g_i^{-1/2}\bra{i}\right)\gamma\left(\sum_{j=0}^{d-1}g_j^{-1/2}\ket{j}\right)
    = d\left(\sum_{i=0}^{d-1}g_i^{-1/2}\bra{i}\right)\left(\sum_{j=0}^{d-1}g_j^{1/2}\ket{j}\right) = d\sum_{i=0}^{d-1}1 = d^2.\\
    \bra{\phi}\hat{\gamma}\ket{\phi} &= \bra{\phi}\sum_{i,j=0}^{d-1}\sum_{k=0}^{d-1}g_i^{1/2}g_j^{1/2}g_k^{-1/2}\ket{i}\!\bra{j}\ket{k}
    = d\bra{\phi}\sum_{i=0}^{d-1}g_i^{1/2}\ket{i}
    = d\left(\sum_{\ell=0}^{d-1}g_\ell^{-1/2}\bra{\ell}\right)\left(\sum_{i=0}^{d-1}g_i^{1/2}\ket{i}\right)
    = d\sum_{i=0}^{d-1}1 = d^2,
\end{align}
Thus, $\bra{\phi}d\gamma - \hat{\gamma}\ket{\phi} = d^2 - d^2 = 0$. Besides, we note that $d\gamma - \hat{\gamma} = \gamma^{1/2}(d\mathbbm{1} - O_d)\gamma^{1/2}$ where $O_d \equiv \sum_{i,j=0}^{d-1}\ket{i}\!\bra{j}$. Since $d\mathbbm{1}-O_d \ge 0$ and $\gamma$ is full-rank, we have $d\gamma - \hat{\gamma} \ge 0$. Therefore, for Eq.~(\ref{appeq:gamma_>=_f*hat_gamma}) to holds, we have
\begin{align}
    &d\le \frac{d^3(\lambda/2+\alpha - 1)(\lambda/2 - \alpha)}{\alpha(1-\alpha)}, \\
    \Rightarrow \quad & \lambda \ge 1 + \sqrt{(1-2\alpha)^2 + 4d^{-2}\alpha(1-\alpha)}
\end{align}
where in the second line we used $\alpha\in(0,1)$, $\lambda\ge 1$ and $\lambda/2 > \max\{\alpha, 1-\alpha\}$ to neglect other ranges of $\lambda$. We thus have
\begin{align}
    R_\mathrm{T}(\Lambda_\mathdutchcal{C}^{(\Gamma)}) = R_\mathrm{T}[\Lambda_\mathdutchcal{C}^{(\Gamma)}(\gamma_A)] = \lambda_{\min} - 1 =\sqrt{(1-2\alpha)^2 + 4d^{-2}\alpha(1-\alpha)},\quad \forall\, \alpha\in [0,1],\,\varphi\in[0,2\pi],
\end{align}
where we have included the case of $\alpha = 0$ or $1$ because the expression correctly gives $R_\mathrm{T}(\Lambda_\mathdutchcal{C}^{(\Gamma)}) = 1$ as $\alpha = 0$ or $1$. Eq.~(\ref{eq:R_T(Lambda_C)}) is hence derived.

\end{document}